\documentclass[reqno,12pt]{amsart}
\usepackage[a4paper,height=20cm,top=26mm,bottom=26mm,left=26mm,right=26mm]{geometry}

\numberwithin{equation}{section}

\newtheorem{thm}{Theorem}[section]
\newtheorem{cor}[thm]{Corollary}
\newtheorem{lem}[thm]{Lemma}
\newtheorem{prop}[thm]{Proposition}
\theoremstyle{definition}

\newtheorem{example}[thm]{Example}

\newtheorem{conjecture}[thm]{Conjecture}
\theoremstyle{remark}
\newtheorem{remark}[thm]{Remark}

\usepackage{graphicx,color}
\usepackage{here}

\def\C{{\mathbb C}}
\def\R{{\mathbb R}}
\def\Q{{\mathbb Q}}
\def\Z{{\mathbb Z}}

\def\ve{{\varepsilon}}
\def\bw{\mathbf{w}}
\def\bx{\mathbf{x}}
\def\by{\mathbf{y}}
\def\bz{\mathbf{z}}
\def\bX{\mathbf{X}}
\def\bY{\mathbf{Y}}
\def\bZ{\mathbf{Z}}
\def\bW{\mathbf{W}}

\def\bp{\mathbf{p}}
\def\bq{\mathbf{q}}

\def\b1{\bar 1}
\def\bt{\bar 2}

\def\ddo{\text{\scriptsize$\frac{1}{3}$}}
\def\ddt{\text{\scriptsize$\frac{2}{3}$}}
\def\ddf{\text{\scriptsize$\frac{4}{3}$}}
\def\ddfi{\text{\scriptsize$\frac{5}{3}$}}
\def\dds{\text{\scriptsize$\frac{7}{3}$}}

\def\val{\mathrm{val}}
\def\trop{\mathrm{trop}}

\begin{document}
\title[SCA associated with infinite reduced words]{Soliton cellular automata associated with infinite reduced words}

\author[Max Glick]{Max Glick}
\address{Max Glick, Department of Mathematics, The Ohio State University, Columbus, OH 43210, USA.}
\email{glick.107@osu.edu}
\thanks{}

\author[Rei Inoue]{Rei Inoue}
\address{Rei Inoue, Department of Mathematics and Informatics,
   Faculty of Science, Chiba University,
   Chiba 263-8522, Japan.}
\email{reiiy@math.s.chiba-u.ac.jp}
\thanks{R.~I. was partially supported by JSPS KAKENHI Grant Number
26400037 and 16H03927.}

\author[Pavlo Pylyavskyy]{Pavlo Pylyavskyy} 
\address{Pavlo Pylyavskyy, School of Mathematics, University of Minnesota, 
Minneapolis, MN 55414, USA.}
\email{ppylyavs@umn.edu}
\thanks{P.~P. was partially supported by NSF grants DMS-1148634, DMS-1351590, and Sloan Fellowship.}

\date{December 23, 2017, revised on July 24, 2018}

\subjclass[2000]{}


\begin{abstract}
We consider a family of cellular automata $\Phi(n,k)$ 
associated with infinite reduced elements on the affine symmetric group
$\hat S_n$, which is a tropicalization of the rational maps introduced in
\cite{GP16}. 
We study the soliton solutions for $\Phi(n,k)$ and explore 
a `duality' with 
the $\mathfrak{sl}_n$-box-ball system.
\end{abstract}

\maketitle

\section{Introduction}

\subsection{Soliton cellular automata}

A soliton cellular automaton (SCA) is a cellular automaton which has solitonic solutions. The first example, and a beautiful one at that, of SCA is the \emph{box-ball system} (BBS) discovered by Takahashi and Satsuma in 1990 \cite{TS}, which is a dynamical system of finitely many balls in an infinite number of boxes arranged in a line.
We will present its concrete definition in \S \ref{sec:BBS}, and here we simply show a typical time evolution of BBS:\begin{align*}
t=0: & ~\cdots 11222111121111111111111 \cdots
\\
t=1: & ~\cdots 11111222112111111111111 \cdots
\\
t=2: & ~\cdots 11111111221221111111111 \cdots
\\
t=3: & ~\cdots 11111111112112221111111 \cdots
\\
t=4: & ~\cdots 11111111111211112221111 \cdots
\end{align*}
where $1$ and $2$ respectively denote an empty box and a box occupied by 
a ball. One can observe here the notion of soliton that 
(i) a soliton (a sequence of balls) moves to the right with a constant velocity
proportional to the size of the soliton (the length of the sequence), and
(ii) a bigger soliton eventually passes a smaller one, after a scattering
with resulting `shifts' of their locations. 
The shifts caused by scattering provide evidence that the BBS is a nonlinear system.

Though the original definition of BBS seems to be far from known integrable 
systems, the piecewise-linear equation which describes the system turned out to be related to
the piecewise-linear limit (ultradiscretization or tropicalization) of 
the discrete KdV equation \cite{TTMS}.
Another remarkable property of BBS is that its initial value problem is 
independently solved 
by using completely different mathematics, crystal base theory 
\cite{FOY,HHIKTT} and tropical geometry \cite{IT}. 
(Also see \cite{IKT12} for a review and a list of references on these topics.)
In any of these strategies, the {\em tau-function} plays an important role
in describing solutions \cite{Hirota-book}.

We are interested in methods to construct SCA and to study their solutions 
applying combinatorics, representation theory and tropical geometry.
To this end, we start with the discrete soliton equations 
introduced in \cite{GP16}, and study the corresponding SCA.

\subsection{The Coxeter discrete KdV}

In \cite{GP16}, two of the authors introduced a new method to
construct dynamical models of discrete space-time coordinates,
associated to a pair of reduced words in the affine symmetric group 
$\hat{S}_n$.  
Let $s_i~(i=0,\ldots,n-1)$ be the generators of $\hat{S}_n$ with relations:
\begin{align*}
&s_i^2 = 1,
\\
&s_i s_{i+1} s_i = s_{i+1} s_i s_{i+1},
\\
&s_i s_j = s_j s_i \quad |i-j| > 1  \mod n,
\end{align*}  
where we take the indices $i$ of $s_i$ modulo $n$.
Consider a pair $u, v$ of reduced elements in $\hat S_n$ such that $v u$ is also reduced. 
Fix reduced decompositions of $u$ and $v$,
$u = s_{i_1} s_{i_2} \cdots s_{i_l}$ and 
$v = s_{j_1} s_{j_2} \cdots s_{j_m}$,
and assign each $s_i$ with a real variable $a$ as $s_i(a)$.
The dynamics called the `Coxeter discrete KdV' is defined as the rational transformation of the $\ell+m$
variables assigned to $vu$, 
induced by `moving $v$ to the right of $u$' with the use of 
the Lusztig relations:
\begin{align}
\label{eq:L-1}
&s_i(a) s_{i+1}(b) s_i(c) = s_{i+1}(bc/(a+c)) s_i(a+c) s_{i+1}(ab/(a+c)),
\\
\label{eq:L-2}
&s_i(a) s_j(b) = s_j(b) s_i(a) \quad |i-j| > 1  \mod n.
\end{align} 
Originally these relations were introduced by Lusztig 
to study the totally positive parts of algebraic groups and the canonical 
bases of quantum groups \cite{L94,BFZ96}. 
In the network model which offers a strong tool in \cite{GP16},
\eqref{eq:L-1} is depicted as in Figure \ref{fig:YBmove} and called 
the Yang-Baxter move \cite{TP12}.

\begin{figure}[ht]
\unitlength=1mm
\begin{picture}(120,30)(0,5)
\multiput(10,10)(35,0){2}{\line(1,0){5}}
\put(25,10){\line(1,0){10}}
\multiput(10,20)(35,0){2}{\line(1,0){5}}
\multiput(10,30)(25,0){2}{\line(1,0){15}}

\multiput(15,10)(20,0){2}{\line(1,1){10}}
\multiput(15,20)(20,0){2}{\line(1,-1){10}}
\put(25,20){\line(1,1){10}}
\put(25,30){\line(1,-1){10}}

\multiput(45,10)(0,10){3}{\vector(1,0){5}}

\put(7,9){\scriptsize $i$}
\put(2,19){\scriptsize $i+1$}
\put(2,29){\scriptsize $i+2$}

\put(20,15){\circle*{1}} \put(16,14){\scriptsize $a$}
\put(30,25){\circle*{1}} \put(26,24){\scriptsize $b$}
\put(40,15){\circle*{1}} \put(36,14){\scriptsize $c$}

\put(54,19){$\longleftrightarrow$}

\multiput(65,30)(35,0){2}{\line(1,0){5}}
\put(80,30){\line(1,0){10}}
\multiput(65,20)(35,0){2}{\line(1,0){5}}
\multiput(65,10)(25,0){2}{\line(1,0){15}}

\multiput(70,30)(20,0){2}{\line(1,-1){10}}
\multiput(70,20)(20,0){2}{\line(1,1){10}}
\put(80,20){\line(1,-1){10}}
\put(80,10){\line(1,1){10}}

\multiput(100,10)(0,10){3}{\vector(1,0){5}}

\put(75,25){\circle*{1}} \put(66,24){\small $\frac{bc}{a+c}$}
\put(85,15){\circle*{1}} \put(76,14){\scriptsize $a+c$}
\put(95,25){\circle*{1}} \put(86,24){\small $\frac{ab}{a+c}$}

\end{picture}
\caption{The Yang-Baxter move}
\label{fig:YBmove}
\end{figure}
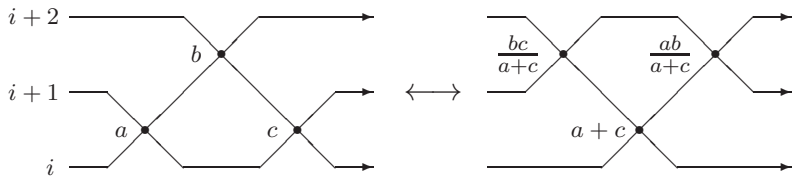

The important property of the rational transformation studied in \cite{GP16} 
is that 
they have soliton solutions.
By applying the network model 
to a pictorial representation of reduced elements in $\hat S_n$,
we identify
{\it vertex variables} and {\it chamber variables} of the network
respectively with the dynamical variables and the tau-functions
of the model.
This enables us to reduce the rational transformation to a bilinear equation,
and the multi-soliton solutions are obtained.

\subsection{Main results}

In this paper, we focus on a family of dynamical system $\phi(n,k) ~(
k=0,1,\ldots,n-2)$ for $n > 2$
given by $u=s_1 s_2 \ldots \cdots s_{n-1}$ and 
$v = s_k s_{k-1} \cdots s_{0} s_{n-1}\cdots s_{k+2}$.
As in the general case, the rational transformations are subtraction free because they are built out of Lusztig moves \eqref{eq:L-1} and hence can be written using only the other three arithmetic operations.  Hence $\phi(n,k)$ has a {\it tropicalization} (or, piecewise-linear version) obtained by replacing $(+, \times, \div)$ with $(\min, + ,-)$.  We study the tropicalization $\Phi(n,k)$ of $\phi(n,k)$, 
and mainly consider the soliton solutions for $\Phi(n,k)$ on $\Z$.

Our main results are that a form of each soliton for $\Phi(n,k)$ on $\Z$
is parametrized by $(n-2)$ positive integers
(Theorem \ref{thm:solitons}),
and that the evolution rule of the solitons has a `duality' with  
that of a well-known soliton cellular 
automaton called the $\mathfrak{sl}_n$-box-ball system \cite{FOY,HHIKTT}
(Conjecture \ref{thm:BBS}, which is a theorem for $n=3,4$).
The second result is clarified via 
the combinatorial $R$-matrix acting on the product of crystals corresponding 
to the symmetric tensor representation of $U_q'(\hat{\mathfrak{sl}}_n)$.
%
%
Additionally, we find that the rational map $\phi(n,k)$  
is a limit of another $R$-matrix action (Proposition \ref{prop:phi-formula}), namely
the geometric $R$-matrix acting on the product of 
geometric versions of the above (symmetric) crystal and its dual.

We remark that the network model \cite{TP12}
has a potential to be a useful tool to
study integrable rational maps. Once we can formulate the map using 
a network, we may get not only the information of bilinear equations
as in \cite{GP16}, 
but also the combinatorial information of the Lax form.   
For an example of the latter application, see \cite{ILP}.

\bigskip 
This paper is organized as follows:
in Section 2, following \cite{GP16} we define the dynamical system
and introduce $\phi(n,k)$ and $\Phi(n,k)$. We give the notion of soliton
for $\Phi(n,k)$ in \S \ref{subsec:soliton},
and state the first main result at Theorem \ref{thm:solitons}. 
In Section 3, the explicit formula for $\phi(n,k)$ and the relation to
the geometric $R$-matrix is shown (Proposition \ref{prop:phi-formula}).
In Section 4, by making use of the pictorial representation of reduced 
words in $\hat S_n$,
the tau-function and the bilinear equation for the model are obtained.
Sections 5 and 6 are devoted to computing soliton solutions
for $\Phi(n,k)$ in different two ways. 
We compute the tropicalization of the geometric solutions for $\phi(n,k)$ obtained in \cite{GP16} in Section 5, and see that almost all geometric solutions vanish 
in tropicalization except for the simplest ones.  
In Section 6, we naively solve the tropical bilinear equations
and prove Theorem \ref{thm:solitons}.
In Section 7, we study the duality between the soliton solutions for $\Phi(n,k)$ and the $\mathfrak{sl}_n$-box-ball system. After a brief introduction of 
the box-ball system,
we present Conjecture \ref{thm:BBS}.
We explain a strategy to prove it in \S \ref{subsec:strategy}, and 
give the proof in the cases of $n=3$ and $4$ in \S \ref{subsec:n=3} and 
\S \ref{subsec:n=4} respectively. 
In the last section, we present other interesting numerical phenomena for $\Phi(n,k)$,
including negative solitons, pulsars, and relaxations of solitons and pulsars.
We add an Appendix to explain the basics of the tropical semifield used in this paper.

\section{Description of the model}

\subsection{Discrete dynamical system in infinite reduced words}

Following \cite{GP16}, we introduce a discrete dynamical system 
associated to a pair of reduced words.

Define an automorphism $\rho$ on the affine symmetric group $\hat{S}_n$ 
by $\rho(s_i) = s_{i+1}$.
Let $\eta$ be a homomorphism from $\hat{S}_n$ to the symmetric group $S_n$ 
given by $\eta(s_i) = (i,i+1)$ (as usual, $i$ is considered modulo $n$).
We call a reduced element $g \in \hat{S}_n$ a {\it glide}, when  
$$
\eta(g) = 
\begin{pmatrix}
1 & 2 & \cdots & n-k & n-k+1 & \cdots & n \\
1+k & 2+k & \cdots & n & 1 & \cdots & k
\end{pmatrix}  
$$ 
holds for some $k=0,1,\ldots,n-1$. We call this $k$ the {\it offset} of 
the glides.

Fix the expressions of two glides
$u = s_{i_1} s_{i_2} \cdots s_{i_l}$ and 
$v = s_{j_1} s_{j_2} \cdots s_{j_m}$, and assume that $v u$ is  reduced.
Let $k_1$ and $k_2$ be the {\it offsets} of $u$ and $v$ respectively.
By \cite[Lemma 2.1]{GP16}, we have
\begin{align}\label{eq:uv}
v u = \rho^{k_2}(u) \rho^{-k_1}(v).
\end{align}
Consider an (semi-)infinite word 
$$
u \cdot \rho^{-k_1}(u) \cdot \rho^{-2 k_1}(u)
\cdot \rho^{-3 k_1}(u) \cdots.
$$
Put $v$ on the left side of the infinite word,
and move it to the right using \eqref{eq:uv}:  
\begin{align}\label{eq:v-evolution}
\begin{split}
&v \cdot u \cdot \rho^{-k_1}(u) \cdot \rho^{-2 k_1}(u) \cdot \rho^{-3 k_1}(u)
\cdots
\\
& \qquad = \rho^{k_2}(u) \cdot \rho^{-k_1}(v) \cdot \rho^{-k_1}(u) \cdot \rho^{-2 k_1}(u) \cdot \rho^{-3 k_1}(u) \cdots
\\
& \qquad = \cdots
\\
& \qquad
= \rho^{k_2}(u) \cdot \rho^{-k_1+k_2}(u) \cdot \rho^{-2 k_1 +k_2}(u) \cdot \rho^{-3 k_1+k_2}(u)\cdots.
\end{split}
\end{align}
At the first line of \eqref{eq:v-evolution},
we assign positive variables $\bz_0 = (z_{0,1},\ldots,z_{0,m})$ to $v$ as 
$s_{j_1}(z_{0,1}) \cdots s_{j_m}(z_{0,m})$, and 
$\by_i = (y_{i,1},y_{i,2},\ldots,y_{i,l})$ to 
$\rho^{- i k_1}(u)$ as $s_{i_1 - i k_1}(y_{i,i_1}) \cdots s_{i_l - i k_1}(y_{i,i_l})$ for $i \in \Z_{\geq 0}$.
By using the Lusztig relations,
we define a rational transformation of the parameters 
corresponding to \eqref{eq:v-evolution};  
\begin{align}\label{eq:zy-evol}
\begin{split}
&(\bz_0, \by_0, \by_1, \by_2,\by_3,\ldots)
\\ 
&\quad \mapsto (\by'_0, \bz_1, \by_1, \by_2,\by_3, \ldots)
\\ 
&\quad \mapsto \cdots 
\\ 
&\quad \mapsto (\by'_0, \by'_1, \by'_2 ,\by'_3,\ldots, \bz_\infty),
\end{split}
\end{align}
where $(\bz_i,\by_i)$ is transformed into $(\by_i',\bz_{i+1})$ corresponding to 
$\rho^{-i k_1}(v) \cdot \rho^{-i k_1}(u) =
\rho^{-i k_1+k_2}(u) \cdot \rho^{-(i+1) k_1}(v)$.
We call each $\by_i$ a {\it state}, and each $\bz_i$ a {\it carrier}.

A {\it commuting pair} for $u$ and $v$ is a choice of
$\bw = (w_1,\ldots,w_l)$ and $\bz = (z_1,\ldots,z_m)$ satisfying
\begin{align*}
&s_{j_1}(z_1) s_{j_2}(z_2) \cdots s_{j_m}(z_m) \cdot
s_{i_1}(w_1) s_{i_2}(w_2) \cdots s_{i_l}(w_l)
\\
& \quad = 
s_{i_1+k_2}(w_1) s_{i_2+k_2}(w_2) \cdots s_{i_l+k_2}(w_l) \cdot 
s_{j_1-k_1}(z_1) s_{j_2-k_1}(z_2) \cdots s_{j_m-k_1}(z_m).
\end{align*}
We call such $\bw$ and $\bz$ a {\it vacuum state}  and 
an {\it initial carrier} respectively.

We fix a commuting pair $(\bw,\bz)$, we  
set $\bz_0 = \bz$, and we assume $\lim_{i \to \infty} \by_i = \bw$. 
For $t > 0$, define $(\by_i^t)_{i \in \Z_{\geq 0}}$ inductively by 
\begin{align}\label{eq:def-system}
(\bz, \by_0^t, \by_1^t,\by_2^t,\by_3^t, \ldots) \mapsto 
(\by_0^{t+1}, \by_1^{t+1},\by_2^{t+1},\by_3^{t+1}, \ldots).
\end{align}
Empirically, at the end of each step the final carrier $\bz_{\infty}$ equals the initial one $\bz$ and the new states again satisfy $\lim_{i \to \infty} \by_i^{t+1} = \bw$.

\begin{remark}
For the sake of exposition, we have simplified the system from \cite{GP16} wherein the state sequence $\ldots, \by_{-1}, \by_0, \by_1, \by_2, \ldots$ is bi-infinite and assumed to approach $\bw$ in both directions.  The description above corresponds to the special case in which $\by_i = \bw$ for all $i < 0$.  To properly define the general system one must insert the initial carrier farther and farther left and take a limit.  We ignore this issue, because our focus will be on the tropicalization of the system for which it is consistent to assume for each $t$ that only finitely many $\by_i^t$ differ from $\bw$. 
\end{remark}

\subsection{Dynamical system $\phi(n,k)$}

Let us focus on the case that both $u$ and $v$ have length $n-1$ as
\begin{align}\label{eq:Phi-uv}
\begin{split}
&u = s_1 s_2 s_3 \cdots s_{n-1},
\\
&v = v(k) = s_k s_{k-1} \cdots s_{0} s_{n-1}\cdots s_{k+2};
~k=0,1,\ldots,n-2,
\end{split}
\end{align}
whose offsets are $k_1 = 1$ and $k_2 = n-1$ respectively.

\begin{lem}\label{lem:comm-pair}
The following pair $(\bw,\bz)$ is a commuting pair for $(u,v(k))$:
\begin{align*}
&\bz = (\alpha_k-\alpha_{k+1}, \alpha_{k-1}-\alpha_{k+1},\ldots,
\alpha_{1}-\alpha_{k+1}, \alpha_{n}-\alpha_{k+1},\alpha_{n-1}-\alpha_{k+1},\ldots, \alpha_{k+2}-\alpha_{k+1}),
\\
&\bw = (\alpha_n-\alpha_1,\alpha_n-\alpha_2,\ldots,\alpha_n-\alpha_{n-1}),
\end{align*}
where we assume $\alpha_{k+1} < \alpha_i < \alpha_n$ for $i=\{1,\ldots,n-1\} \setminus \{k+1\}$, and the indices $i$ of $\alpha_i$ are taken modulo $n$. 
\end{lem}

\begin{proof}[Proof of Lemma \ref{lem:comm-pair}]
To the wiring diagram of $v(k) u$ introduced at Figure \ref{fig:uv} 
in \S \ref{subsec:tau-phi}, 
we apply the {\it wire ansatz} in \cite[Section 4]{GP16} by 
replacing $\alpha_i$ with $-\alpha_i$ on the $i$-th wire.
\end{proof}

We write $\phi(n,k)$ for the dynamical system \eqref{eq:def-system}
given by the two glides 
$u$ and $v(k)$ of \eqref{eq:Phi-uv} together with the commuting pair
of Lemma \ref{lem:comm-pair}.

\begin{example}\label{ex:n=3}
The system $\phi(3,1)$ corresponds to the words $u=s_1s_2$, $v = s_1s_0$.  The equation \eqref{eq:uv} reads
\begin{displaymath}
vu = s_1s_0s_1s_2 = s_0s_1s_0s_2 = \rho^2(u)\rho^{-1}(v)
\end{displaymath}
which follows from a single braid move.  By \eqref{eq:L-1} the weights evolve according to
\begin{align}\label{eq:phi(31)}
\begin{split}
((z_{i,1}^t,z_{i,2}^t),(y_{i,1}^t,y_{i,2}^t))
~\mapsto~ &((y_{i,1}^{t+1},y_{i,2}^{t+1}),(z_{i+1,1}^t,z_{i+1,2}^t))
\\
&=  
\left(\left(\frac{z_{i,2}^t y_{i,1}^t}{z_{i,1}^t+y_{i,1}^t},z_{i,1}^t+y_{i,1}^t
\right),
\left(\frac{z_{i,1}^t z_{i,2}^t}{z_{i,1}^t+y_{i,1}^t},y_{i,2}^t \right)
 \right).  
\end{split}
\end{align}
The commuting pair is $\bz = (\alpha_1-\alpha_2,\alpha_3-\alpha_2)$,
$\bw =(\alpha_3-\alpha_1,\alpha_3-\alpha_2)$ and one can check easily from the formula that $(\bz,\bw) \mapsto (\bw,\bz)$ does hold.
The full system inputs $\by_0^t, \by_1^t, \ldots$ with $\lim_{i\to\infty} \by_i^t = \bw$ and uses \eqref{eq:phi(31)} (with $\bz_0^t = \bz$) to calculate the $\by_i^{t+1}$.
\end{example}

\subsection{Tropical dynamical system $\Phi(n,k)$}


For $A,B,C \in \R$, we define the {\it tropicalization} of 
the Lusztig relations \eqref{eq:L-1} and \eqref{eq:L-2} as
\begin{align}\label{eq:tropL-1}
&s_i(A) s_{i+1}(B) s_i(C) 
= s_{i+1}(A') s_i(B') s_{i+1}(C'),
\\
&s_i(A) s_j(B) = s_j(B) s_i(A) \quad |i-j| > 1  \mod n,
\end{align} 
where 
$$
(A',B',C') := (B+C-\min(A,C), \min(A,C), A+B+-\min(A,C)).
$$
See Appendix A for preliminaries on tropicalization.
It is straightforward to tropicalize Lemma \ref{lem:comm-pair}.

\begin{lem}\label{lem:trop-pair}
The following pair $(\bW, \bZ)$ is a tropical commuting pair 
for $(u,v(k))$:
\begin{align}\label{trop-commpair}
\bZ = (A_k,A_{k-1},\ldots,A_1,A_{n},A_{n-1},\ldots,A_{k+2}),
\quad
\bW = (A_n,A_n,\ldots,A_n),
\end{align}
where we assume 
\begin{align}\label{eq:A-condition}
A_{k+1} > A_i > A_n; ~i \in \{1,\ldots,n-1\} \setminus \{k+1\}.
\end{align}
\end{lem}

Using the tropical Lusztig relations for the glides $u$ and $v(k)$ of
\eqref{eq:Phi-uv}
we define the piecewise-linear transformation of the real variables
$\bZ_i = (Z_{i,1},\ldots,Z_{i,n-1})$
and $\bY_i = (Y_{i,1},Y_{i,2},\ldots,Y_{i,n-1})$ 
for $i \in \Z_{\geq 0}$ in the same way as \eqref{eq:zy-evol}.  

We write $\Phi(n,k)$ for the tropical dynamical system 
given by the two glides \eqref{eq:Phi-uv} and 
the commuting pair \eqref{trop-commpair}.
As in the rational case, we call each $\bZ_i$ and each $\bY_i$
a {\it carrier} and a {\it state} respectively.
Also, we call $\bW$ the {\it vacuum state}, and $\bZ$ the {\it initial carrier}.

\begin{example}
Corresponding to Example \ref{ex:n=3}, we have $\Phi(3,1)$ given by
\begin{align}\label{eq:tropevol-3}
\begin{split}
&((Z_{i,1}^t,Z_{i,2}^t),(Y_{i,1}^t,Y_{i,2}^t))
\\
& \quad \mapsto~ ((Y_{i,1}^{t+1},Y_{i,2}^{t+1}),(Z_{i+1,1}^t,Z_{i+1,2}^t))
\\
& \qquad \quad =  
\left((Z_{i,2}^t + Y_{i,1}^t - \min[Z_{i,1}^t, Y_{i,1}^t],
\min[Z_{i,1}^t, Y_{i,1}^t]),
(Z_{i,1}^t + Z_{i,2}^t - \min[Z_{i,1}^t, Y_{i,1}^t], Y_{i,2}^t)\right)
\end{split}
\end{align}
with the commuting pair $\bZ = (A_1,A_3)$ and $\bW = (A_3,A_3)$.
\end{example}

In this paper, we mainly study the tropical dynamics on $\Z \subset \R$. 
In particular we consider the case with $A_n=0$, $A_i = 1$ 
for $i=\{1,\ldots,n-1\} \setminus \{k+1\}$, and $A_{k+1} > 1$,
so that the commuting pair is 
\begin{align}\label{eq:Phi-comm}
\bZ = (\underbrace{1,\ldots,1}_k,0,\underbrace{1,\ldots,1}_{n-2-k}),
\quad 
\bW = (0,\ldots,0).
\end{align}

\subsection{Solitons}\label{subsec:soliton}
We define a {\it one-soliton} for $\Phi(n,k)$ 
to be a finite sequence $\bX_1,\ldots, \bX_m$ of non-vacuum states satisfying the following conditions:
\begin{itemize}
\item[(i)] 
The sequence moves to the right with a constant velocity, i.e. for some $a$ and $b$ the input
\begin{displaymath}
\bX_1,\ldots, \bX_M, \bW, \bW, \ldots
\end{displaymath} 
is carried under $a$ steps of $\Phi(n,k)$ to
\begin{displaymath}
\underbrace{\bW,\ldots, \bW}_b, \bX_1,\ldots, \bX_M, \bW, \bW, \ldots.
\end{displaymath}

\item[(ii)]
For each $t$, the final carrier equals the initial one, 
i.e. $\bZ_{i}^t = \bZ$ for $i \gg 0$ (unlike the rational case, we know of inputs for which this condition fails, see Section \ref{subsec:relax}).
\end{itemize}
An amazing feature of soliton systems is the existence of \emph{multi-soliton} solutions, which we define in our setting to be an input consisting of several one-solitons separated by vacuums such that 
\begin{itemize}
\item[(iii)]
for $t \gg 0$ the outcome is a collection of one-solitons, arranged in increasing order of velocity from left to right, with the same set of veloicities as the initial solitons. 
\end{itemize}
We will see that different one-solitons can have the same speed and in particular that the components of a multi-soliton for $t \gg 0$ may differ from the initial ones.

For two states $\bY_i^t$ and $\bY_{i'}^{t'}$, we say  
$\bY_i^t$ is bigger than $\bY_{i'}^{t'}$, if $Y_{i,j}^t \geq Y_{i',j}^{t'}$ for all $j=1,\ldots,n-1$ and there is at least one $j$ such that $Y_{i,j}^t > Y_{i',j}^{t'}$.
In the same manner, we say $\bY_i^t$ is smaller than $\bY_{i'}^{t'}$, if $Y_{i,j}^t \leq Y_{i',j}^{t'}$ for all $j=1,\ldots,n-1$ and there is at least one $j$ such that $Y_{i,j}^t < Y_{i',j}^{t'}$.
We say a soliton is {\it positive} (resp. {\it negative}) 
when all states of the soliton 
are bigger (resp. smaller) than the vacuum state.

The following are several examples of positive solitons,
where we show $(\bY_i^t)_i$ for each $t$.  

\begin{example}\label{ex:one-solitonsY}
One-solitons.
\\
(i) $\Phi(3,1)$:
\begin{align*}
t = 0: & (00)(31)(00)(00)(00)(00)(00)(00)(00)
\\
t = 1: & (00)(21)(10)(00)(00)(00)(00)(00)(00)
\\
t = 2: & (00)(11)(20)(00)(00)(00)(00)(00)(00)
\\
t = 3: & (00)(01)(30)(00)(00)(00)(00)(00)(00)
\\
t = 4: & (00)(00)(31)(00)(00)(00)(00)(00)(00)
\\
t = 5: & (00)(00)(21)(10)(00)(00)(00)(00)(00)
\end{align*}
(ii) $\Phi(4,1)$:
\begin{align*}
t = 0: & (000)(312)(000)(000)(000) (000) (000) 
\\
t = 1: & (000)(212)(100)(000) (000) (000) (000) 
\\
t = 2: & (000)(112)(200)(000) (000) (000) (000) 
\\
t = 3: & (000)(012)(300)(000) (000) (000) (000) 
\\
t = 4: & (000) (002)(310) (000) (000) (000) (000) 
\\ 
t = 5: & (000) (001) (311)(000)(000) (000) (000) 
\\
t = 6: & (000)(000)(312)(000)(000) (000) (000) 
\\
t = 7: & (000)(000)(212)(100)(000) (000) (000) 
\end{align*}
\end{example}

For a one-soliton  we define its {\it minimal length} to be the minimal lattice
length the soliton occupies in propagation.
We also define the {\it velocity} of a soliton,
which is the ratio of `the minimal number of time steps 
it takes to recover the initial sequence'
and `the lattice length it propagates during the time steps' (in the notation of (i), velocity $= b/a$) . 
In the first case of Example \ref{ex:one-solitonsY}, 
the soliton occupies one lattice at $t=0,4$ and two lattices at $t=1,2,3,5$.
The sequence at $t=0$ is recovered at $t=4$, 
moving one lattice to the right.
Hence it has minimal length one, and velocity $1/4$. 
Similarly, the soliton in the second case has
minimal length one, and velocity $1/6$.

\begin{thm}\label{thm:solitons}
Consider the system $\Phi(n,k)$ on $\Z$ with commuting pair \eqref{eq:Phi-comm}.  Then 
$$
X = (b_1,b_2,\ldots,b_{n-1})
$$
is a one-soliton with minimal length one for any $b_1,\ldots, b_{n-1} \in \Z_{\geq 1}$ with $b_{k+1} = 1$.
Its velocity is $(\sum_{k=1}^{n-1} b_k)^{-1}$.
\end{thm}


In principle one can verify that $X$ above is a soliton directly using formulas for the system we develop in \S \ref{sec:phi-formula}.  Instead we give the proof in \S \ref{sec:1soliton} at which point we are able to give explicit descriptions of one-solitons in terms of tau-functions (Proposition \ref{prop:n-onesoliton}).  One upshot of this approach is that it mimics what is done for similar systems and suggests that Theorem \ref{thm:solitons} gives all positive solitons (and in particular, all positive solitons have minimal length one).  Moreover, the tau-functions could be useful in constructing multi-solitons, which we demonstrate in the case of $n=3$.

To denote a positive soliton we use the form at the minimal length 
of Theorem \ref{thm:solitons},
and call it the {\it minimal form} of the soliton.
For example, in Example \ref{ex:one-solitonsY}
the one-solitons have the minimal forms $(3,1)$ for (i) and 
$(3,1,2)$ for (ii).  
In the rest we often call a positive soliton just `a soliton'.

Here are examples of soliton scatterings whose combinatorial property
will be studied in \S \ref{sec:BBS}.

\begin{example}\label{ex:n=3Y}
Two-solitons.
\\ 
(i) The case of $\Phi(3,1)$; $(1,1) \times (3,1) \mapsto (3,1) \times (1,1)$: 
\begin{align*}
t = 0: & (00)(11)(00)(31)(00)(00)(00)(00)(00)(00)(00)(00)
\\
t = 1: & (00)(01)(10)(21)(10)(00)(00)(00)(00)(00)(00)(00)
\\
t = 2: & (00)(00)(11)(11)(20)(00)(00)(00)(00)(00)(00)(00)
\\
t = 3: & (00)(00)(01)(20)(21)(00)(00)(00)(00)(00)(00)(00)
\\
t = 4: & (00)(00)(00)(21)(11)(00)(00)(00)(00)(00)(00)(00)
\\
t = 5: & (00)(00)(00)(11)(20)(11)(00)(00)(00)(00)(00)(00)
\\
t = 6: & (00)(00)(00)(01)(30)(01)(10)(00)(00)(00)(00)(00)
\\
t = 7: & (00)(00)(00)(00)(31)(00)(11)(00)(00)(00)(00)(00)
\end{align*}
(ii) The case of $\Phi(4,1)$;
$(2,1,1) \times (3,1,2) \mapsto (4,1,1) \times (1,1,2)$:
\begin{align*}
t = 0: & (000)(211)(001)(311)(000) (000) (000) (000) (000)(000)
\\
t = 1: & (000)(111)(100)(312) (000) (000) (000) (000) (000) (000) 
\\
t = 2: & (000)(011)(200)(212) (100) (000) (000) (000) (000) (000)
\\
t = 3: & (000)(001)(210)(112) (200) (000) (000) (000) (000) (000) 
\\
t = 4: & (000)(000)(211)(012) (300) (000) (000) (000) (000) (000)
\\ 
t = 5: & (000) (000)(111)(102)(310) (000) (000) (000) (000) (000)
\\
t = 6: & (000)(000)(011)(201)(311) (000) (000) (000) (000) (000) 
\\
t = 7: & (000)(000)(001)(210)(312) (000) (000) (000) (000) (000)  
\\
t = 8: & (000)(000)(000)(211)(212) (100) (000) (000) (000) (000)  
\\
t = 9: & (000)(000)(001)(111)(302) (110) (000) (000) (000) (000)  
\\
t = 10: & (000)(000)(000)(011)(401) (111) (000) (000) (000) (000)  
\\
t = 11: & (000)(000)(000)(001)(410) (112) (000) (000) (000) (000)  
\\
t = 12: & (000)(000)(000)(000)(411) (012) (100) (000) (000) (000)  
\\
t = 13: & (000)(000)(000)(000)(311) (102) (110) (000) (000) (000)  
\\
t = 14: & (000)(000)(000)(000)(211) (201) (111) (000) (000) (000)  
\end{align*}
where the change of internal structure of solitons is observed. 
\end{example}

The behavior witnessed in these and other examples suggest the following.

\begin{conjecture}
Combining positive solitons gives rise to multi-soliton solutions as in condition (iii) at the beginning of this subsection.
\end{conjecture}

In the cases of $n=3,4$, this conjecture is a theorem
which follows from the duality with the $\mathfrak{sl}_n$-box-ball system
proved in \S \ref{subsec:n=3} and \S \ref{subsec:n=4}.

\section{The formula for $\phi(n,k)$}\label{sec:phi-formula}

Let $\bar{\mathcal{R}}$ be a rational map on $\Q(\bp,\bq)$
with non-negative variables 
$\bp = (p_i)_{i=1,\ldots,n}$, $\bq = (q_i)_{i=1,\ldots,n}$,
given by $\bar{\mathcal{R}}: (\bp, \bq) \mapsto (\bq',\bp')$;
\begin{align}\label{geomR-wc}
  p_i' = p_i \frac{p_{i+1}+q_{i+1}}{p_{i}+q_{i}},
  \quad 
  q_i' = q_i \frac{p_{i+1}+q_{i+1}}{p_{i}+q_{i}}.
\end{align} 
This map originates from the geometric version of the combinatorial $R$-matrix,
the isomorphism between the tensor products of crystals 
$B_{\bar m} \otimes B_\ell \stackrel{\sim}{\to} B_\ell \otimes B_{\bar m}$.
Here $B_\ell$ is the crystal corresponding to 
the $\ell$-fold symmetric tensor representation of 
$U_q'(\hat{\mathfrak{sl}}_n)$, and $B_{\bar{\ell}}$ is the dual of $B_\ell$. 
As a set, $B_\ell$ and $B_{\bar \ell}$ are the same:
\begin{align}\label{eq:cryetalB}
  B_\ell = B_{\bar \ell} = \{\bx = (x_1,x_2,\ldots,x_n) \in (\Z_{\geq 0})^n;
  ~ \sum_{i=1}^n x_i = \ell \}.
\end{align}
See \cite[\S 3.1 and \S11.9]{TP12} for details of the map $\bar{\mathcal{R}}$.

\begin{prop}
\label{prop:phi-formula}
The transformation
$(\bz_i,\by_i) \mapsto (\by_i',\bz_{i+1})$
of the dynamical system $\phi(n,k)$ is described by $\bar{\mathcal{R}}$, 
by setting 
\begin{align*}
  &\bp = 
 (z_{i,k+1},z_{i,k+2},\ldots,z_{i,n-1},0,z_{i,1},z_{i,2},\ldots,z_{i,k}), 
 \\
  &\bq = 
 (0, y_{i,n-1},y_{i,n-2},\ldots,y_{i,k+2},y_{i,k+1},y_{i,k},\ldots,y_{i,1}).
\end{align*}
More explicitly, the transformation 
$(\bz_i,\by_i) \mapsto (\by_i',\bz_{i+1})$ is given by
\begin{align}\label{eq:phi-formula}
  y_{i,j}' = y_{i,j} \frac{z_{i,k+2-j} + y_{i,j-1}}{z_{i,k+1-j} + y_{i,j}},
  \quad 
  z_{i+1,j} = z_{i,j} \frac{z_{i,j+1} + y_{i,k-j}}{z_{i,j} + y_{i,k+1-j}}.
\end{align}
Here we assume that the second subscript $j$ of $z_{i,j}$ and $y_{i,j}$
is taken modulo $n$, and set $y_{i,n}=z_{i,n}=0$. 
\end{prop}

First let us prove the following lemma.  Assuming $y_{i,0}=z_{i,0}=0$, denote 
\begin{align}
  y_{i,j}'' = y_{i,j} \frac{z_{i,k+2-j} + y_{i,j-1}}{z_{i,k+1-j} + y_{i,j}},
  \quad 
  z_{i+1,j}'' = z_{i,j} \frac{z_{i,j+1} + y_{i,k-j}}{z_{i,j} + y_{i,k+1-j}}.
\end{align}

\begin{lem}
For any $j$ we have 
$$s_j(z_{i,k+1-j}) s_{j-1}(z_{i,k+2-j} + y_{i,j-1}) s_j(y_{i,j}) = s_{j-1}(y''_{i,j}) s_j(z_{i, k+1-j} + y_{i,j}) s_{j-1}(z''_{i+1, k+1-j}).$$
\end{lem}

\begin{proof}
Direct substitution: 
$$\frac{(z_{i, k+2-j} + y_{i,j-1})y_{i,j}}{z_{i,k+1-j} + y_{i,j}} = y_{i,j}'', \;\; \frac{z_{i,k+1-j} (z_{i,k+2-j} + y_{i,j-1})}{z_{i,k+1-j} + y_{i,j}} = z''_{i+1, k+1-j}.$$
\end{proof}

Now we are ready to prove the proposition.

\begin{proof}[Proof of Proposition \ref{prop:phi-formula}]
By the commutativity relations
\begin{align*}
&s_k(z_{i,1}) \dotsc s_j(z_{i,k+1-j}) \dotsc s_{k+2}(z_{i,n-1}) s_1(y_{i,1}) \dotsc s_j(y_{i,j}) \dotsc s_{n-1}(y_{i,n-1})
\\
&= s_k(z_{i,1}) \dotsc s_1(z_{i,k}) s_0(z_{i,k+1}) s_1(y_{i,1}) \dotsc s_k(y_{i,k}) 
\\
&\qquad \cdot s_{n-1}(z_{i,k+2}) \dotsc s_{k+2}(z_{i,n-1}) s_{k+1}(y_{i,k+1}) s_{k+2}(y_{i,k+2}) \dotsc s_{n-1}(y_{i,n-1}).
\end{align*}
Applying the lemma several times we see that 
\begin{align*}
&s_k(z_{i,1}) \dotsc s_1(z_{i,k}) s_0(z_{i,k+1}) s_1(y_{i,1}) \dotsc s_k(y_{i,k})
\\
&=s_k(z_{i,1}) \dotsc s_1(z_{i,k}) s_0(z_{i,k+1} + y_{i,0}) s_1(y_{i,1}) \dotsc s_k(y_{i,k})
\\
&=s_0(y''_{i,1}) s_k(z_{i,1}) \dotsc s_2(z_{i,k-1}) s_1(z_{i,k}+y_{i,1}) s_2(y_{i,2}) \dotsc s_k(y_{i,k}) s_0(z''_{i+1,k})
\\
&= \ldots =
s_0(y''_{i,1}) \dotsc s_{k-1}(y''_{i,k}) s_k(z_{i,1}+y_{i,k}) s_{k-1}(z''_{i+1,1}) \dotsc s_{0}(z''_{i+1,k})
\\
&= 
s_0(y''_{i,1}) \dotsc s_{k-1}(y''_{i,k}) s_k(y''_{i,k+1}) s_{k-1}(z''_{i+1,1}) \dotsc s_{0}(z''_{i+1,k}).
\end{align*}
Similarly, by application of the lemma, we get 
\begin{align*}
&s_{n-1}(z_{i,k+2}) \dotsc s_{k+2}(z_{i,n-1}) s_{k+1}(y_{i,k+1}) s_{k+2}(y_{i,k+2}) \dotsc s_{n-1}(y_{i,n-1}) 
\\
&=s_{n-1}(z_{i,k+2}) \dotsc s_{k+2}(z_{i,n-1}) s_{k+1}(z_{i,0} + y_{i,k+1}) s_{k+2}(y_{i,k+2}) \dotsc s_{n-1}(y_{i,n-1})
\\
&= \ldots 
= s_{k+1}(y''_{i,k+2}) \dotsc s_{n-2}(y''_{i,n-1}) s_{n-1}(z_{i,k+2} + y_{i,n-1}) s_{n-2}(z''_{i+1,k+2}) \dotsc s_{k+1}(z''_{i+1,n-1})
\\
&=s_{k+1}(y''_{i,k+2}) \dotsc s_{n-2}(y''_{i,n-1}) s_{n-1}(z''_{i+1,k+1}) \dotsc s_{k+1}(z''_{i+1,n-1}).
\end{align*}
Putting the two parts together we get
\begin{align*}
&s_0(y''_{i,1}) \dotsc s_{k-1}(y''_{i,k}) s_k(y''_{i,k+1}) s_{k-1}(z''_{i+1,1}) \dotsc s_{0}(z''_{i+1,k}) 
\\
&\qquad \cdot s_{k+1}(y''_{i,k+2}) \dotsc s_{n-2}(y''_{i,n-1}) s_{n-1}(z''_{i+1,k+1}) \dotsc s_{k+1}(z''_{i+1,n-1})
\\
&=s_0(y''_{i,1}) \dotsc s_{k-1}(y''_{i,k}) s_k(y''_{i,k+1}) s_{k+1}(y''_{i,k+2}) \dotsc s_{n-2}(y''_{i,n-1}) 
\\ 
& \qquad \cdot s_{k-1}(z''_{i+1,1}) \dotsc s_{0}(z''_{i+1,k}) s_{n-1}(z''_{i+1,k+1}) \dotsc s_{k+1}(z''_{i+1,n-1}) 
\\
&=s_0(y''_{i,1}) \dotsc s_{n-2}(y''_{i,n-1}) s_{k-1}(z''_{i+1,1}) \dotsc  s_{k+1}(z''_{i+1,n-1}).
\end{align*}

This separates into a carrier with parameters $z''_{i+1,j}$ and a state with parameters $y''_{i,j}$. We conclude that $y'_{i,j}=y''_{i,j}$ and $z_{i+1,j} = z''_{i+1,j}$, as desired.

\end{proof}

In the following, we use the same notation $\phi(n,k)$ to denote 
the map on $(\R^{n-1}_{>0})^2$
which comprises the dynamics $\phi(n,k)$, so we write
$\phi(n,k): (\bz_i,\by_i) \mapsto (\by_i',\bz_{i+1})$.
Let $\iota$ be a map on $\R^{n-1}$ given by
$\mathbf{a} = (a_1,a_2,\ldots,a_{n-1}) \mapsto (a_{n-1},a_{n-2},\ldots,a_{1})$,
and let $\tilde \rho$ be the map on $(\R^{n-1})^2$ given by
$(\mathbf{a},\mathbf{b}) \mapsto 
(\iota(\mathbf{b}),\iota(\mathbf{a}))$.

\begin{prop}\label{prop:phi-symmetry}
As maps on $(\R^{n-1}_{>0})^2$, it holds that
\begin{align}\label{eq:phi-phi}
  \phi(n,n-2-k) = \tilde \rho \circ \phi^{-1}(n,k) \circ \tilde \rho
\end{align}
for $k=0,1,\ldots,n-2$.
\end{prop}

\begin{proof}
For $(\bz,\by) = (z_j,y_j)_{j=1,\ldots,n-1}$
we show $\phi(n,n-2-k) \circ \tilde \rho \circ \phi(n,k)
(\bz,\by) = \tilde \rho(\bz,\by)$ by a direct calculation.
From \eqref{eq:phi-formula} we write 
$$
(\bz',\by') := \phi(n,k)(\bz,\by)
=
\left(y_{j} \frac{z_{k+2-j} + y_{j-1}}{z_{k+1-j} + y_{j}},
z_{j} \frac{z_{j+1} + y_{k-j}}{z_{j} + y_{k+1-j}} \right)_{j=1,\ldots,n-1},
$$
where we assume $z_n = y_n = 0$.
By substituting this into
\begin{align*}
\phi(n,n-2-k) \circ \tilde \rho(\bz',\by')
=
\left(y'_{n-j}\frac{z'_{k+j}+y'_{n-j+1}}{z'_{k+1+j}+y'_{n-j}},
z'_{n-j}\frac{z'_{n-j-1}+y'_{k+2+j}}{z'_{n-j}+y'_{k+1+j}} \right)_{j=1,\ldots,n-1},
\end{align*}
we see the claim.
\end{proof}

\section{Chamber variables}

As is typical in integrable systems, we will express various solutions to our systems in terms of tau-functions, by applying Hirota's bilinear method \cite{Hirota-book}.  A tau-function can be though of as an auxiliary collection of variables that have relations among themselves imposed by the original evolution equations.  Before proceeding, we define certain networks which provide a good visualization of how the tau-functions fit in.

\subsection{Wiring diagram and chamber variables}

Following \cite{GP16} we use the network model introduced in 
\cite{TP12} to describe our system.
Consider a semi-infinite cylinder with $n$ horizontal directed wires 
forming an infinite wiring diagram, which is a pictorial representation
of an infinite reduced word in $\hat{S}_n$.
Away from the crossings, the $n$ wires run along the cylinder at $n$ positions,
and corresponding to $s_i \in \hat{S}_n$ we cross the wires in positions 
$i$ and $i+1$. Corresponding to the parametrized version $s_i(a)$ of $s_i$, 
we assign the crossing with $a$, and call it a {\it vertex variable}. 
The Lusztig relation \eqref{eq:L-1} corresponds to the Yang-Baxter move of 
the wires in positions $i$, $i+1$ and $i+2$ as in Figure \ref{fig:YBmove}.

%
%
%
%
%
%
%
%
%
%

For $i=1,\ldots,n$, let $e_i$ be the $i$-th unit vector in $\Z^n$.
For integers $1 \leq i,j \leq n$, we set  
$$
  e_{[i,j]} = 
  \begin{cases}
  \sum_{k=i}^j e_k & \text{ if } i \leq j
  \\
  0 & \text{ if } i > j 
  \end{cases}.
$$
For $S \in \Z^n$, we often write $S_i$ for $S + e_i$, and 
$[S]$ for the class of $S$ in $\Z^n / \Z \,e_{[1,n]}$.
Following \cite[\S 5.1]{GP16} (cf. \cite{BFZ96}),
we label the chambers of the wiring diagram with elements of $\Z^n / \Z \, e_{[1,n]}$,
and at the left end of wires, label the wire at position $i$ with $\alpha_i \in \R_{>0}$.
Fix a label of one chamber, and extend it to the others  
by labeling the surrounding four chambers at each crossing of the wires,
as shown in Figure \ref{fig:chamber}.

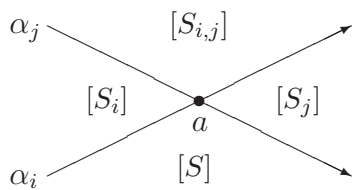
\begin{figure}[ht]
\unitlength=1mm
\begin{picture}(60,30)(0,0)
\put(5,5){\vector(2,1){40}}
\put(5,25){\vector(2,-1){40}}
\put(25,15){\circle*{1,5}}

\put(24,11){$a$}
\put(0,4){$\alpha_i$}
\put(0,24){$\alpha_j$}

\put(22,5){\small $[S]$}
\put(21,24){\small $[S_{i,j}]$}
\put(10,14){\small $[S_i]$}
\put(35,14){\small $[S_j]$}
\end{picture}
\caption{Chamber labeling at the crossing of wires $\alpha_i$ and $\alpha_j$,
where $a$ is the vertex variable associated with this crossing.}
\label{fig:chamber}
\end{figure}

The {\it chamber variables} make up a {\it tau-function} denoted $\tau$ whose domain is
$\Z^n$. We assume $\tau$ is periodic on a cylinder:
\begin{align}\label{eq:tau-delta}
  \tau(S + e_{[1,n]}) = \tau(S).
\end{align}
The enriched Yang-Baxter move, see Figure \ref{fig:eYBmove},
corresponds to the relation among the chamber variables
\begin{align}\label{eq:tautau}
(\alpha_i - \alpha_k) \tau(S_j) \, \tau(S_{i,k}) 
=
(\alpha_i-\alpha_j ) \tau(S_k) \, \tau(S_{i,j}) 
+ (\alpha_j-\alpha_k )\tau(S_i) \, \tau(S_{j,k}).
\end{align}

\begin{remark}
 Several names are used in the literature for the enriched Yang-Baxter move:
  {\it {Hirota bilinear difference equation}} \cite{Z}, {\it {discrete
analogue of generalized Toda equation}} and {\it {lattice KP equation}} \cite{N}, {\it { bilinear lattice KP equation}} \cite{ZFSZ}, or {\it {Hirota-Miwa equation}} \cite{LNQ}.
It goes back to the works of Miwa \cite{Mi} and Hirota \cite{Hirota-book}.
\end{remark}

\begin{figure}[ht]
\unitlength=1mm
\begin{picture}(120,40)(0,0)
\multiput(10,10)(35,0){2}{\line(1,0){5}}
\put(25,10){\line(1,0){10}}
\multiput(10,20)(35,0){2}{\line(1,0){5}}
\multiput(10,30)(25,0){2}{\line(1,0){15}}

\multiput(15,10)(20,0){2}{\line(1,1){10}}
\multiput(15,20)(20,0){2}{\line(1,-1){10}}
\put(25,20){\line(1,1){10}}
\put(25,30){\line(1,-1){10}}

\multiput(45,10)(0,10){3}{\vector(1,0){5}}

\put(6,9){\scriptsize $\alpha_i$}
\put(6,19){\scriptsize $\alpha_j$}
\put(6,29){\scriptsize $\alpha_k$}

\put(20,15){\circle*{1}} \put(16,14){\scriptsize $a$}
\put(30,25){\circle*{1}} \put(26,24){\scriptsize $b$}
\put(40,15){\circle*{1}} \put(36,14){\scriptsize $c$}

\put(28,4){\scriptsize $[S]$}
\put(8,14){\scriptsize $[S_i]$} \put(27,14){\scriptsize $[S_j]$} \put(45,14){\scriptsize $[S_k]$} 
\put(16,24){\scriptsize $[S_{i,j}]$} \put(37,24){\scriptsize $[S_{j,k}]$} 
\put(25,34){\scriptsize $[S_{i,j,k}]$}

\put(54,19){$\longleftrightarrow$}

\multiput(70,30)(35,0){2}{\line(1,0){5}}
\put(85,30){\line(1,0){10}}
\multiput(70,20)(35,0){2}{\line(1,0){5}}
\multiput(70,10)(25,0){2}{\line(1,0){15}}

\multiput(75,30)(20,0){2}{\line(1,-1){10}}
\multiput(75,20)(20,0){2}{\line(1,1){10}}
\put(85,20){\line(1,-1){10}}
\put(85,10){\line(1,1){10}}

\multiput(105,10)(0,10){3}{\vector(1,0){5}}

\put(66,9){\scriptsize $\alpha_i$}
\put(66,19){\scriptsize $\alpha_j$}
\put(66,29){\scriptsize $\alpha_k$}

\put(100,25){\circle*{1}} \put(96,24){\scriptsize $c'$}
\put(90,15){\circle*{1}} \put(86,14){\scriptsize $b'$}
\put(80,25){\circle*{1}} \put(76,24){\scriptsize $a'$}

\put(85,34){\scriptsize $[S_{i,j,k}]$}
\put(68,24){\scriptsize $[S_{i,j}]$} \put(87,24){\scriptsize $[S_{i,k}]$} \put(105,24){\scriptsize $[S_{k,j}]$} 
\put(76,14){\scriptsize $[S_{i}]$} \put(97,14){\scriptsize $[S_{k}]$} 
\put(88,4){\scriptsize $[S]$}

\end{picture}
\caption{The enriched Yang-Baxter move, where $(a',b',c')=(\frac{bc}{a+c},a+c,\frac{ab}{a+c})$.}
\label{fig:eYBmove}
\end{figure}
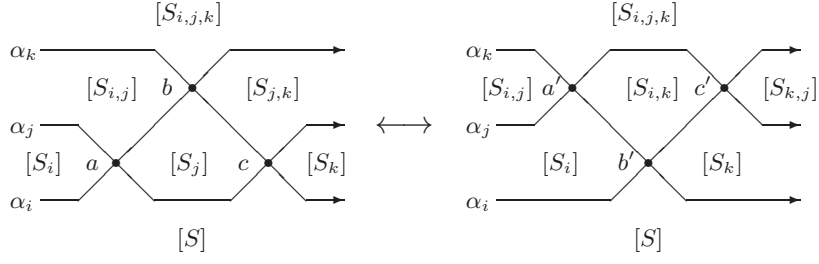

The following lemma is shown easily \cite[Lemma 3.1]{GP16}:

\begin{lem}\label{lem:YB-eYB}
The enriched Yang-Baxter move on chamber variables  
\eqref{eq:tautau} induces the Yang-Baxter move \eqref {eq:L-1}
on vertex variables  
via the transformation: 
\begin{align}
a = (\alpha_i - \alpha_j) 
    \frac{\tau(S_{i,j}) \, \tau(S)}{\tau(S_i) \, \tau(S_j)}
\end{align}
with labels as in Figure \ref{fig:chamber}.
\end{lem}

\subsection{The tau-function for $\phi(n,k)$.}\label{subsec:tau-phi}

We start with the wiring diagram of the semi-infinite word 
\eqref{eq:v-evolution} with $u$ and $v=v(k)$ given by \eqref{eq:Phi-uv},
as depicted at Figure \ref{fig:uv}.
At the left end of the diagram, the wire at position $i$ is assigned 
with $\alpha_i \in \R_{>0}$, and the chamber between the $n$-th and the $1$-st 
wires is labelled with $0 \in \Z^n$.
From the diagram we see that $n-2$ Yang-Baxter moves are applied 
in calculating the dynamics of $\phi(n,k)$, 
$(\bz_i^t,\by_i^t) \mapsto (\by_i^{t+1},\bz_{i+1}^t)$,
hence the corresponding enriched Yang-Baxter moves give 
$n-2$ relations among chamber variables:
\begin{align}
\begin{split}\label{eq:chamber1}
(\alpha_n& - \alpha_{k+1}) \, \tau(e_{[1,p-1]}+e_{k+1}) \, \tau(e_{[1,p]}-e_n)
\\
= &(\alpha_n - \alpha_{p}) \, \tau(e_{[1,p-1]}+e_{k+1}-e_n)\, \tau(e_{[1,p]})
\\
&+ (\alpha_p - \alpha_{k+1}) \,\tau(e_{[1,p-1]}) \,\tau(e_{[1,p]}+e_{k+1}-e_n)
; ~p= 1,\ldots,k,
\end{split}
\\
\begin{split}\label{eq:chamber2}
(\alpha_n& - \alpha_{k+1}) \, \tau(e_{[1,k+p]}+e_{k+1}) \, 
\tau(e_{[1,k+p+1]}-e_n)
\\
= &(\alpha_n - \alpha_{k+p+1}) \,\tau(e_{[1,k+p]}+e_{k+1}-e_n) \,\tau(e_{[1,k+p+1]})
\\
&+ (\alpha_{k+p+1} - \alpha_{k+1})\, \tau(e_{[1,k+p]})\, \tau(e_{[1,k+p+1]}+e_{k+1}-e_n); ~p= 1,\ldots,n-k-2,
\end{split}
\end{align}

\begin{figure}[ht]
\unitlength=1mm
\begin{picture}(150,110)(0,0)
\multiput(10,40)(0,10){2}{\line(1,0){5}}
\put(15,50){\vector(1,-1){50}}
\multiput(15,40)(10,-10){5}{\line(1,1){10}}
\put(10,30){\line(1,0){15}}
\put(10,20){\line(1,0){25}}
\put(10,10){\line(1,0){35}}

\put(35,100){\line(1,-1){40}}
\put(35,90){\vector(1,1){10}}
\multiput(45,80)(10,-10){3}{\line(1,1){10}}
\put(10,90){\line(1,0){25}}
\put(10,80){\line(1,0){35}}
\put(10,70){\line(1,0){45}}
\put(10,60){\line(1,0){55}}

\put(85,20){\line(1,1){60}}
\multiput(85,30)(10,10){6}{\line(1,-1){10}}
\put(55,20){\line(1,0){30}}
\put(45,30){\line(1,0){40}}
\put(35,40){\line(1,0){60}}
\put(25,50){\line(1,0){80}}
\put(75,60){\line(1,0){40}}
\put(75,70){\line(1,0){50}}
\put(65,80){\line(1,0){70}}
\put(55,90){\line(1,0){30}} \put(85,90){\vector(1,1){10}}
\put(85,100){\line(1,-1){10}} \put(95,90){\vector(1,0){55}}

\put(145,80){\vector(1,0){5}}
\put(145,70){\vector(1,0){5}}
\put(135,60){\vector(1,0){15}}
\put(125,50){\vector(1,0){25}}
\put(115,40){\vector(1,0){35}}
\put(105,30){\vector(1,0){45}}
\put(95,20){\vector(1,0){55}}
\put(65,10){\line(1,0){70}} \put(145,10){\vector(1,0){5}}
\put(135,10){\line(1,-1){10}} \put(135,0){\line(1,1){10}}

\put(6,9){\scriptsize $\alpha_n$}
\put(6,19){\scriptsize $\alpha_1$}
\put(6,29){\scriptsize $\alpha_2$}
\put(6,39){\scriptsize $\alpha_k$}
\put(3,49){\scriptsize $\alpha_{k+1}$}
\put(3,59){\scriptsize $\alpha_{k+2}$}
\put(3,69){\scriptsize $\alpha_{n-1}$}
\put(6,79){\scriptsize $\alpha_{n}$}
\put(6,89){\scriptsize $\alpha_{1}$}

\multiput(20,45)(10,-10){5}{\circle*{1}}
\multiput(40,95)(10,-10){4}{\circle*{1}}
\multiput(90,25)(10,10){6}{\circle*{1}}
\put(90,95){\circle*{1}}
\put(140,5){\circle*{1}}

\put(18,4){\scriptsize $-e_n$}
\put(18,14){\scriptsize $0$}
\put(18,24){\scriptsize $e_1$} 
\put(15,34){\scriptsize $e_{[1,2]}$}
\put(8,44){\scriptsize $e_{[1,k]}$}

\put(80,4){\scriptsize $e_{k+1}-e_{n-1}-e_n$}
\put(70,14){\scriptsize $e_{k+1}-e_n$}
\put(60,24){\scriptsize $e_{k+1}$}
\put(60,34){\scriptsize $e_{1} +e_{k+1}$}
\put(60,44){\scriptsize $e_{[1,2]} +e_{k+1}$}
\put(60,54){\scriptsize $e_{[1,k+1]}$}
\put(30,64){\scriptsize $e_{[1,k+2]}$}
\put(30,74){\scriptsize $e_{[1,n-1]}$}
\put(30,84){\scriptsize $e_{[1,n]}$}

\put(85,64){\scriptsize $e_{[1,k+1]}+e_{k+1}$}
\put(85,74){\scriptsize $e_{[1,k+2]}+e_{k+1}$}
\put(85,84){\scriptsize $e_{[1,n-1]}+e_{k+1}$}
\put(60,94){\scriptsize $e_{[1,n]}+e_{k+1}$}
\put(100,94){\scriptsize $e_{[1,n]}+e_1+e_{k+1}$}

\put(110,24){\scriptsize $e_{1} +e_{k+1}-e_n$}
\put(115,34){\scriptsize $e_{[1,2]} +e_{k+1}-e_n$}
\put(120,44){\scriptsize $e_{[1,k+1]} -e_n$}
\put(125,54){\scriptsize $e_{[1,k+1]} +e_{k+1}-e_n$}
\put(135,64){\scriptsize $e_{[1,k+2]} +e_{k+1}-e_n$}

\put(22,44){\scriptsize ${z_1}$}
\put(32,34){\scriptsize ${z_2}$}
\put(42,24){\scriptsize ${z_k}$}
\put(52,14){\scriptsize ${z_{k+1}}$}
\put(62,4){\scriptsize ${z_{k+2}}$}

\put(72,64){\scriptsize ${z_{n-1}}$}
\put(62,74){\scriptsize ${z_{k+2}}$}
\put(52,84){\scriptsize ${z_{k+1}}$}
\put(42,94){\scriptsize ${z_{k}}$}

\put(85,24){\scriptsize ${y_1}$}
\put(95,34){\scriptsize ${y_2}$}
\put(105,44){\scriptsize ${y_k}$}
\put(112,54){\scriptsize ${y_{k+1}}$}
\put(122,64){\scriptsize ${y_{k+2}}$}
\put(132,74){\scriptsize ${y_{n-1}}$}

\put(85,94){\scriptsize ${y_1}$}

\multiput(6,18)(3,0){49}{\line(1,0){1}}
\multiput(6,88)(3,0){49}{\line(1,0){1}}

\end{picture}
\caption{The wiring diagram for $v(k) u$ on the universal covering of 
a cylinder. A fundamental domain is between two dashed lines.}
\label{fig:uv}
\end{figure}

For a glide $g \in \hat{S}_n$ we define its {\it trajectory} $\kappa(g) \in \Z^n$ in the following way:
draw the wiring diagram of $g$, and label the wire at position $i$ with 
$\alpha_i$ at the left end.
Let $S$ (resp. $S'$) 
be the label of the chamber between wires $1$ and $2$ at the leftmost (resp. rightmost) point of the diagram for $g$.
Then the trajectory of $g$ is $\kappa(g) := S' - S$.

\begin{lem}\label{lem:traj-uv}
(i) The trajectories of $v:=v(k)$ and $\tilde{u} := \rho^{n-1}(u)$ are
\begin{align}\label{eq:traj}
\kappa(\tilde u) = -e_n,
\quad 
\kappa(v) = e_{k+1}.
\end{align}
(ii) 
When we label the chamber to the left of the vertex $y_{0,j}^{t=0}$
with $S \in \Z^n$, then that of the vertex $y_{i,j}^t$ is
$$
S + i \kappa(\tilde{u}) - t \kappa(v).
$$
\end{lem}

\begin{proof}
(i) The trajectory $\kappa(v)$ is the difference of the labels of the left chamber of $y_2$ and the left chamber of $z_k$, which is $(e_1+e_{k+1})-e_1$.
Since the offset of $v$ is $k_2$, in the diagram the vertex of 
$y_1$ is a crossing of the wires $n$ and $1$.
Hence the difference of the labels of the right chamber of $z_{k+1}$ and 
the left chamber of $y_1$, $(e_{k+1}-e_n) - e_{k+1}$, is the trajectory
of $\rho^{n-1}(u)$.
(ii) It follows from \cite[Proposition 5.4]{GP16}.
\end{proof}

Label the left chamber of $z_{i,k+1}^t$ with $\kappa(\tilde{u})i - \kappa(v) t$, and define
\begin{align}\label{eq:tau-it}
\tau^{i,t}(S) := \tau(S+\kappa(\tilde{u})i - \kappa(v) t)
\end{align}
for $S \in \Z^n$.
The following proposition immediately follows from Lemma \ref{lem:YB-eYB}.

\begin{prop}\label{prop:chamber-yz}
For $i,t \in \Z$, set 
\begin{align}\label{eq:yz-tau}
\begin{split}
&y_{i,p}^t = (\alpha_n - \alpha_p) 
\frac{\tau^{i,t}(e_{k+1}+e_{[1,p]}) \,\tau^{i,t}(e_{k+1}-e_n+e_{[1,p-1]})}
     {\tau^{i,t}(e_{k+1}+e_{[1,p-1]}) \,\tau^{i,t}(e_{k+1}-e_n+e_{[1,p]})}
; ~p= 1,\ldots,n-1,
\\
&z_{i,p}^t = (\alpha_{k+1-p}-\alpha_{k+1}) 
\frac{\tau^{i,t}(e_{k+1}+e_{[1,k+1-p]}) \, \tau^{i,t}(e_{[1,k-p]})}{\tau^{i,t}(e_{[1,k+1-p]})\,\tau^{i,t}(e_{[1,k-p]}+e_{k+1})}; ~p= 1,\ldots,k,
\\
&z_{i,k+p}^t = (\alpha_{n+1-p}-\alpha_{k+1})
\frac{\tau^{i,t}(e_{k+1}+e_{[1,n-p+1]}) \, \tau^{i,t}(e_{[1,n-p]})}
{\tau^{i,t}(e_{[1,n-p+1]})\,\tau^{i,t}(e_{k+1}+e_{[1,n-p]})}; ~p= 1,\ldots,n-k-1.
\end{split}
\end{align}
Then a solution of \eqref{eq:chamber1} and \eqref{eq:chamber2}
gives the solution for $\Phi(n,k)$. 
\end{prop}

Via tropicalization, we see the following facts.
Let $T$ be the tropical chamber variables (tau-functions) on $\Z^n$ 
satisfying the tropicalization of \eqref{eq:chamber1} and \eqref{eq:chamber2}:
\begin{align}
\begin{split}\label{eq:trop-tau1}
&T(e_{[1,p-1]}+e_{k+1}) + T(e_{[1,p]}-e_n)
\\
& \quad = \min[ T(e_{[1,p-1]}+e_{k+1}-e_n) + T(e_{[1,p]}),
\\
&\qquad \qquad \quad 
A_p - A_n + T(e_{[1,p-1]}) + T(e_{[1,p]}+e_{k+1}-e_n)]; ~p= 1,\ldots,k,
\end{split}
\\
\begin{split}\label{eq:trop-tau2}
&T(e_{[1,k+p]}+e_{k+1}) + T(e_{[1,k+p+1]}-e_n)
\\
& \quad 
= \min[ T(e_{[1,k+p]}+e_{k+1}-e_n) + T(e_{[1,k+p+1]}),
\\
& \qquad \qquad \quad 
A_{k+p+1} - A_n + T(e_{[1,k+p]}) + T(e_{[1,k+p+1]}+e_{k+1}-e_n)]; ~p= 1,\ldots,n-k-2.
\end{split}
\end{align}
A solution for these equations gives that of $\Phi(n,k)$ via 
\begin{align}\label{eq:yz-Tau}
\begin{split}
&Y_{i,p}^t = A_n + 
T^{i,t}(e_{k+1}+e_{[1,p]}) + T^{i,t}(e_{k+1}-e_n+e_{[1,p-1]})
\\ &\qquad \quad 
- T^{i,t}(e_{k+1}+e_{[1,p-1]}) - T^{i,t}(e_{k+1}-e_n+e_{[1,p]});
~p= 1,\ldots,n-1,
\\
&Z_{i,p}^t = A_{k+1-p} +
T^{i,t}(e_{k+1}+e_{[1,k+1-p}]) + T^{i,t}(e_{[1,k-p]})
\\ &\qquad \quad 
- T^{i,t}(e_{[1,k+1-p]})-T^{i,t}(e_{[1,k-p]}+e_{k+1}); ~p= 1,\ldots,k,
\\
&Z_{i,k+p}^t = A_{n+1-p} + 
T^{i,t}(e_{k+1}+e_{[1,n-p+1]}) + T^{i,t}(e_{[1,n-p]})
\\ &\qquad \quad 
- T^{i,t}(e_{[1,n-p+1]})- T^{i,t}(e_{k+1}+e_{[1,n-p]}); ~p= 1,\ldots,n-k-1.
\end{split}
\end{align}
Here we set 
$T^{i,t}(S) = T(S+\kappa(\tilde{u}) i - \kappa(v) t)$.

\section{Soliton solutions for $\Phi(n,k)$ via tropicalization 
of geometric solution}
\label{sec:trop-soliton1}

We study the soliton solution for $\Phi(n,k)$ on $\R$ by tropicalizing 
those for $\phi(n,k)$ on $\R_{>0}$ in \cite{GP16}.
In this section we prove the following theorem.

\begin{thm}\label{thm:trop-sol}
Tropicalizing the soliton solutions for $\phi(n,k)$ studied in 
\cite{GP16}, we get the soliton solutions for $\Phi(n,k)$ whose minimal form 
is $(x,\ldots,x)$ for $A_n < x < A$,
where $A = \min [A_i; ~i \in \{1,\ldots,n-1\} \setminus \{k+1\}]$.
Their velocity is $1/(n-1)$, independent of $x$.
\end{thm}

Besides the condition \eqref{eq:A-condition}, for simplicity we assume that 
all $A_i$ are distinct. 
Using $\{i_p; ~p=1,\ldots,n\} = \{1,\ldots,n \}$ we write the ordering of 
the $A_i$ as
$$
A_{i_1} < A_{i_2} < \cdots < A_{i_n} 
$$
where $i_1 = n$ and $i_n=k+1$.
We define $A_{[i_1,\ldots,i_p]} := A_{i_1} + A_{i_2} +\cdots + A_{i_p}$.
Let $K = \C\{\{t\}\} := \cup_{n \geq 1}\C((t^{1/n}))$ be the field of 
Puiseux series over $\C$. Let $\val$ be the valuation map 
$\val: K \to \R \cup \{\infty\}$ (see Appendix \ref{sec:app2}
for the precise definition.)
For $k=1,\ldots,n$, fix $\alpha_k \in K$ to satisfy $\val(\alpha_k) = A_k$.
Let $b, c \in K$ satisfy
\begin{align}\label{eq:b-c}
(b-\alpha_1)(b-\alpha_2) \cdots (b-\alpha_n)
=
(c-\alpha_1)(c-\alpha_2) \cdots (c-\alpha_n),
\end{align}
and $A_{i_p} \leq \val(b) < \val(c) < A_{i_{p+1}}$ for some $p \in \{1,\ldots,n-1\}$.

Let $F(X,Y)$ be the tropical polynomial: 
$$
F(X,Y) = \min[Y, nX, (n-1)X + A_{i_1}, (n-2)X + A_{[i_1,i_2]}, \ldots,
X+A_{[i_1,i_{n-1}]},A_{[i_1,i_n]}].
$$
The affine tropical curve $\Gamma$ determined by $F(X,Y)$ is a graph in $\R^2$
defined as
$$
\Gamma = \{(X,Y) \in \R^2 ~|~ F(X,Y) \text{ is indifferentiable} \}.
$$ 
See Figure \ref{fig:Gamma}. 
The tropical curve $\Gamma$ is the tropicalization of 
the affine curve $\gamma$ in $K^2$ given by
$$
y = y(x) = (x-\alpha_1)(x-\alpha_2) \cdots (x-\alpha_n).
$$
Precisely, $\Gamma$ is the closure of the set of valuations of points 
in $\gamma$.

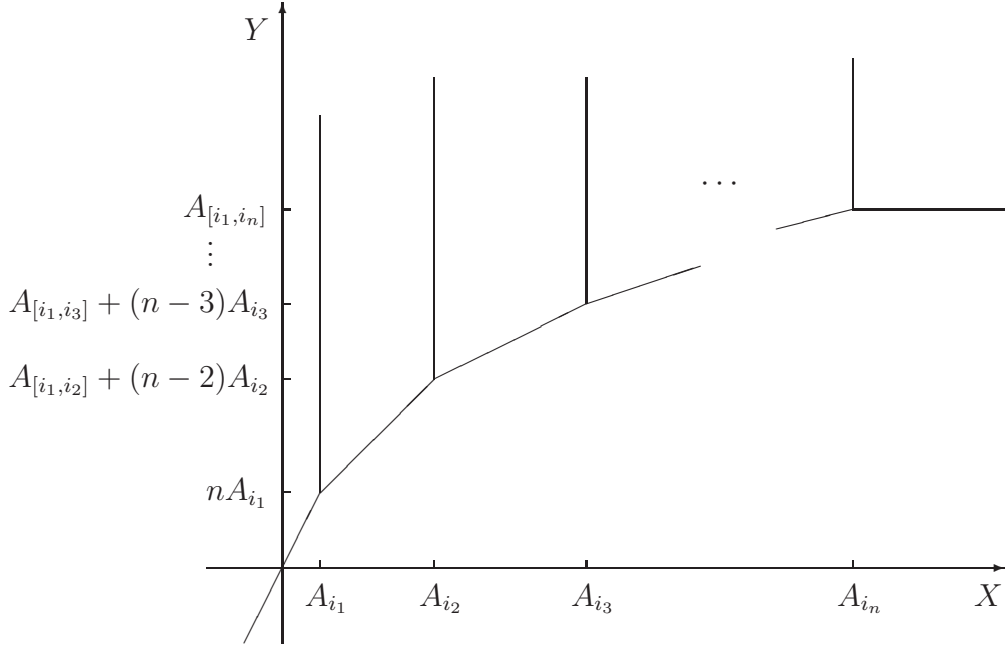
\begin{figure}[ht]
\unitlength=1mm
\begin{picture}(110,90)(-10,5)
\put(0,20){\vector(1,0){105}}
\put(10,10){\vector(0,1){85}}
\put(5,10){\line(1,2){10}}
\put(15,30){\line(1,1){15}}
\put(30,45){\line(2,1){20}}
\put(50,55){\line(3,1){15}}
\put(75,65){\line(4,1){10}}
\put(85,67.5){\line(1,0){20}}

\put(15,30){\line(0,1){50}}
\put(30,45){\line(0,1){40}}
\put(50,55){\line(0,1){30}}
\put(65,70){$\cdots$}
\put(85,67.5){\line(0,1){20}}

\put(15,20){\line(0,1){1}}
\put(30,20){\line(0,1){1}}
\put(50,20){\line(0,1){1}}
\put(85,20){\line(0,1){1}}
\put(13,15){$A_{i_1}$}
\put(28,15){$A_{i_2}$}
\put(48,15){$A_{i_3}$}
\put(83,15){$A_{i_n}$}

\put(10,30){\line(1,0){1}}
\put(10,45){\line(1,0){1}}
\put(10,55){\line(1,0){1}}
\put(10,67.5){\line(1,0){1}}
\put(0,29){$nA_{i_1}$}
\put(-26,44){$A_{[i_1,i_2]}+(n-2)A_{i_2}$}
\put(-26,54){$A_{[i_1,i_3]}+(n-3)A_{i_3}$}
\put(0,60){$\vdots$}

\put(-3,66.5){$A_{[i_1,i_n]}$}

\put(101,15){$X$}
\put(5,90){$Y$}
\end{picture}
\caption{The tropical curve $\Gamma$ separating $\R^2$ into $n+2$ domains.}
\label{fig:Gamma}
\end{figure}

Let $X_c$ satisfy $A_{i_p} < X_c < A_{i_{p+1}}$ for some 
$p \in \{1,2,,\ldots,n-1\}$. 
We consider a pair of intersection points of the curve $\Gamma$ and a line  
parallel to the $X$-axis:
$$
(A_{i_p},(n-p)X_c + A_{[i_1,i_p]}),
\quad
(X_c,(n-p)X_c + A_{[i_1,i_p]}).
$$  
This pair is the image under $\val$ of a pair of points $(b,y(b))$ and $(c,y(c))$ 
on the curve $\gamma$, such that $y(b) = y(c)$ and  
$A_{i_p} \leq \val(b) < \val(c)=X_c < A_{i_{p+1}}$.
Note that it turns out that $\val(b) = A_{i_p}$ from the graph $\Gamma$.

\begin{lem}\label{lem:valBC}
It holds that 
\begin{align*}
&\val(b + \alpha_i) 
= 
\begin{cases}
A_i & i=i_1,\ldots,i_{p-1}, \\
A_{i_p} + (n-p)(X_c-A_{i_p})  & i=i_p, \\
A_{i_p} & i=i_{p+1},\ldots,i_n.
\end{cases}
\\
&\val(c + \alpha_i) 
= 
\begin{cases}
A_i & i=i_1,\ldots,i_p, \\
X_c & i=i_{p+1},\ldots,i_n.
\end{cases}
\end{align*}
\end{lem}

\begin{proof}
By the assumption on $X_c$ and Lemma~\ref{lem:non-arch-eq}, we obtain 
$\val(c - \alpha_i)$ for all $i$ 
and $\val(b-\alpha_i)$ for $i \neq i_p$ immediately.
As for $\val(b-\alpha_{i_p})$, taking the valuation of \eqref{eq:b-c};  
$$
\sum_{i=1}^n \val(b-\alpha_i) = \sum_{i=1}^n \val(c-\alpha_i),
$$
it follows that
\begin{align*}
\val(b-\alpha_{i_p})
&=
\sum_{i=1}^n \val(c-\alpha_i) - \sum_{i \neq i_p} \val(b-\alpha_i)
\\
&=
A_{[i_1,i_p]} +(n-p)X_c - A_{[i_1,i_{p-1}]}- (n-p)A_{i_p},  
\end{align*}
and we obtain the result.
\end{proof}

The one soliton solution for $\phi(n,k)$ constructed in \cite{GP16} is 
given in terms of chamber variable
\begin{align}\label{eq:1-soliton-tau}
\tau(S) = 1+ \ell B_1^{s_1} B_2^{s_2} \cdots B_n^{s_n};
\quad S=(s_1,\ldots,s_n) \in \Z^n, ~ \ell \in K \setminus \{0\},
\end{align}
where 
\begin{align}
B_j = \frac{b-\alpha_j}{c-\alpha_j}; ~j=1,\ldots,n.
\end{align}
Due to \eqref{eq:b-c}, it holds that $\tau(S+e_{[1,n]}) = \tau(S)$. 
We are interested in the tropicalization of the tau-function.
As a corollary of Lemma \ref{lem:valBC} we obtain the following.

\begin{cor}\label{cor:trop-B}
We have 
$$
\val(B_i) =  
\begin{cases}
0 & i=i_1,\ldots,i_{p-1}, \\
(n-p)(X_c-A_{i_p}) & i=i_p, \\
A_{i_p}-X_c & i=i_{p+1},\ldots,i_n.
\end{cases}
$$
\end{cor}

Remark that $\sum_{i=1}^n \val(B_i) = 0$ holds.
Define $W(X_c) := (\val(B_i))_{i=1,\ldots,n} \in \R^n$.
Now it is easy to show the following.

\begin{prop}\label{prop:W-uv}
(i) The one soliton solution for $\Phi(n,k)$ is given by
\begin{align}\label{eq:GP-tau}
T(S)
= 
\min[0, L + S \cdot W(X_c)]; \quad  S \in \Z^n, ~L \in \R.
\end{align}
(ii) We have
\begin{align*}
&W(X_c) \cdot \kappa(\rho^{-1}(u))
=
\begin{cases}
-(n-1)(X_c - A_{i_1}) & \text{for } A_{i_1} < X_c < A_{i_2}, 
\\
0 & \text{for } A_{i_2} < X_c < A_{i_n}. 
\end{cases}
\\[1mm]
&W(X_c) \cdot \kappa(v) 
=
A_{i_p} - X_c  \quad \text{for } A_{i_p} < X_c < A_{i_{p+1}}.
\end{align*}
In particular, the soliton exists
only when $A_{i_1} < X_c < A_{i_2}$, and its velocity is $(n-1)^{-1}$,
independent of $X_c$. 
\end{prop}

\begin{proof}
(i) It follows from Corollary \ref{cor:trop-B}.
(ii) Due to Lemma \ref{lem:traj-uv} (ii), 
when $W(X_c) \cdot \kappa(\rho^{-1}(u)) \neq 0$, the velocity 
of the soliton is given by (Cf. \cite[\S 8.1]{GP16})
$$
  \frac{W(X_c) \cdot \kappa(v)}{W(X_c) \cdot \kappa(\rho^{-1}(u))}.
$$
The result is obtained by using Lemma \ref{lem:traj-uv} (i) and Corollary 
\ref{cor:trop-B}.
\end{proof}

By substituting \eqref{eq:GP-tau} into \eqref{eq:yz-Tau}, we obtain  
the soliton solution whose minimal form is 
$(\underbrace{X_c,\ldots,X_c}_{n-1})$ 
with $A_n  < X_c < A_{i_2}$.
Since all of these have velocity $1/(n-1)$, there is no scattering
among solitons. It ends the proof of Theorem \ref{thm:trop-sol}.

\begin{remark}
For any $N \geq 1$, an $N$-soliton solution corresponding to 
$X_n < X_{c_1} \leq X_{c_2} \leq \cdots \leq X_{c_N} < A$
is given by 
$$ 
  T(S) = \min\left[0, \min[ L_i + S \cdot W(X_{c_i}); ~i=1,\ldots,N] \right].
$$
with appropriate $L_i \in \R$.
\end{remark}

\begin{remark}\label{rem:trop-gamma1}
When we set $A := A_{i_2} = \cdots = A_{i_{n-1}}$ as \eqref{eq:Phi-comm},
the tropical curve $\Gamma$ is degenerated and $X_c$ has only two possibilities
$A_n < X_c < A$ or $A < X_c < A_{k+1}$.
The similar calculation of valuation shows 
$$
W(X_c) = 
\begin{cases}
(X_c-A_n)(-1,\ldots,-1,n-1)  & \text{for } A_n < X_c < A,
\\[1mm]
\displaystyle{\frac{X_c-A}{n-2}}\,
(\underbrace{1,\ldots,1}_{k},-(n-2),\underbrace{1,\ldots,1}_{n-k-2},0)  & \text{for } A < X_c < A_{k+1}.
\end{cases}
$$
Thus we see that $W(X_c) \cdot \kappa(\rho^{-1}(u))$ and $W(X_c) \cdot \kappa(v)$ 
have the same expression as Proposition \ref{prop:W-uv} (ii),
and solitons only have volocity $(n-1)^{-1}$.
In the limit $X_c \to A$,
it gives the positive soliton whose minimal form is 
$(\underbrace{1,1,\ldots,1}_{n-1})$ in Theorem \ref{thm:solitons}.
\end{remark}

\section{Soliton solutions for $\Phi(n,k)$ from tropical tau function}
\label{sec:1soliton}

We study the soliton solutions for $\Phi(n,k)$ by naively solving
the tropical bilinear equations \eqref{eq:trop-tau1} and \eqref{eq:trop-tau2},
instead of tropicalizing the geometric solutions.
It turns out that there are various solitons besides those presented in 
\S \ref{sec:trop-soliton1}, a reflection of the general phenomenon that tropicalization is not reversible.
We first present the solution for $\Phi(3,k)$ on $\R$,
and next show the solution for $\Phi(n,k)$ on $\Z$ for general $n$.   
Theorem \ref{thm:solitons} finally follows from 
Propositions \ref{prop:3-onesoliton} and \ref{prop:n-onesoliton}.

\subsection{The $n=3$ case}

Let $\tau(i,t)$ be a function of $(i,t) \in \Z^2$ satisfying a relation
\begin{align}\label{eq:n=3tau}
\begin{split}
&(1+\delta) \tau(i,t-1) \,\tau(i+2,t+1)
\\
& \qquad = \tau(i+1,t+1) \, \tau(i+1,t-1) + \delta \,\tau(i,t) \,\tau(i+2,t).
\end{split}
\end{align}

We first demonstrate the case of $\Phi(3,1)$ in detail.
\begin{prop}
A solution of \eqref{eq:n=3tau} gives the solution of $\phi(3,1)$ via 
\begin{align}\label{eq:yz-n=3}
\begin{split}
&y_{i,1}^t = (\alpha_3 - \alpha_1) 
\frac{\tau(i+1,t) \,\tau(i+1,t-1)}
     {\tau(i,t-1) \,\tau(i+2,t)},
~~
y_{i,2}^t = (\alpha_3 - \alpha_2) 
\frac{\tau(i+1,t-1) \,\tau(i+2,t)}
     {\tau(i+1,t) \,\tau(i+2,t-1)},
\\
&z_{i,1}^t = (\alpha_1-\alpha_2) 
\frac{\tau(i+1,t) \,\tau(i,t)}
     {\tau(i+1,t+1) \,\tau(i,t-1)},
~~
z_{i,2}^t = (\alpha_3-\alpha_2)
\frac{\tau(i,t-1) \,\tau(i+1,t)}
     {\tau(i,t) \,\tau(i+1,t-1)}.
\end{split}
\end{align}
\end{prop}
\begin{proof}
From \eqref{eq:chamber1} and \eqref{eq:chamber2},  
the bilinear equation for chamber variables $\tau$ of $\phi(3,1)$ is
\begin{align}\label{eq:n=3cham}
\begin{split}
&(\alpha_3 - \alpha_2) \tau(S+e_1-e_3) \, \tau(S+e_2) 
\\
& \quad =
(\alpha_3-\alpha_1 ) \tau(S+e_2-e_3) \, \tau(S+e_1) 
+ (\alpha_1-\alpha_2 )\tau(S) \, \tau(S+e_1+e_2-e_3).
\end{split}
\end{align}
We change the chamber coordinate generated by the unit vectors 
$e_k ~(k=1,2,3)$ to that generated by  
$\kappa(\tilde u) = -e_3$, $\kappa(v)=e_2$ and $e_{[1,3]}$.
By using $e_1=e_{[1,3]}+\kappa(\tilde u)-\kappa(v)$, $e_2 = \kappa(v)$ and
$e_3 = - t(u)$, and transforming $\tau(i,t,j) := \tau(S)$  
when $S = i \cdot \kappa(\tilde u) - t \cdot \kappa(v) + j \cdot e_{[1,3]}$,
we rewrite \eqref{eq:n=3cham} and obtain 
\begin{align*}
&(\alpha_3 - \alpha_2) \tau(i+2,t+1,j+1) \, \tau(i,t-1,j) 
\\
& \quad = (\alpha_3-\alpha_1 )\tau(i+1,t-1,j+1) \, \tau(i+1,t+1,j)  
+ (\alpha_1-\alpha_2 ) \tau(i,t,j) \,\tau(i+2,t,j+1).
\end{align*} 
We set $\delta = (\alpha_1-\alpha_2)/(\alpha_3-\alpha_1)$,
and ignore the third coordinate by taking into account the periodic condition
\eqref{eq:tau-delta}. Then \eqref{eq:n=3tau} is obtained.

By identifying $\tau^{i,t}(i' \cdot \kappa(\tilde u) - t' \cdot \kappa(v) + j' \cdot e_{[1,3]})$ with $\tau(i+i',t+t',j')$ which reduces to $\tau(i+i', t+t')$,
we obtain \eqref{eq:yz-n=3} from  \eqref{eq:yz-tau}.
\end{proof}

For $\phi(3,0)$, instead of \eqref{eq:n=3cham} we have 
\begin{align*}
\begin{split}
&(\alpha_3-\alpha_1 ) \tau(S+e_{[1,2]}-e_3) \, \tau(S+2 e_1) 
\\
& \quad =
(\alpha_3 - \alpha_2) \tau(S+2 e_1-e_3) \, \tau(S+e_{[1,2]}) 
+ (\alpha_2-\alpha_1 )\tau(S+e_1) \, \tau(S+e_{[1,2]}+e_1-e_3).
\end{split}
\end{align*}
In the same manner using $\kappa(v) = e_1$ and 
$\delta = (\alpha_2-\alpha_1)/(\alpha_3-\alpha_2)$,
we again obtain \eqref{eq:n=3tau}.  
Then a solution of \eqref{eq:n=3tau} gives the solution of $\phi(3,0)$ via
\begin{align}\label{eq:yz-n=3b}
\begin{split}
&y_{i,1}^t = (\alpha_3 - \alpha_1) 
\frac{\tau(i,t-2) \,\tau(i+1,t-1)}
     {\tau(i,t-1) \,\tau(i+1,t-2)},
~~
y_{i,2}^t = (\alpha_3 - \alpha_2) 
\frac{\tau(i+1,t-1) \,\tau(i+1,t-2)}
     {\tau(i,t-2) \,\tau(i+2,t-1)},
\\
&z_{i,1}^t = (\alpha_3-\alpha_1) 
\frac{\tau(i,t-1) \,\tau(i+1,t)}
     {\tau(i,t) \,\tau(i+1,t-1)},
~~
z_{i,2}^t = (\alpha_2-\alpha_1)
\frac{\tau(i+1,t-1) \,\tau(i,t-1)}
     {\tau(i+1,t) \,\tau(i,t-2)},
\end{split}
\end{align}
which originate from \eqref{eq:yz-tau}.

In the following we set $\trop(\delta) = 1$, corresponding to 
the condition \eqref{eq:Phi-comm}.
We study the soliton solutions for $\Phi(3,k)$
by solving the tropicalization of \eqref{eq:n=3tau}: 
\begin{align}\label{eq:n=3Tau}
\begin{split}
&T(i,t-1) + T(i+2,t+1)
\\
& \quad 
= \min \left[T(i+1,t+1) + T(i+1,t-1), ~1 + T(i,t) + T(i+2,t)\right].
\end{split}
\end{align}

\begin{prop}\label{prop:n=3-soliton}
The soliton solutions of \eqref{eq:n=3Tau} are given by
\begin{enumerate}
\item one-soliton: 
\begin{align}\label{eq:n=3-onesoliton}
  T(i,t) = 
  \begin{cases}
  \min[0,L+i P -t] & L \in \R, ~P \geq 2,
  \\
  \min \left[0,L+i P - t \frac{P}{2}  \right] & L \in \R, ~0 < P \leq 2.
  \end{cases}
\end{align}

\item two-soliton:
Define $C_j(i,t) := L_j + i P_j - t$ and $C'_j(i,t) := L_j + i P_j - t P_j/2$.
$$
  T(i,t) = 
  \begin{cases}
  \min[0,C_1(i,t),C_2(i,t), C_1(i,t) + C_2(i,t) + Z_{1,2}];
  \\ \hspace{4cm} 
   L_j\in \R, ~P_j \geq 2 \text{ such that } ~P_1 \neq P_2,
  \\
  \min[0,C_1(i,t),C'_2(i,t), C_1(i,t) + C'_2(i,t) + Z'_{1,2}]; 
  \\ \hspace{4cm} 
  L_j\in \R, ~0 < P_2 < 2 \leq P_1, 
  \end{cases}
$$
where 
$$
  Z_{1,2} := 2 \min[P_1,P_2] - 1, 
  \qquad
  Z'_{1,2} := \frac{3}{2} P_2.  
$$
\end{enumerate}
\end{prop}

\begin{proof}
We compute the soliton solutions using Hirota's method
\cite[\S 1.5]{Hirota-book}.
(1) A one soliton solution for \eqref{eq:n=3tau} is set to be a form as 
$\tau(i,t) = 1 + \ell p^i q^t$
with $\ell, p, q \in \R_{>0}$. By substituting it into \eqref{eq:n=3tau}, we 
obtain an algebraic equation for $p$ and $q$:
$$
(1+\delta)(q^{-1}+p^2 q) = pq + p q^{-1} + \delta(1+p^2).
$$  
Then $P:=\trop(p)$ and $Q:=\trop(q)$ are required to satisfy
\begin{align}\label{eq:PQ}
\min \left[-Q,~ 2 P +Q] = \min[P+Q, ~P-Q,~ 1+\min[0,~2P] \right].
\end{align}
We solve this equation assuming $P > 0$ without loss of generality.
When $Q > 0$, \eqref{eq:PQ} reduces to
$-Q = \min [P-Q,1]$, which has no solution.
When $Q \leq 0$, \eqref{eq:PQ} reduces to
$\min[-Q,2P+Q] = \min[P+Q,1]$.
If $-Q \leq 2P+Q$ and $1 \leq P+Q$, we have $Q=-1$ and $P \geq 2$.
If $-Q \leq 2P+Q$ and $1 > P+Q$, we have $Q=-P/2$ and $0< P < 2$.
If $-Q > 2P-Q$, we have no solution.

(2) By substituting a form of two soliton solution 
$\tau(i,t) = 1 + \ell_1 p_1^i q_1^t + \ell_2 p_2^i q_2^t + z_{1,2} \ell_1 p_1^i q_1^t \ell_2 p_2^i q_2^t$ with $\ell_i, p_i, q_i > 0$ for $i=1,2$
into \eqref{eq:n=3tau}, and taking the order of $\ell_1 \ell_2$, 
we have
$$
  z_{1,2} = \frac{p_1 p_2(q_1^2+q_2^2)+\delta q_1 q_2(p_1^2 + p_2^2) - (1+\delta)(p_1^2 q_1^2 + p_2^2 q_2^2)}
{(1+\delta)(1+p_1^2 q_1^2 p_2^2 q_2^2) - (p_1 p_2(1+\delta p_1 p_2 q_1 q_2)+q_1 q_2(\delta + p_1 p_2 q_1 q_2))}.
$$
Set $\trop(q_i) = Q_i = -1$ and $\trop(p_i)=P_i \geq 2$ for $i=1,2$ such that 
$P_1 \neq P_2$.
In the numerator of $z_{1,2}$, we have
\begin{align*}
&\trop(p_1 p_2(q_1^2+q_2^2)+\delta q_1 q_2(p_1^2 + p_2^2)) = 
\min[P_1+P_2-2,2P_1 -1, 2 P_2-1] = 2 \min[P_1,P_2]-1,
\\
&\trop((1+\delta)(p_1^2 q_1^2 + p_2^2 q_2^2)) =
2 \min[P_1,P_2]-2,
\end{align*}
and in the denominator, we have 
\begin{align*}
&\trop((1+\delta)(1+p_1^2 q_1^2 p_2^2 q_2^2)) = 
\min[0,2 (P_1+P_2)-4] = 0,
\\
&\trop(p_1 p_2(1+\delta p_1 p_2 q_1 q_2)+q_1 q_2(\delta + p_1 p_2 q_1 q_2)) =
\min[P_1+P_2,-1,P_1+P_2-4] = -1.
\end{align*}
Thus the dominant terms of tropicalization 
in the numerator and denominator of $z_{i,2}$
have the same sign, and we obtain 
$$
Z_{1,2} := \trop(z_{1,2}) =  2 \min[P_1, P_2]-2 - (-1) = 2 \min[P_1, P_2]-1.
$$
In the same manner, when $Q_1 = -1$, $P_1 \geq 2$, $Q_2= -P_2/2$ and $0 < P_2 < 2$, we obtain $\trop(z_{1,2}) = 3/2 P_2$.

\end{proof}

By substituting \eqref{eq:n=3-onesoliton} into the tropicalization of 
\eqref{eq:yz-n=3} or \eqref{eq:yz-n=3b}
we obtain the following:

\begin{prop}\label{prop:3-onesoliton}
Assume $L \in \Z$ and $P \in \Z_{\geq 2}$ in \eqref{eq:n=3-onesoliton}.
Then 
we obtain a positive soliton whose minimal form is $(P-1,1)$ for $\Phi(3,1)$, 
and a positive soliton whose minimal form is 
$(1,P-1)$ for $\Phi(3,0)$. 
In both cases the velocity of soliton is $1/P$.
\end{prop}

The proof is included in that of Proposition \ref{prop:n-onesoliton} (ii).

\begin{remark}
For general $L \in \R$ and $P \in \R_{>0}$, 
a soliton may have the minimal length more than one.
See \S \ref{subsec:phase}. 
\end{remark}

\subsection{General $n$ case}

In the same way as the $n=3$ case, we study soliton solutions 
for $\Phi(n,k)$ for $n>3$. Unfortunately we obtain only one-soliton solutions 
for technical reason.  
For simplicity we study only integral solutions.

We transform the bilinear equations for $\tau^{i,t}$ on the chambers, 
\eqref{eq:chamber1} and \eqref{eq:chamber2}, 
to those for an $(n-2)$-tuple of tau-functions 
$(\tau_p)_{p=0,1,\ldots,n-3}$ on $\Z^2$
using the following rule.
Recall the definition of $\tau^{i,t}$ \eqref{eq:tau-it} and that we have 
labelled the left chamber of $z_{i,k+1}^t$ with $\kappa(\tilde{u})i - \kappa(v) t$.  
Note that in the universal covering of the wiring diagram for $v(k) u$, 
the chamber labels lie in the subset of $\Z^n$
$$ 
\mathcal{C}_k :=
\left\{e_{[1,p]} + i \cdot e_n + t \cdot e_{k+1}+ j \cdot e_{[1,n]}; ~
p \in \{0,1,\ldots,n-1\}, i,t,j \in \Z \right\}.
$$ 
When $k \geq  1$, we set 
\begin{align}\label{eq:tau-trans-k>0}
\tau^{i,t}(e_{[1,p]}) = 
\begin{cases}
\tau_{0}(i,t) & \text{ if $p=0$,}
\\
\tau_{p}(i+1,t+1) & \text{ if $p=1,\ldots,k$},
\\
\tau_{p-1}(i+1,t) & \text{ if $p=k+1,\ldots,n-2$},
\\
\tau_{0}(i+1,t) & \text{ if $p=n-1$},
\end{cases}
\end{align}
and uniquely extend it to $\mathcal{C}_k$ using
\begin{align}\label{eq:tau-trans}
\tau^{i,t}(-i' e_n - t'e_{k+1}+j e_{[1,n]}) = \tau^{i+i',t+t'}(0)
\end{align}
for any $i',t',j \in \Z$.
For instance, when $n=4$ and $k= 2$, 
it holds that $\tau^{i,t}(e_{k+1} + e_{[1,2]}) = 
\tau^{i,t-1}(e_{[1,2]}) = \tau_2(i+1,t)$.
Remark that the periodic condition for chamber variables \eqref{eq:tau-delta}
is hidden by \eqref{eq:tau-trans}.
When $k=0$, instead of \eqref{eq:tau-trans-k>0} we set 
\begin{align*}
\tau^{i,t}(e_{[1,p]}) = 
\begin{cases}
\tau_0(i,t) & \text{ if $p=0$},
\\
\tau_{0}(i,t-1) & \text{ if $p=1$},
\\
\tau_{p-1}(i+1,t) & \text{ if $p=2,\ldots,n-2$},
\\
\tau_{0}(i+1,t) & \text{ if $p=n-1$},
\end{cases}
\end{align*}
and use \eqref{eq:tau-trans}.

\begin{prop}
Set $\alpha_i = \alpha$ for $i=\{1,2,\ldots,n-1\}\setminus\{k+1\}$, 
and define $\delta := (\alpha - \alpha_{k+1})/(\alpha_n-\alpha)$. 
Via the above introduced transformation,
\eqref{eq:chamber1} and \eqref{eq:chamber2} reduce to
the following equations for the $\tau_p(i,t)$, which are `independent' of $k$:
\begin{align}\label{eq:n-tau1}
\begin{split}
&(1+\delta) \tau_0(i,t-1) \,\tau_1(i+2,t+1)
\\
& \qquad = \tau_1(i+1,t+1) \, \tau_0(i+1,t-1) 
 + \delta \,\tau_0(i,t) \,\tau_1(i+2,t),
\end{split}
\\
\label{eq:n-tau2}
\begin{split}
&(1+\delta) \tau_p(i+1,t-1) \,\tau_{p+1}(i+2,t)
\\
& \qquad = \tau_{p+1}(i+1,t) \, \tau_{p}(i+2,t-1) 
+ \delta \,\tau_p(i+1,t) \,\tau_{p+1}(i+2,t-1); ~p=1,\ldots,n-3,
\end{split}
\end{align}
where we assume $\tau_{n-2}(i,t) = \tau_0(i,t)$.
\end{prop}

\begin{proof}
When $k \geq 1$, from \eqref{eq:chamber1} and \eqref{eq:chamber2} we obtain
\begin{align*}
\begin{split}
&(\alpha_n-\alpha_{k+1}) \tau_0(i,t-1) \,\tau_1(i+2,t+1)
\\
& \qquad = (\alpha_n-\alpha_1)\tau_1(i+1,t+1) \, \tau_0(i+1,t-1) 
 + (\alpha_1-\alpha_{k+1}) \,\tau_0(i,t) \,\tau_1(i+2,t),
\end{split}
\\
\begin{split}
&(\alpha_n-\alpha_{k+1}) \tau_{p-1}(i+1,t) \,\tau_{p+1}(i+2,t+1)
\\
& \qquad = (\alpha_{n}-\alpha_{p})\tau_{p}(i+1,t+1) \, \tau_{p-1}(i+2,t) 
+ (\alpha_{p}-\alpha_{2}) \,\tau_{p-1}(i+1,t+1) \,\tau_{p}(i+2,t); 
\\
& \qquad \qquad p=2,\ldots,k,
\end{split}
\\
\begin{split}
&(\alpha_n-\alpha_{k+1}) \tau_{p-2}(i+1,t-1) \,\tau_{p-1}(i+2,t)
\\
& \qquad = (\alpha_{n}-\alpha_{p})\tau_{p-1}(i+1,t) \, \tau_{p-2}(i+2,t-1) 
+ (\alpha_{p}-\alpha_{k+1}) \,\tau_{p-2}(i+1,t) \,\tau_{p-1}(i+2,t-1); 
\\
& \qquad \qquad p=k+2,\ldots,n-1.
\end{split}
\end{align*}
Using the defined $\delta$, we see that the first equation reduces to 
\eqref{eq:n-tau1}, and the next two equations reduce to 
\eqref{eq:n-tau2}.

When $k=0$, we only have \eqref{eq:chamber2} which turns out to be 
\eqref{eq:n-tau1} when $p=1$, and 
\eqref{eq:n-tau2} when $p=2,3,\ldots,n-2$.

\end{proof}

From Proposition \ref{prop:chamber-yz} we obtain the following:

\begin{cor}
A solution of \eqref{eq:n-tau1} and \eqref{eq:n-tau2} gives 
the solution of $\phi(n,k)$ via the following formulae.
In the case of $k \geq 1$:
\begin{align}
\label{eq:y-tau-1}
&y_{i,1}^t = (\alpha_n - \alpha) 
\frac{\tau_1(i+1,t) \,\tau_0(i+1,t-1)}
     {\tau_0(i,t-1) \,\tau_1(i+2,t)},
\\
\label{eq:y-tau-p1}
&y_{i,p}^t = (\alpha_n - \alpha) 
\frac{\tau_{p}(i+1,t) \,\tau_{p-1}(i+2,t)}
     {\tau_{p-1}(i+1,t) \,\tau_{p}(i+2,t)}; ~ p=2,\ldots,k,
\\
\label{eq:y-tau-k+1}
&y_{i,k+1}^t = (\alpha_n - \alpha) 
\frac{\tau_k(i+1,t-1) \,\tau_k(i+2,t)}
     {\tau_k(i+1,t) \,\tau_k(i+2,t-1)},
\\
\label{eq:y-tau-p2}
&y_{i,p}^t = (\alpha_n - \alpha) 
\frac{\tau_{p-1}(i+1,t-1) \,\tau_{p-2}(i+2,t-1)}
     {\tau_{p-2}(i+1,t-1) \,\tau_{p-1}(i+2,t-1)}; ~ p=k+2,\ldots,n-1,
\\
&z_{i,p}^t = (\alpha-\alpha_{k+1}) 
\frac{\tau_{k+1-p}(i+1,t) \,\tau_{k-p}(i+1,t+1)}
     {\tau_{k+1-p}(i+1,t+1) \,\tau_{k-p}(i+1,t)}; ~p=1,\ldots,k-1,
\\
\label{eq:z-tau-k}
&z_{i,k}^t = (\alpha-\alpha_{k+1})
\frac{\tau_1(i+1,t) \,\tau_0(i,t)}
     {\tau_1(i+1,t+1) \,\tau_0(i,t-1)},
\\
\label{eq:z-tau-k+1}
&z_{i,k+1}^t = (\alpha_n-\alpha_{k+1})
\frac{\tau_{0}(i,t-1) \,\tau_{0}(i+1,t)}
     {\tau_{0}(i,t) \,\tau_{0}(i+1,t-1)}.
\\
\label{eq:z-tau}
&z_{i,k+p}^t = (\alpha-\alpha_{k+1})
\frac{\tau_{n-p}(i+1,t-1) \,\tau_{n-p-1}(i+1,t)}
     {\tau_{n-p}(i+1,t) \,\tau_{n-p-1}(i+1,t-1)}; ~ p=2,\ldots,n-k-1.
\end{align}
In the case of $k = 0$:
\begin{align*}
&y_{i,1}^t = (\alpha_n - \alpha) 
\frac{\tau_0(i,t-2) \,\tau_0(i+1,t-1)}
     {\tau_0(i,t-1) \,\tau_0(i+1,t-2)},
\\
&y_{i,2}^t = (\alpha_n - \alpha) 
\frac{\tau_1(i+1,t-1) \,\tau_0(i+1,t-2)}
     {\tau_0(i,t-2) \,\tau_1(i+2,t-1)},
\end{align*}
and $y_{i,p}^t$ for $p=3,\ldots,n-1$ has the same expression as 
\eqref{eq:y-tau-p2}.
The variable $z_{i,1}^t$ is as \eqref{eq:z-tau-k+1},
$z_{i,p}^t$ for $p=2,\ldots,n-2$ is as \eqref{eq:z-tau}.
The variable $z_{i,n-1}^{t+1}$ has the same expression as \eqref{eq:z-tau-k}. 
\end{cor}

In the following, we set 
$\trop(\delta) = 1$ corresponding to \eqref{eq:Phi-comm},
and study solutions of the tropicalization of \eqref{eq:n-tau1}
and \eqref{eq:n-tau2}:
\begin{align}
\label{eq:n-ttau1}
\begin{split}
&T_0(i,t-1) + T_1(i+2,t+1)
\\
& \quad 
= \min \left[T_1(i+1,t+1) + T_0(i+1,t-1),  1 + T_0(i,t) + T_1(i+2,t)\right],
\end{split}
\\
\label{eq:n-ttau2}
\begin{split}
&T_p(i+1,t-1) + T_{p+1}(i+2,t)
\\
& \quad 
= \min \left[T_{p+1}(i+1,t) + T_p(i+2,t-1),  1 + T_p(i+1,t) + T_{p+1}(i+2,t-1)\right],
\\
& \qquad \qquad p=1,\ldots,n-3.
\end{split}
\end{align}

\begin{prop}\label{prop:n-onesoliton}
(i) A one-soliton solution for $\Phi(n,k)$ is given by
\begin{align*}
&T_0(i,t) = \min[0, L+iP-t],
\\
&T_p(i,t) = \min[0, L+iP-t+ \sum_{j=p}^{n-3} R_j]; \quad p=1,\ldots,n-3,
\end{align*}
where $L \in \Z$, $P \in \Z_{\geq n-1}$, $R_j \in \Z_{\leq -1}$ 
for $j=1,\ldots,n-3$, and
$P+\sum_{j=1}^{n-3} R_j \geq 2$.
\\
(ii) The minimal form of this soliton is 
$$
(P-1+\sum_{j=1}^{n-3}R_j , -R_1, -R_2 ,\ldots,-R_{k-1},1, -R_k,\ldots,-R_{n-3}).
$$
\end{prop}

\begin{proof}
(i) We substitute $\tau_0(i,t) = 1+ \ell c^i q^t$ and 
$\tau_s(i,t) = 1+ \ell c^i q^t \prod_{k=s}^{n-3} r_k$ into 
\eqref{eq:n-tau1} and \eqref{eq:n-tau2}, and obtain 
\begin{align}
\label{eq:one-soliton1}
&(1+\delta)(c^{-1}+c^2 q r) = c q r + c q^{-1} + \delta(1+ c^2 r),
\\
\label{eq:one-soliton2}
&(1+\delta)(q^{-1} r_s+c) = 1+ c q^{-1}r_s + \delta(r_s+ c q^{-1});
\quad s=1,\ldots,n-3,
\end{align}
where $r := \prod_{k=1}^{n-3} r_k$.
By eliminating the $r_s$, we have
\begin{align}\label{eq:f-g}
(1+\delta-c-\delta q)^{n-2} = c q (q+\delta c -(1+\delta) c q)^{n-2}.
\end{align}
Assume $P := \trop(c) > 0$ and set $Q:=\trop(q)$ as in the $n=3$ case. 

When $n$ is odd, set $n_0 := \frac{n-3}{2}$.
Eq. \eqref{eq:f-g} is expanded as
\begin{align}\label{eq:n-odd}
\begin{split}
&\sum_{i=0}^{n_0} \begin{pmatrix} n-2 \\ 2i \end{pmatrix}
(1+\delta)^{n-2-2i}(c+\delta q)^{2i}
+ c q \sum_{i=0}^{n_0} \begin{pmatrix} n-2 \\ 2i+1 \end{pmatrix}
(q+\delta c)^{n-3-2i} (c q(1+\delta))^{2i+1}
\\
&=
\sum_{i=0}^{n_0} \begin{pmatrix} n-2 \\ 2i+1 \end{pmatrix}
(1+\delta)^{n-3-2i}(c+\delta q)^{2i+1}
+ pq \sum_{i=0}^{n_0} \begin{pmatrix} n-2 \\ 2i \end{pmatrix}
(q+\delta c)^{n-2-2i} (c q(1+\delta))^{2i}.
\end{split}
\end{align}
Assume $Q \in \Z_{<0}$. 
It holds that $\trop(c+\delta q) = \min[P,Q+1] = Q+1 \leq 0$
and $\trop(q+\delta c) = \min[Q,1+P] = Q$. 
Thus the tropicalization of \eqref{eq:n-odd} reduces to
$$
\min[(n-3)(Q+1), 2P+2Q+(n-3)Q] 
= \min[(n-2)(Q+1), P+Q+(n-2)Q],
$$
from which we obtain $\min[n-3,2P+2Q]= Q +\min[P+Q,n-2].$
When $P+Q \geq n-2$, it follows that $Q=-1$ and $P \geq n-1$.
When $n_0 < P+Q < n-2$, it holds that $n-3 = P + 2 Q$ which contradicts 
$Q \in \Z_{<0}$.
When $P+Q \leq n_0$, we obtain $P=0$ which is a contradiction.
When $Q \geq 0$, there is no solution since 
the tropicalization of the l.h.s. and the r.h.s of \eqref{eq:n-odd}
are respectively zero and positive.

When $n$ is even, by a similar discussion we see that
the solution is that
$Q = -1$ and $P \in \Z_{\geq n-1}$.

From \eqref{eq:one-soliton2} with $Q = -1$ and $P \geq n-1$, we obtain
$\min[R_s+1,P] = \min[0,P+1+R_s,1+R_s,2+P]$,
which holds for any $R_s \leq -1$.
Further, from \eqref{eq:one-soliton1}
it follows that $P+\sum_{k=1}^{n-3} R_k \geq 2$. 
\\
(ii) Define $\Delta_p := \sum_{i=p}^{n-3} R_i$ for $p=1,\ldots,n-3$.
We show that the minimal form is obtained  
by setting $t=1$, $i=-1$ and $L = -\Delta_k$, for $k \geq 2$.
The case of $k=1$ is similar.

First, we consider $Y_{i,k+1}^{t=1}$ with $L = -\Delta_k$. By tropicalizing 
\eqref{eq:y-tau-k+1} we have 
\begin{align*}
Y_{i,k+1}^{t=1} 
&= T_k(i+1,0) + T_k(i+2,1) - T_k(i+1,1) - T_k(i+2,0)
\\
&= \min[0,(i+1)P] + \min[0,(i+2)P-1] - \min[0,(i+1)P-1] - \min[0,(i+2)P].
\end{align*}
When $i = -1$, all but the third term are zero in the last line of 
the above formula, and we obtain $Y_{-1,k+1}^{1} = 1$.
It is satisfied that $\min[0,*] = *$ for all terms when $i < -1$, 
and that $\min[0,*] = 0$ for all terms when $i > -1$. 
Thus $Y_{i,k+1}^{t=1} = 0$ when $i \neq -1$.

Next we calculate $Y_{i,1}^{1}$. By tropicalizing \eqref{eq:y-tau-1} we have
\begin{align*}
Y_{i,1}^{1} 
&= T_1(i+1,1) + T_0(i+1,0) - T_0(i,0) - T_1(i+2,1)
\\
&= \min[0,L+(i+1)P-1+\Delta_1] + \min[0,L+(i+1)P]
\\ 
& \qquad - \min[0,L+iP] - \min[0,L+(i+2)P-1+\Delta_1].
\end{align*}
When $i=-1$, the first and the third terms in the above formula are nonzero,
and we obtain 
$Y_{-1,1}^{1} = (L-1+\Delta_1) - (L-P) = P + \Delta_1-1$.
When $i \neq -1$, we obtain $Y_{i,1}^1 = 0$ for the same reason as for $Y_{i,k+1}^{1}$.

In the case of $Y_{i,k}^{1}$, from \eqref{eq:y-tau-p1} we obtain
\begin{align*}
Y_{i,k}^{1} 
&= T_k(i+1,1) + T_{k-1}(i+2,1) - T_{k-1}(i+1,1) - T_k(i+2,1)
\\
&= \min[0,(i+1)P -1] + \min[0,L+(i+2)P-1+\Delta_{k-1}] 
\\
& \qquad - \min[0,L+(i+1)P-1+\Delta_{k-1}] - \min[0,(i+2)P-1].
\end{align*}
Using the conditions $L+\Delta_{k-1} = R_{k-1}$, $R_j \leq -1$ and 
$P+\sum_{k=1}^{n-3} R_k \geq 2$, we see that 
$Y_{i,k}^{1} = 0$ when $i \neq -1$ in the same way as $Y_{i,k+1}^1$.
When $i=-1$, we obtain $Y_{-1,k}^{1} = -R_{k-1}$. 
Similarly, we obtain 
$Y_{-1,p}^{1} = -R_{p-1}$ when $p=2,\ldots,k-1$, 
$Y_{-1,p}^{1} = -R_{p-2}$ when $p=k+2,\ldots,n-1$,
and $Y_{i,p}^{1} = 0$ otherwise.
\end{proof}

\begin{remark}
Eq. \eqref{eq:one-soliton2} requires that the $r_s$ should be the same
for all $s=1,\ldots,n-3$,
but the tropicalization of \eqref{eq:one-soliton2} is weaker so that
the $R_s$ can differ. 
\end{remark}

\section{Duality with box-ball system}\label{sec:BBS}

\subsection{Basics of the $\mathfrak{sl}_n$ box-ball system}\label{subsec:BBS}

The $\mathfrak{sl}_n$ box-ball system (BBS) is a cellular automaton,
defined as a dynamical system of finitely many balls with $n-1$ `colors'
in an infinite number of boxes arranged along a line. 
(See \cite{IKT12} for a review and a list of comprehensive references 
of the BBS.)
The original BBS mentioned in Introduction is the simplest case of $n=2$.
In this paper
we study the case that $n \geq 3$ and each box can contain one ball at most.
By writing $1$ for an empty box and $p$ for a box containing a $p$-ball 
(a ball of color $p$) for $p=2,\ldots,n$, 
we represent a configuration of the system
as a (infinite) word on $\{1,2,\ldots,n \}$.
The evolution of one time step $t \to t + 1$ is given as follows:
Do the following procedure for $p$ from $n$ down to $2$.
\begin{enumerate}
\item[(i)] 
Exchange the leftmost $p$ with its nearest $1$ to the right.
\item[(ii)]
Exchange the leftmost $p$ among the rest of the $p$ with its $1$ to the right.
\item[(iii)]
Repeat (ii) until all of the $p$ are moved exactly once.
\end{enumerate}
The resulting word corresponds to the configuration at time $t + 1$.

It is known that a soliton of the system is 
a nonincreasing sequence of $2,\ldots,n$, 
and that a collection of such sequences gives a multi-soliton solution, 
in a sense of (i)--(iii) in \S \ref{subsec:soliton}.

\begin{example}\label{ex:bbsn=3}
$n=3$
\\
(i) two solitons:
\begin{align*}
t=0: & 1133321111321111111111111111 \cdots
\\
t=1: & 1111113332113211111111111111 \cdots
\\
t=2: & 1111111111332133211111111111 \cdots
\\
t=3: & 1111111111111321133321111111 \cdots
\\
t=4: & 1111111111111113211113332111 \cdots
\end{align*}
(ii) three solitons:
\begin{align*}
t=0: & 113322211113221112111111111111111111111111 \cdots
\\
t=1: & 111111133222113221211111111111111111111111 \cdots
\\
t=2: & 111111111111332113122222111111111111111111 \cdots
\\
t=3: & 111111111111111332311111222221111111111111 \cdots
\\
t=4: & 111111111111111111233311111112222211111111 \cdots
\\
t=5: & 11111111111111111121133311111111122222111 \cdots
\end{align*}

\end{example}

The symmetry of the $\mathfrak{sl}_n$-BBS is known to be described by 
the $\hat{\mathfrak{sl}}_n$-crystal for the symmetric tensor representation 
of $U_q'(\hat{\mathfrak{sl}}_n)$ \cite{FOY,HHIKTT}. 
Now we present the minimum needed prerequisites concerning the $\hat{\mathfrak{sl}}_n$-crystal for the BBS. 
Recall the $\hat{\mathfrak{sl}}_n$-crystal $B_\ell$ corresponding to 
the $\ell$-fold symmetric tensor representation of 
$U_q'(\hat{\mathfrak{sl}}_n)$, and given by
\eqref{eq:cryetalB} as a set.
Let $R_{m \ell}: ~B_m \otimes B_\ell \stackrel{\sim}{\to} B_\ell \otimes B_m$
be the combinatorial $R$-matrix defined by 
$$
  R_{m \ell}
  ; \quad  \bw \otimes \bx \mapsto \bx' \otimes \bw',
$$
where 
\begin{align}\label{eq:sln-comb-R}
\begin{split}
&x_i' = x_i + K_{i+1} - K_i, \qquad 
w_i' = w_i + K_{i} - K_{i+1},
\\
&K_i := K_i(\bx,\bw) 
= \min_{j=0,1,\ldots,n-1}
  \left[\sum_{p=1}^j w_{i-p} + \sum_{p=j+2}^n x_{i-p} \right].
\end{split}
\end{align}
Here we assume that the subscripts of $w_i$ and $x_i$ are modulo $n$. 

With the notion of a `carrier' which moves balls, 
the combinatorial $R$-matrix describes the above time evolution in the following way.
The configuration space of the BBS is $B_1^{\otimes L}$ 
for some large number $L$, where
$\bw = (w_i)_{i=1,\ldots,n} \in B_1$ denotes an empty box (resp. a box containing a $p$-ball) 
when $w_1=1$ (resp. $w_p=1$) and the other $w_i$ are zero.
A carrier of capacity $\ell$ is an element in $B_\ell$.
Write $\bw_i^t \in B_1$ for a state at the $i$-th component of $B_1^{\otimes L}$,
at time $t$.
We assume that the initial carrier $\bx_0 = (x_{0,i})_{i=1,\ldots,n} \in B_\ell$ has no ball, i.e., $x_{0,1}=\ell$ and the other $x_{0,i}$ are zero,
and that $\bw_i^{t=0}$ for $i \gg 1$ is an empty box.
Then the time evolution is given by applying the combinatorial $R$-matrix as 
\begin{align}\label{eq:evol-BBS}
\begin{split}
&\bw_L^t \otimes \cdots \otimes \bw_2^t \otimes \bw_1^t \otimes \bw_0^t \otimes \bx_0
\\
&\mapsto 
\bw_L^t \otimes \cdots \otimes \bw_2^t \otimes \bw_1^t \otimes R_{1 \ell}(\bw_0^t \otimes \bx_0) 
=
\bw_L^t \otimes \cdots \otimes \bw_2^t \otimes \bw_1^t \otimes \bx_1^t \otimes \bw_0^{t+1} 
\\
&\mapsto  
\bw_L^t \otimes \cdots \otimes \bw_2^t \otimes R_{1 \ell}(\bw_1^t \otimes \bx_1^t) \otimes \bw_0^{t+1} 
= 
\bw_L^t \otimes \cdots \otimes \bw_2^t \otimes \bx_2^t \otimes \bw_1^{t+1} \otimes \bw_0^{t+1} 
\\
&\mapsto 
\qquad \cdots \qquad 
= 
\bx_0 \otimes \bw_L^{t+1} \otimes \cdots \otimes \bw_2^{t+1} \otimes \bw_1^{t+1} \otimes \bw_0^{t+1} 
\end{split}
%
\end{align}
where we denote 
$R_{1 \ell}(\bw_i^t \otimes \bx_{i}^t) = \bx_{i+1}^{t} \otimes \bw_{i}^{t+1}$ for 
$i \geq 0$ and set $\bx_0^t = \bx_0$, at any time $t$.
Note that $R_{1 \ell}(\bw \otimes \bx_0) = \bx_0 \otimes \bw$ holds 
if $\bw$ is an empty box,  
thus we have $\bx_i^t = \bx_0$ for $i \gg 1$ by the assumption.  
In the limit $\ell \to \infty$, the original $\mathfrak{sl}_n$-BBS 
is obtained.
We remark that in \eqref{eq:evol-BBS}, 
`right and left' is opposite to that in the original description of the BBS
(i)--(iii).   
We also write the action of the $R$-matrix  using a diagram:

\begin{figure}[H]
\unitlength=1.5mm
\begin{picture}(20,20)(0,0)
\put(5,10){\vector(1,0){10}}
\put(10,15){\vector(0,-1){10}}

\put(1,9){$\bw_{i}^{t}$}
\put(16,9){$\bw_{i}^{t+1}$}
\put(9,17){$\bx_i^t$}
\put(9,2){$\bx_{i+1}^t$}
\end{picture}
\caption{BBS by the combinatorial $R$-matrix}
\label{fig:comb-R}
\end{figure}

We will use the following lemma later.

\begin{lem}\label{lem:shift-R}
Let $\sigma \in S_n$ be 
$$
\sigma = 
\begin{pmatrix}
1 & 2 & 3 & \cdots & n \\
1 & n & n-1 & \cdots & 2 
\end{pmatrix},
$$
and let $\rho$ be a map on $\Z^n \otimes \Z^n$
given by 
$(a_1,a_2,\ldots,a_n) \otimes (b_1,b_2,\ldots,b_n)
\mapsto (b_{\sigma(1)},\ldots,b_{\sigma(n)}) \otimes 
(a_{\sigma(1)},\ldots,a_{\sigma(n)})$. 
Then it holds that 
$R_{m \ell} \circ \rho = \rho \circ R_{\ell m}$.
\end{lem}

We omit the proof, as it is easy. 
Note that $\rho$ induces a map 
$B_m \otimes B_\ell \to B_\ell \otimes B_m$ for any $m, \ell$,
and that $\rho \circ \rho$ is an identity.

\subsection{Observation}

Our claim is that the positive soliton solutions of $\Phi(n,k)$ is `dual'
to those of the BBS. 
Precisely, the dynamics of carriers (resp. states) in $\Phi(n,k)$ for positive solitons 
coincides with that of states (resp. carriers) of the $\mathfrak{sl}_n$-BBS.

\begin{conjecture}\label{conj:Z-fin}
When we have only the positive solitons in $\Phi(n,k)$,
the carriers $\bZ_i^t$ take values in a finite set 
$$
M := \left\{ 
\begin{matrix}
m_p:=
(\underbrace{1, \ldots, 1}_{p-1},0,\underbrace{1, \ldots, 1}_{n-p-1}); 
~p=1,\ldots,n-1, 
\\
m_n := (1,1,\ldots,1) 
\end{matrix}
\right\} \subset \{0,1\}^{n-1}. 
$$ 
\end{conjecture}

Define a map $\beta_k : M \to \tilde{B}_1:=\{(x_2,x_3,\ldots,x_n) \in (\Z_{\geq 0})^{n-1}; ~\sum_{i=2}^n x_i \leq 1 \}$ by 
\begin{align*}
m_p \mapsto m_n - m_{p-1-k} = 
\begin{cases}
(\underbrace{0, \ldots, 0}_{p-2-k},1,\underbrace{0, \ldots, 0}_{n-p+k})
& p > k,
\\
(\underbrace{0, \ldots, 0}_{n+p-2-k},1,\underbrace{0, \ldots, 0}_{-p+k})
& p \leq k,
\end{cases}
\end{align*}
for $p=1,\ldots,n$ and $k=0,1,\ldots,n-2$,
where the subscript $i$ of $m_i$ is modulo $n$.

\begin{conjecture}\label{thm:BBS}
Assume that the initial configuration $(\bY_i^0)_i$ of $\Phi(n,k)$ 
includes only positive solitons.
By interchanging space and time coordinates
the rules of state and carrier are swapped
i.e. $\bY_i^t$ and $\bZ_i^t$ are respectively regarded as a carrier and a state at time $i$ of space $t$.  
Then, via the map $\beta_k$, the dynamics of the $\bZ_i^t$ is identified with 
that of the $\mathfrak{sl}_n$-BBS
where $(0,\ldots,0) \in \tilde{B}_1$ denotes $1$ (an empty box) and 
$(\underbrace{0,\ldots,0}_{p-2},1,\underbrace{0,\ldots,0}_{n-p}) \in \tilde{B}_1$ 
denotes $p$ (a box containing a $p$-ball) for $p=2,\ldots,n$.
If $(\bY_i^t)_i$ includes a soliton of the minimal form 
$(b_1,b_2,\ldots,b_{n-1})$ with $b_{k+1}=1$, 
then the corresponding configuration of BBS includes a soliton as
$$
\underbrace{n \ldots n}_{b_1} \underbrace{n-1 \ldots n-1}_{b_2} 
\underbrace{n-2 \ldots n-2}_{b_3}
\ldots \underbrace{2 \ldots 2}_{b_{n-1}},
$$ 
whose velocity is $\sum_{p=1}^{n-1} b_p$. 
\end{conjecture}

\begin{example}
We show $(\bZ_{i}^t)_i$ for each $t$,
corresponding Example \ref{ex:one-solitonsY} and \ref{ex:n=3Y}.
Non-initial states are coloured in red.
\\
$\Phi(3,1)$ (Example \ref{ex:one-solitonsY} (i)):
\begin{align*}
t = 0: & (10)(10)({\color{red}01})(10)(10)(10)(10)(10)(10)(10)(10)(10)
\\
t = 1: & (10)(10)({\color{red}01})(10)(10)(10)(10)(10)(10)(10)(10)(10)
\\
t = 2: & (10)(10)({\color{red}01})(10)(10)(10)(10)(10)(10)(10)(10)(10)
\\
t = 3: & (10)(10)({\color{red}11})(10)(10)(10)(10)(10)(10)(10)(10)(10)
\\
t = 4: & (10)(10)(10)({\color{red}01})(10)(10)(10)(10)(10)(10)(10)(10)
\\
t = 5: & (10)(10)(10)({\color{red}01})(10)(10)(10)(10)(10)(10)(10)(10)
\\
t = 6: & (10)(10)(10)({\color{red}01})(10)(10)(10)(10)(10)(10)(10)(10)
\\
t = 7: & (10)(10)(10)({\color{red}11})(10)(10)(10)(10)(10)(10)(10)(10)
\end{align*}
$\Phi(4,1)$ (Example \ref{ex:one-solitonsY} (ii)):
\begin{align*}
t = 0: & (101)(101)({\color{red}011})(101)(101) (101) (101) (101) (101)(101)
\\
t = 1: & (101)(101)({\color{red}011})(101)(101) (101) (101) (101) (101)(101)
\\
t = 2: & (101)(101)({\color{red}011})(101)(101) (101) (101) (101) (101)(101)
\\
t = 3: & (101)(101)({\color{red}111})(101) (101) (101) (101) (101) (101) (101) 
\\
t = 4: & (101)(101)({\color{red}110})(101) (101) (101) (101) (101) (101) (101)
\\
t = 5: & (101)(101)({\color{red}110})(101) (101) (101) (101) (101) (101) (101) 
\\
t = 6: & (101) (101)(101) ({\color{red}011}) (101) (101) (101) (101) (101) (101)\\
t = 7: & (101) (101)(101) ({\color{red}011}) (101) (101) (101) (101) (101) (101
\\
t = 8: & (101) (101)(101) ({\color{red}011}) (101) (101) (101) (101) (101) (101 \\
t = 9: & (101) (101)(101)({\color{red}111})(101) (101) (101) (101) (101) (101)
\\
t = 10: & (101) (101)(101)({\color{red}110})(101) (101) (101) (101) (101) (101)
\end{align*}
$\Phi(3,1)$ (Example \ref{ex:n=3Y} (i)):
\begin{align*}
t = 0: & (10)(10)({\color{red}01})(10)({\color{red}01})(10)(10)(10)(10)(10)(10)(10)
\\
t = 1: & (10)(10)({\color{red}11})(10)({\color{red}01})(10)(10)(10)(10)(10)(10)(10)
\\
t = 2: & (10)(10)(10)({\color{red}01})({\color{red}11})(10)(10)(10)(10)(10)(10)(10)
\\
t = 3: & (10)(10)(10)({\color{red}11})(10)({\color{red}01})(10)(10)(10)(10)(10)(10)
\\
t = 4: & (10)(10)(10)(10)({\color{red}01})({\color{red}11})(10)(10)(10)(10)(10)(10)
\\
t = 5: & (10)(10)(10)(10)({\color{red}01})(10)({\color{red}01})(10)(10)(10)(10)(10)
\\
t = 6: & (10)(10)(10)(10)({\color{red}11})(10)({\color{red}11})(10)(10)(10)(10)(10)
\\
t = 7: & (10)(10)(10)(10)(10)({\color{red}01})(10)({\color{red}01})(10)(10)(10)(10)
\\
t = 8: & (10)(10)(10)(10)(10)({\color{red}01})(10)({\color{red}11})(10)(10)(10)(10)
\\
t = 9: & (10)(10)(10)(10)(10)({\color{red}01})(10)(10)({\color{red}01})(10)(10)(10)
\\
t = 10: & (10)(10)(10)(10)(10)({\color{red}11})(10)(10)({\color{red}11})(10)(10)(10)
\end{align*}
The last case is dual with Example \ref{ex:bbsn=3} (i).
\end{example}

\begin{remark}
The positive solitons for $\Phi(n,k)$ do not correspond to all 
$\mathfrak{sl}_n$-BBS solitons; a BBS soliton related to a positive soliton of $\Phi(n,k)$ 
should include
at least one $p$-ball for $p \in \{2,\ldots,n\}\setminus \{n-k\}$, 
and exactly one $(n-k)$-ball. 
\end{remark}

\subsection{Strategy to prove Conjectures \ref{conj:Z-fin} and \ref{thm:BBS}}
\label{subsec:strategy}

Let $\tilde B_\ell$ be a set as
$$
  \tilde B_\ell = \{\bx = (y_1,y_2,\ldots,y_{n-1}) \in (\Z_{\geq 0})^{n-1};
  ~ \sum_{i=1}^n y_i \leq \ell \}.
$$
We have a natural isomorphism $\gamma_\ell: B_\ell \to \tilde B_\ell$ given by
$(x_1,x_2,\ldots,x_n) \mapsto (x_2,x_3,\ldots,x_n)$,
where the inverse map $\gamma_\ell^{-1}$ is given by
$(y_1,y_2,\ldots,y_{n-1}) \mapsto 
(\ell-\sum_{i=1}^{n-1} y_i,y_1,\ldots,y_{n-1})$.
Define $\tilde R_{m \ell}$ as
$$
  \tilde R_{m \ell}:= (\gamma_\ell \otimes \gamma_m) \circ R_{m \ell} \circ 
(\gamma_m^{-1} \otimes \gamma_\ell^{-1}): 
\tilde B_m \otimes \tilde B_\ell \stackrel{\sim}{\to} 
\tilde B_\ell \otimes \tilde B_m.
$$
Recall the maps $\iota$ on $\R^{n-1}$ and $\tilde \rho$ on $(\R^{n-1})^2$ 
defined in Section \ref{sec:phi-formula}.
We use the same notations  $\iota$  and $\tilde \rho$ 
for their restrictions on $\Z$.
We identify $(\Z^{n-1})^2$ with $\Z^{n-1} \otimes \Z^{n-1}$, 
following the expression of the combinatorial $R$-matrix.

\begin{lem}\label{lem:rho-R-beta}
The followings hold:
\\
(i) 
$(\gamma_\ell \otimes \gamma_m) \circ \rho \circ 
(\gamma_m^{-1} \otimes \gamma_\ell^{-1}) = \tilde \rho$.
\\
(ii)
$\tilde \rho \circ \tilde R_{m \ell} \circ \tilde \rho
= \tilde R_{\ell m}$ on $\tilde B_\ell \otimes \tilde B_m$.
\\
(iii) 
$\tilde \rho \circ (\beta_k \otimes \iota) 
= (\iota \otimes \beta_{n-2-k}) \circ \tilde \rho$  on $M \otimes \tilde B_m$,
for $k=0,1,\ldots,n-2-k$.
\end{lem}
\begin{proof}
(i) is easy. (ii) follows from (i) and Lemma \ref{lem:shift-R}.
We check (iii). For $m_p \otimes \bw \in M \otimes \tilde B_m$, we have 
\begin{align*}
  &(\iota \otimes \beta_{n-2-k}) \circ \tilde \rho (m_p \otimes \bw)
  =(\iota \otimes \beta_{n-2-k}) \left(\iota(\bw) \otimes m_{n-p}\right)
  =\bw \otimes (m_n - m_{k+1-p}),
  \\
  &\tilde \rho \circ (\beta_k \otimes \iota)(m_p \otimes \bw)
  = \tilde \rho \left((m_n-m_{p-1-k}) \otimes \iota(\bw)\right)
  =\bw \otimes (m_n-m_{-p+1+k}),
\end{align*}
for $p=1,\ldots,n$ and $k=0,1,\ldots,n-2-k$.
\end{proof}

As with $\phi(n,k)$ in Section \ref{sec:phi-formula}, 
we use the same notation $\Phi(n,k)$ to denote 
the map on $\Z^{n-1} \otimes \Z^{n-1}$
which is the building block of the dynamics $\Phi(n,k)$. We write
$\Phi(n,k) : \bZ_i^t \otimes \bY_i^t \mapsto \bY_i^{t+1} \otimes \bZ_{i+1}^t$,
with a diagram,

\begin{figure}[H]
\unitlength=1.5mm
\begin{picture}(20,20)(0,0)
\put(5,10){\vector(1,0){10}}
\put(10,15){\vector(0,-1){10}}

\put(1,9){$\bZ_i^t$}
\put(16,9){$\bZ_{i+1}^t ~~.$}
\put(9,17){$\bY_i^t$}
\put(9,2){$\bY_i^{t+1}$}
\end{picture}
\end{figure}


\begin{prop}\label{prop:phi-phi-R}
If it holds that
\begin{align}\label{eq:phiR-nk}
  &\tilde R_{1 \infty}\circ (\beta_k \otimes \iota) (m_p \otimes \bw)
  = (\iota \otimes \beta_k) \circ \Phi(n,k) (m_p \otimes \bw)
\end{align}
for some $k \in \{0,1,\ldots,n-2\}$ and 
some $m_p \otimes \bw \in M \otimes \tilde B_\infty$,
then 
\begin{align}
  \tilde R_{1 \infty}\circ (\beta_{n-2-k} \otimes \iota)
  (m_{n-p'} \otimes \iota(\bw'))
  = (\iota \otimes \beta_{n-2-k}) \circ \Phi(n,n-2-k)
    (m_{n-p'} \otimes \iota(\bw')),
\end{align}
where $(\bw' \otimes m_p') := \Phi(n,k)(m_p \otimes \bw)$.
\end{prop}

\begin{proof}
Eq. \eqref{eq:phiR-nk} can be rewritten as
\begin{align*}
  \tilde R_{1 \infty}\circ (\beta_k \otimes \iota) \Phi(n,k)^{-1}
  (\bw' \otimes m_p')
  = (\iota \otimes \beta_k) (\bw' \otimes m_p').
\end{align*}
By tropicalizing \eqref{eq:phi-formula}, it holds that
\begin{align}
  \Phi(n,n-2-k) = \tilde \rho \circ \Phi^{-1}(n,k) \circ \tilde \rho
\end{align}
on $\Z^{n-1} \otimes \Z^{n-1}$.
By using this and Lemma \ref{lem:rho-R-beta},
the l.h.s. of the first equation becomes 
\begin{align*}
  &\tilde \rho \circ \tilde R_{\infty 1} \circ \tilde \rho \circ 
     (\beta_k \otimes \iota) \circ 
     \tilde \rho \circ \Phi(n,n-2-k) \circ \tilde \rho (\bw' \otimes m_p')
  \\
  & \quad 
  = \tilde \rho \circ \tilde R_{\infty 1} \circ (\iota \otimes \beta_{n-2-k})
    \Phi(n,n-2-k) ( m_{n-p'} \otimes \iota(\bw')).
\end{align*}
On the other hand, the r.h.s. becomes 
$
  \tilde \rho \circ (\beta_{n-2-k} \otimes \iota) 
  ( m_{n-p'} \otimes \iota(\bw')).
$
Thus the claim follows.
\end{proof}

To prove Conjectures \ref{conj:Z-fin} and \ref{thm:BBS}
we have to check that all configurations which appear
in propagating positive solitons have the form  
$m_p \otimes \bw \mapsto \bw' \otimes m_{p'}$
and satisfy \eqref{eq:phiR-nk}.
Proposition \ref{prop:phi-phi-R} means that 
the claims in the conjectures for $\Phi(n,n-2-k)$ follow from those for 
$\Phi(n,k)$. In the next two subsections, we prove these conjectures 
in the cases of $n=3$ and $4$.

\subsection{Proof for $\Phi(3,k)$}\label{subsec:n=3}

We prove the case of $\Phi(3,1)$,
from which the case of $\Phi(3,0)$ follows
due to Proposition \ref{prop:phi-phi-R}. 

We say a finite sequence of states is {\it stable} when
the carrier returns to its initial state after passing through the sequence. For example, in the case of $\Phi(3,1)$ 
the vacuum state $(0,0)$ is stable,
and a sequence $(3,1)(2,0)$ is stable
but $(3.1)(2,1)$ is not, as shown by diagrams:
\begin{figure}[H]
\unitlength=1.0mm
\begin{picture}(110,25)(0,0)

\multiput(10,10)(20,0){2}{\vector(1,0){10}}
\multiput(15,15)(20,0){2}{\vector(0,-1){10}}

\put(0.5,9){$(1,0)$}
\put(20.5,9){$(0,1)$}
\put(40.5,9){$(1,0)$}

\put(10.5,17){$(3,1)$}
\put(30.5,17){$(2,0)$}

\put(10.5,1){$(2,1)$}
\put(30.5,1){$(3,0)$}


\multiput(75,10)(20,0){2}{\vector(1,0){10}}
\multiput(80,15)(20,0){2}{\vector(0,-1){10}}

\put(65.5,9){$(1,0)$}
\put(85.5,9){$(0,1)$}
\put(105.5,9){$(1,1)$}

\put(75.5,17){$(3,1)$}
\put(95.5,17){$(2,1)$}

\put(75.5,1){$(2,1)$}
\put(95.5,1){$(3,1)$}

\end{picture}
\end{figure}

\begin{lem}\label{lem:n=3-seq}
The following sequences of states are stable:
\begin{enumerate}
 \item[(a)] $(i,1)(0,0); ~ i > 0$,
 \item[(b)] $(i,1)(j,0); ~ i ,j > 0$,
 \item[(c)] $(0,1)(i,0); ~ i > 0$,
 \item[(d)] $(i,1)(j',1)(k,0); ~ i, k > 0, ~j' \geq 0$,
\end{enumerate} 
Assume that the initial configuration of $(\bY_i^{t=0})_i$ for $\Phi(3,1)$
consists of the vacuum state $(0,0)$ and a finite number of the 
above sequences.
Then the configuration for $t =1$ again consist of 
the vacuum state and these sequences. 
\end{lem}

\begin{proof}
From \eqref{eq:tropevol-3} we see that all sequences (a)--(d) are stable 
by the diagrams in Figure \ref{fig:c-stable-n=3}.
Note that the vacuum state $(0,0)$ is also stable:
$$
  (1,0) \otimes (0,0) \mapsto (0,0) \otimes (1,0).
$$ 
When the configuration at $t=0$ is given by a composition of 
these stable sequences,
the configuration at $t=1$ is obtained by simply combining 
the diagrams in Figure \ref{fig:c-stable-n=3}, due to the stability. 
It turns out that (a) and (b) change to the form of (b) or (c),
(c) changes to the form of $(0,0)$(a)$^\ast$, and
(d) changes to the form of (b)(a)$^\ast$ or (c)(a)$^\ast$.
Here we define (a)$^\ast:= (i,1)$ for $i > 0$. 
Thus, at $t=1$ the sequence immediately to the right of (a)$^\ast$ is 
always $(0,0)$, (b) or (c), but not (a).
If it is $(0,0)$, we obtain the form of (a).
If it is (b) or (c), we obtain the form of (d).
Since the number of (a)--(d) at $t=0$ is finite,
all (a)$^\ast$ which appear at $t=1$ turn out to be a part of 
a new (a) or (d). Then the claim follows.  
\end{proof}

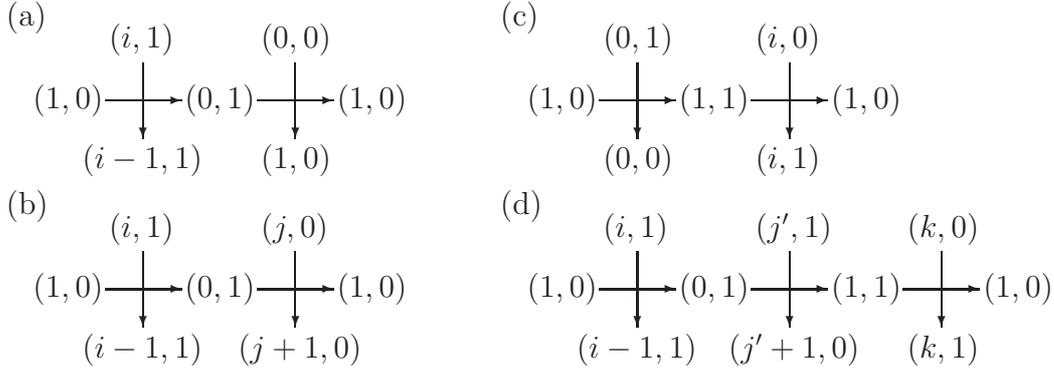
\begin{figure}[ht]
\unitlength=1.0mm
\begin{picture}(130,50)(5,20)

\put(-3,70){(a)}

\multiput(10,60)(20,0){2}{\vector(1,0){10}}
\multiput(15,65)(20,0){2}{\vector(0,-1){10}}

\put(0.5,59){$(1,0)$}
\put(20.5,59){$(0,1)$}
\put(40.5,59){$(1,0)$}

\put(10.5,67){$(i,1)$}
\put(30.5,67){$(0,0)$}

\put(7,51){$(i-1,1)$}
\put(30.5,51){$(1,0)$}


\put(-3,45){(b)}

\multiput(10,35)(20,0){2}{\vector(1,0){10}}
\multiput(15,40)(20,0){2}{\vector(0,-1){10}}

\put(0.5,34){$(1,0)$}
\put(20.5,34){$(0,1)$}
\put(40.5,34){$(1,0)$}

\put(10.5,42){$(i,1)$}
\put(30.5,42){$(j,0)$}

\put(7,26){$(i-1,1)$}
\put(27.5,26){$(j+1,0)$}


\put(62,70){(c)}

\multiput(75,60)(20,0){2}{\vector(1,0){10}}
\multiput(80,65)(20,0){2}{\vector(0,-1){10}}

\put(65.5,59){$(1,0)$}
\put(85.5,59){$(1,1)$}
\put(105.5,59){$(1,0)$}

\put(75.5,67){$(0,1)$}
\put(95.5,67){$(i,0)$}

\put(75.5,51){$(0,0)$}
\put(95.5,51){$(i,1)$}


\put(62,45){(d)}

\multiput(75,35)(20,0){3}{\vector(1,0){10}}
\multiput(80,40)(20,0){3}{\vector(0,-1){10}}

\put(65.5,34){$(1,0)$}
\put(85.5,34){$(0,1)$}
\put(105.5,34){$(1,1)$}
\put(125.5,34){$(1,0)$}

\put(75.5,42){$(i,1)$}
\put(95.5,42){$(j',1)$}
\put(115.5,42){$(k,0)$}

\put(72,26){$(i-1,1)$}
\put(92,26){$(j'+1,0)$}
\put(115.5,26){$(k,1)$}











\end{picture}
\caption{The stable sequences for $\Phi(3,1)$~ ($i,j,k > 0,~ j' \geq 0$).}
\label{fig:c-stable-n=3}
\end{figure}

The following is easily seen from Figure \ref{fig:ex-n=3}.

\begin{lem}\label{lem:phi(31)-2}
A soliton corresponds to a sequence of the form (a)--(c).
The soliton at its minimal length has the form (a), i.e. it equals $(k,1)$ for some $k>0$, 
The velocity of the soliton is $1/(k+1)$. 
\end{lem}

\begin{figure}[ht]
\unitlength=1.0mm
\begin{picture}(70,90)(0,-10)

\multiput(5,70)(20,0){3}{\vector(1,0){10}}
\multiput(10,75)(20,0){3}{\vector(0,-1){10}}

\multiput(5,55)(20,0){3}{\vector(1,0){10}}
\multiput(10,60)(20,0){3}{\vector(0,-1){10}}

\multiput(5,40)(20,0){3}{\vector(1,0){10}}
\multiput(10,45)(20,0){3}{\vector(0,-1){10}}

\multiput(5,25)(20,0){3}{\vector(1,0){10}}
\multiput(10,30)(20,0){3}{\vector(0,-1){10}}

\multiput(5,10)(20,0){3}{\vector(1,0){10}}
\multiput(10,15)(20,0){3}{\vector(0,-1){10}}

\multiput(5,-5)(20,0){3}{\vector(1,0){10}}
\multiput(10,0)(20,0){3}{\vector(0,-1){10}}

\put(-4,69){$(1,0)$}
\put(16,69){$(0,1)$}
\put(36,69){$(1,0)$}
\put(56,69){$(1,0)$}

\put(-4,54){$(1,0)$}
\put(16,54){$(0,1)$}
\put(36,54){$(1,0)$}
\put(56,54){$(1,0)$}

\put(-4,39){$(1,0)$}
\put(16,39){$(0,1)$}
\put(36,39){$(1,0)$}
\put(56,39){$(1,0)$}

\put(-4,24){$(1,0)$}
\put(16,24){$(0,1)$}
\put(36,24){$(1,0)$}
\put(56,24){$(1,0)$}

\put(-4,9){$(1,0)$}
\put(16,9){$(1,1)$}
\put(36,9){$(1,0)$}
\put(56,9){$(1,0)$}

\put(-4,-6){$(1,0)$}
\put(16,-6){$(1,0)$}
\put(36,-6){$(0,1)$}
\put(56,-6){$(1,0)$}


\put(6,77){$(k,1)$}
\put(3,62){$(k-1,1)$}
\put(9.5,46){$\vdots$}
\put(6,32){$(1,1)$}
\put(6,17){$(0,1)$}
\put(6,2){$(0,0)$}
\put(9.5,-14){$\vdots$}

\put(26,77){$(0,0)$}
\put(26,62){$(1,0)$}
\put(29.5,46){$\vdots$}
\put(23,32){$(k-1,0)$}
\put(26,17){$(k,0)$}
\put(26,2){$(k,1)$}
\put(29.5,-14){$\vdots$}

\put(46,77){$(0,0)$}
\put(46,62){$(0,0)$}
\put(49.5,46){$\vdots$}
\put(46,32){$(0,0)$}
\put(46,17){$(0,0)$}
\put(46,2){$(0,0)$}
\put(49.5,-14){$\vdots$}

\end{picture}
\caption{Propagation of a soliton.}
\label{fig:ex-n=3}
\end{figure}
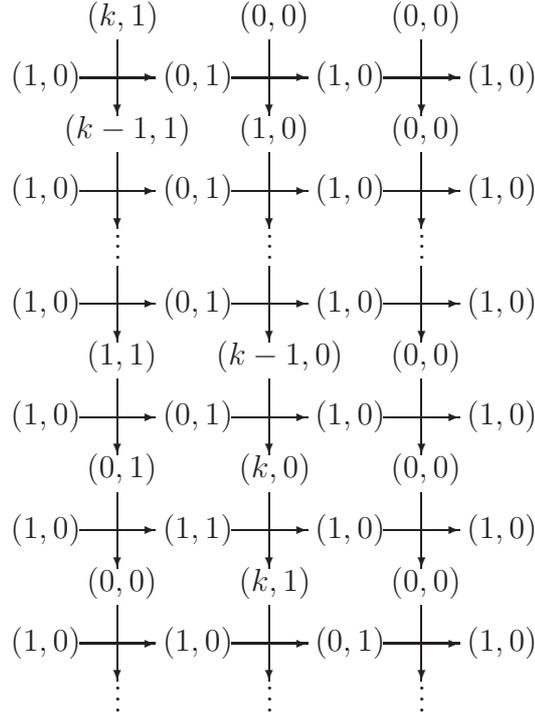

\begin{proof}[Proof of Conjecture \ref{conj:Z-fin} for $\Phi(3,1)$.]
Assume that we start with an initial configuration including $N>1$ 
sequences of forms (a)--(c), which is an $N$-soliton state.
Then a faster soliton catches up with a slower one,
and overtakes it after some scattering states (d).
Finally, the $N$ solitons line up in a way that slower ones are left and 
faster ones are right.
Thus, from Figure \ref{fig:c-stable-n=3}, 
the possible configurations which appear in propagating 
positive solitons are as Figure \ref{fig:conf-n=3a},
where we have only $(1,0)$, $(1,1)$ and $(0,1)$ for the carriers.

\end{proof}

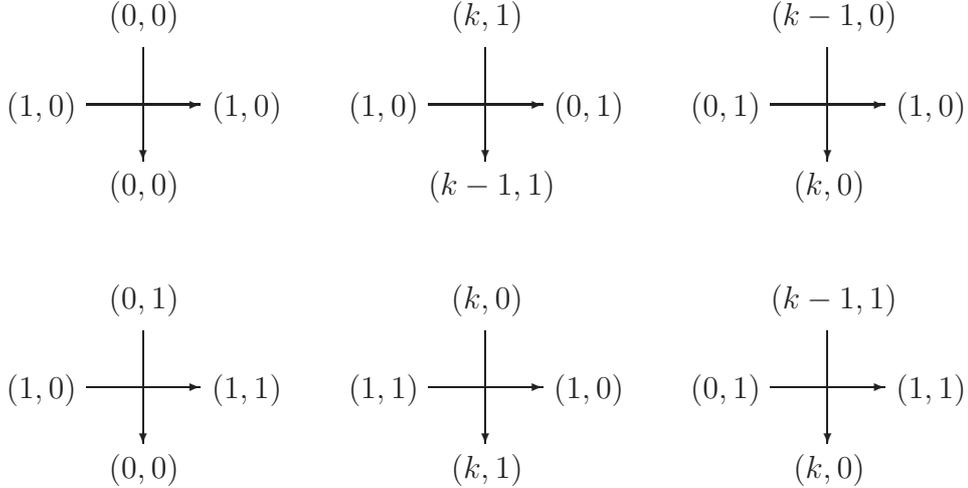
\begin{figure}[ht]
\unitlength=1.5mm
\begin{picture}(90,50)(0,0)
\put(5,35){\vector(1,0){10}}
\put(10,40){\vector(0,-1){10}}
\put(-2,34){$(1,0)$}
\put(16,34){$(1,0)$}
\put(7,42){$(0,0)$}
\put(7,27){$(0,0)$}

\put(35,35){\vector(1,0){10}}
\put(40,40){\vector(0,-1){10}}
\put(28,34){$(1,0)$}
\put(46,34){$(0,1)$}
\put(37,42){$(k,1)$}
\put(35,27){$(k-1,1)$}

\put(65,35){\vector(1,0){10}}
\put(70,40){\vector(0,-1){10}}
\put(58,34){$(0,1)$}
\put(76,34){$(1,0)$}
\put(65,42){$(k-1,0)$}
\put(67,27){$(k,0)$}

\put(5,10){\vector(1,0){10}}
\put(10,15){\vector(0,-1){10}}
\put(-2,9){$(1,0)$}
\put(16,9){$(1,1)$}
\put(7,17){$(0,1)$}
\put(7,2){$(0,0)$}

\put(35,10){\vector(1,0){10}}
\put(40,15){\vector(0,-1){10}}
\put(28,9){$(1,1)$}
\put(46,9){$(1,0)$}
\put(37,17){$(k,0)$}
\put(37,2){$(k,1)$}

\put(65,10){\vector(1,0){10}}
\put(70,15){\vector(0,-1){10}}
\put(58,9){$(0,1)$}
\put(76,9){$(1,1)$}
\put(65,17){$(k-1,1)$}
\put(67,2){$(k,0)$}

\end{picture}
\caption{Possible configurations to propagate solitons in $\Phi(3,1)$
($k \in \Z_{\geq 1}$).}
\label{fig:conf-n=3a}
\end{figure}

\begin{proof}[Proof of Conjecture \ref{thm:BBS} for $\Phi(3,1)$]
In Figure \ref{fig:ex-n=3}, 
we see that during rightward propagation of a soliton $(k,1)$,
a sequence of carriers $(\bZ_1^t)_{t=0,1,\ldots,k} 
=(\underbrace{(0,1),\ldots,(0,1)}_{k},(1,1))$ propagates downward 
with velocity $k+1$. 
By the map $\beta_1$, this sequence is transformed into
$(\underbrace{(0,1),\ldots,(0,1)}_{k},(1,0))$ which corresponds to 
a soliton $\underbrace{33 \cdots 3}_{k}2$ of the $\mathfrak{sl}_3$-BBS.

The configurations in Figure \ref{fig:conf-n=3a} are the local diagrams
appearing in Figure \ref{fig:c-stable-n=3}, which are nothing but those 
that appear in soliton propagations.
We transform them into the diagrams in Figure \ref{fig:n=3-BBS},
by acting by $\beta_1$ on carriers and $\iota$ on states.
Using \eqref{eq:sln-comb-R},
one sees that \eqref{eq:phiR-nk} is fulfilled by $m_p \otimes \bw
\in M \otimes \tilde B_\infty$ appearing in Figure \ref{fig:conf-n=3a}.
It turns out  
that all configurations in Figure \ref{fig:n=3-BBS} are what appear
when $R_{1 \infty}$ propagates solitons  
of the form $\underbrace{33 \cdots 3}_{k > 0}2$, due to the following facts:
in the states $(x_2,x_3)$ on the vertical edges in Figure \ref{fig:n=3-BBS},
$x_2$ takes only $0$ or $1$ which means that 
each soliton includes at most one $2$-ball.
There is neither configuration 
$\tilde R_{1 \infty}((0,0) \otimes (0,1)) = (0,0) \otimes (0,1)$
nor $\tilde R_{1 \infty}((1,0) \otimes (0,0)) = (1,0) \otimes (0,0)$.
This means that there is neither a soliton containing only $3$-balls,
nor a soliton containing only $2$-balls. 
   
\end{proof}

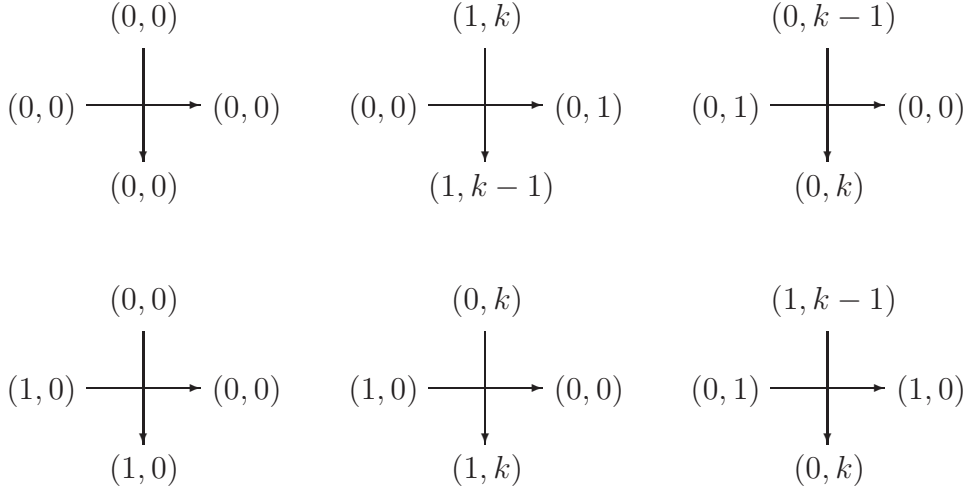
\begin{figure}[H]
\unitlength=1.5mm
\begin{picture}(90,50)(0,0)
\put(5,35){\vector(1,0){10}}
\put(10,40){\vector(0,-1){10}}
\put(-2,34){$(0,0)$}
\put(16,34){$(0,0)$}
\put(7,42){$(0,0)$}
\put(7,27){$(0,0)$}

\put(35,35){\vector(1,0){10}}
\put(40,40){\vector(0,-1){10}}
\put(28,34){$(0,0)$}
\put(46,34){$(0,1)$}
\put(37,42){$(1,k)$}
\put(35,27){$(1,k-1)$}

\put(65,35){\vector(1,0){10}}
\put(70,40){\vector(0,-1){10}}
\put(58,34){$(0,1)$}
\put(76,34){$(0,0)$}
\put(65,42){$(0,k-1)$}
\put(67,27){$(0,k)$}

\put(5,10){\vector(1,0){10}}
\put(10,15){\vector(0,-1){10}}
\put(-2,9){$(1,0)$}
\put(16,9){$(0,0)$}
\put(7,17){$(0,0)$}
\put(7,2){$(1,0)$}

\put(35,10){\vector(1,0){10}}
\put(40,15){\vector(0,-1){10}}
\put(28,9){$(1,0)$}
\put(46,9){$(0,0)$}
\put(37,17){$(0,k)$}
\put(37,2){$(1,k)$}

\put(65,10){\vector(1,0){10}}
\put(70,15){\vector(0,-1){10}}
\put(58,9){$(0,1)$}
\put(76,9){$(1,0)$}
\put(65,17){$(1,k-1)$}
\put(67,2){$(0,k)$}

\end{picture}
\caption{The $\mathfrak{sl}_3$-BBS configurations from $\Phi(3,1)$ 
($k \in \Z_{\geq 1}$).}
\label{fig:n=3-BBS} 
\end{figure}

\subsection{Proof for $\Phi(4,k)$}\label{subsec:n=4}

First we consider $\Phi(4,2)$.
The map $\phi(4,2)$ is given by
\begin{align}\label{eq:phi(42)}
\begin{split}
&(z_{i,1}^t,z_{i,2}^t,z_{i,3}^t) \otimes (y_{i,1}^t,y_{i,2}^t,y_{i,3}^t)
\\
&\qquad \mapsto 
(y_{i,1}^{t+1},y_{i,2}^{t+1},y_{i,3}^{t+1}) \otimes (z_{i+1,1}^t,z_{i+1,2}^t,z_{i+1,3}^t)
\\
& \qquad = 
\left(\frac{z_{i,3}^t y_{i,1}^t}{z_{i,2}^t+y_{i,1}^t},
      \frac{(z_{i,2}^t+y_{i,1}^t) y_{i,2}^t}{z_{i,1}^t+y_{i,2}^t},
      z_{i,1}^t+y_{i,2}^t \right) \otimes
\left(\frac{z_{i,1}^t (z_{i,2}^t+y_{i,1}^t)}{z_{i,1}^t+y_{i,2}^t},
      \frac{z_{i,2}^t z_{i,3}^t}{z_{i,2}^t+y_{i,1}^t},y_{i,3}^t \right).
\end{split}
\end{align}

\begin{lem}\label{lem:n=4-seq}
The following sequences of states are stable:
\begin{enumerate}
 \item[(a1)] $(i,j,1)(0,0,0); ~ i,j > 0$,
 \item[(a2)] $(i_1,j,1)(i_2,0,0); ~ i_1,i_2,j > 0$,
 \item[(b1)] $(0,j,1)(i,0,0); ~ i ,j > 0$,
 \item[(b2)] $(0,j,1)(i,j',0); ~ i ,j_1,j_2 > 0$,
 \item[(c)] $(0,0,1)(i,j,0); ~ i,j > 0$,
 \item[(d)] $(0,j_1,1)(i_1,j',1)(i_2,j_2,0); ~ i_1,i_2,j_1,j_2 > 0, 
            ~j' \geq 0$,
 \item[(e1)] $(i_1,j_1,1)(i',j_2,1)(i_2,j',0); ~ i_1,i_2,j_1,j_2 > 0, 
             ~i',j' \geq 0$,
 \item[(e2)] $(i_1,j_1,1)(i',0,1)(i_2,j_2,0); ~ i_1,i_2,j_1,j_2 > 0, 
             ~i' \geq 0$, 
 \item[(f)] $(i_1,j_1,1)(i',j_2,1)(i_2,j',1)(i_3,j_3,0); ~ i_1,i_2,i_3,j_1,j_2,j_3 > 0, ~i',j' \geq 0$.
\end{enumerate} 
Assume that the initial configuration of $(\bY_i^{t=0})_i$ for $\Phi(4,2)$
consists of 
the vacuum state $(0,0,0)$ and a finite number of these sequences of states.
Then the configuration for $t =1$ again consist of 
the vacuum state and these sequences. 
\end{lem}

\begin{proof}
It is shown in the same way as in the case of $\Phi(3,1)$:
due to the map $\Phi(4,2)$ given by the tropicalization of \eqref{eq:phi(42)},
we obtain diagrams in Figure \ref{fig:c-stable-n=4},
and see that all sequences (a1)--(f) and the vacuum state $(0,0,0)$ are 
stable. 
Then one sees that 
(a1) changes to (a2), 
(a2) changes to (a2) or (b1), 
(b1) changes to (b2) or (c), 
(c) changes to $(0,0,0)$(a)$^\ast$, 
(d) changes to (b2)(a)$^\ast$ or (c)(a)$^\ast$, 
(e1) changes to (e1) or (d) or (e2),
(e2) changes to (a2)(a)$^\ast$, and 
(f) changes to (e1)(a)$^\ast$ or (d)(a)$^\ast$.
Moreover, (a)$^\ast(0,0,0)$ has a form of (a1),
and (a)$^\ast$(x) has a form of (e1) or (e2) or (f)
where (x) is one of (a2)--(e2).
Then the claim follows.
\end{proof}

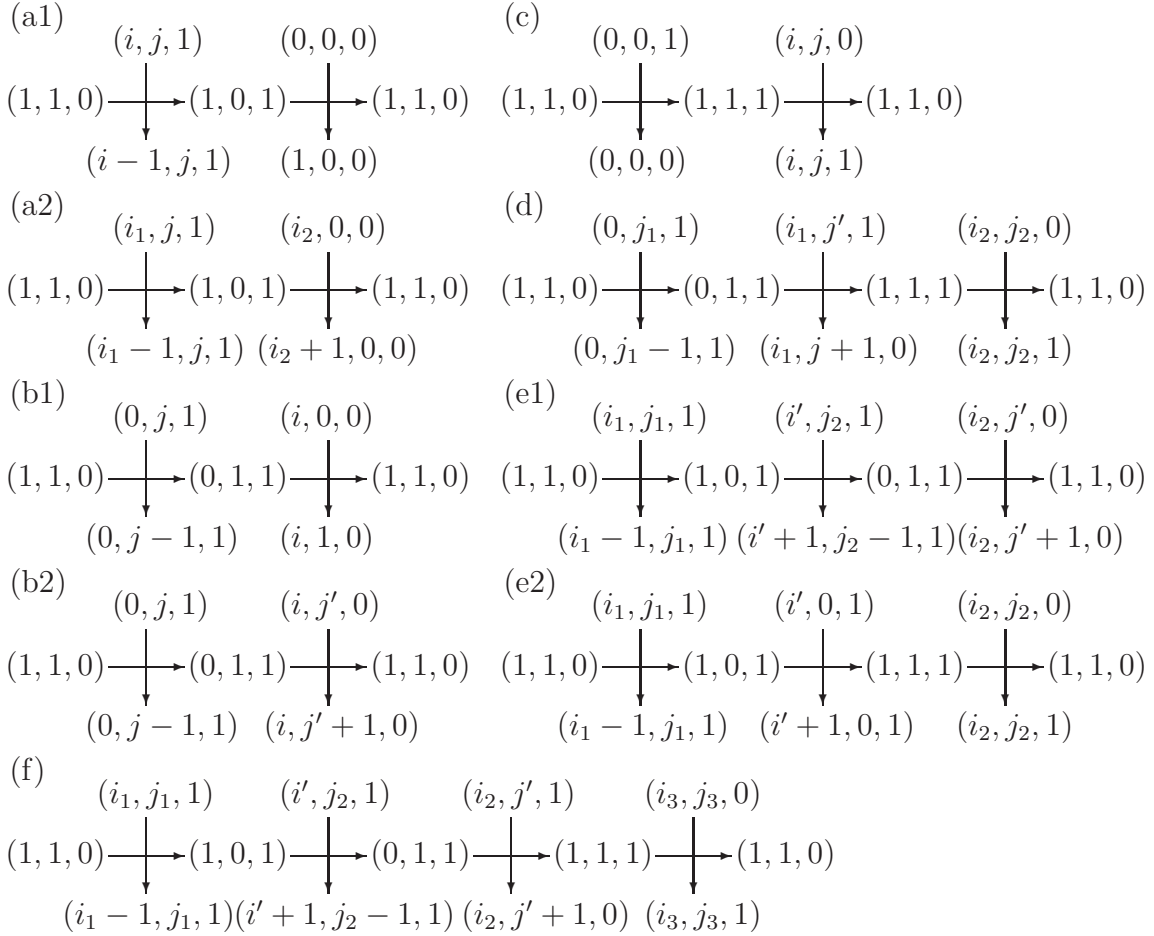
\begin{figure}[ht]
\unitlength=1.0mm
\begin{picture}(130,125)(5,-50)

\put(-3,70){(a1)}

\multiput(10,60)(24,0){2}{\vector(1,0){10}}
\multiput(15,65)(24,0){2}{\vector(0,-1){10}}

\put(-3.5,59){$(1,1,0)$}
\put(20.5,59){$(1,0,1)$}
\put(44.5,59){$(1,1,0)$}

\put(10.5,67){$(i,j,1)$}
\put(32.5,67){$(0,0,0)$}

\put(7,51){$(i-1,j,1)$}
\put(32.5,51){$(1,0,0)$}


\put(-3,45){(a2)}

\multiput(10,35)(24,0){2}{\vector(1,0){10}}
\multiput(15,40)(24,0){2}{\vector(0,-1){10}}

\put(-3.5,34){$(1,1,0)$}
\put(20.5,34){$(1,0,1)$}
\put(44.5,34){$(1,1,0)$}

\put(10.5,42){$(i_1,j,1)$}
\put(32.5,42){$(i_2,0,0)$}

\put(7,26){$(i_1-1,j,1)$}
\put(29.5,26){$(i_2+1,0,0)$}


\put(-3,20){(b1)}

\multiput(10,10)(24,0){2}{\vector(1,0){10}}
\multiput(15,15)(24,0){2}{\vector(0,-1){10}}

\put(-3.5,9){$(1,1,0)$}
\put(20.5,9){$(0,1,1)$}
\put(44.5,9){$(1,1,0)$}

\put(10.5,17){$(0,j,1)$}
\put(32.5,17){$(i,0,0)$}

\put(7,1){$(0,j-1,1)$}
\put(32.5,1){$(i,1,0)$}


\put(-3,-5){(b2)}

\multiput(10,-15)(24,0){2}{\vector(1,0){10}}
\multiput(15,-10)(24,0){2}{\vector(0,-1){10}}

\put(-3.5,-16){$(1,1,0)$}
\put(20.5,-16){$(0,1,1)$}
\put(44.5,-16){$(1,1,0)$}

\put(10.5,-8){$(0,j,1)$}
\put(32.5,-8){$(i,j',0)$}

\put(7,-24){$(0,j-1,1)$}
\put(30.5,-24){$(i,j'+1,0)$}


\put(62,70){(c)}

\multiput(75,60)(24,0){2}{\vector(1,0){10}}
\multiput(80,65)(24,0){2}{\vector(0,-1){10}}

\put(61.5,59){$(1,1,0)$}
\put(85.5,59){$(1,1,1)$}
\put(109.5,59){$(1,1,0)$}

\put(73.5,67){$(0,0,1)$}
\put(97.5,67){$(i,j,0)$}

\put(73,51){$(0,0,0)$}
\put(97.5,51){$(i,j,1)$}


\put(62,45){(d)}

\multiput(75,35)(24,0){3}{\vector(1,0){10}}
\multiput(80,40)(24,0){3}{\vector(0,-1){10}}

\put(61.5,34){$(1,1,0)$}
\put(85.5,34){$(0,1,1)$}
\put(109.5,34){$(1,1,1)$}
\put(133.5,34){$(1,1,0)$}

\put(73.5,42){$(0,j_1,1)$}
\put(97.5,42){$(i_1,j',1)$}
\put(121.5,42){$(i_2,j_2,0)$}

\put(71,26){$(0,j_1-1,1)$}
\put(95.5,26){$(i_1,j+1,0)$}
\put(121.5,26){$(i_2,j_2,1)$}


\put(62,20){(e1)}

\multiput(75,10)(24,0){3}{\vector(1,0){10}}
\multiput(80,15)(24,0){3}{\vector(0,-1){10}}

\put(61.5,9){$(1,1,0)$}
\put(85.5,9){$(1,0,1)$}
\put(109.5,9){$(0,1,1)$}
\put(133.5,9){$(1,1,0)$}

\put(73.5,17){$(i_1,j_1,1)$}
\put(97.5,17){$(i',j_2,1)$}
\put(121.5,17){$(i_2,j',0)$}

\put(69,1){$(i_1-1,j_1,1)$}
\put(92.5,1){$(i'+1,j_2-1,1)$}
\put(121.5,1){$(i_2,j'+1,0)$}


\put(62,-5){(e2)}

\multiput(75,-15)(24,0){3}{\vector(1,0){10}}
\multiput(80,-10)(24,0){3}{\vector(0,-1){10}}

\put(61.5,-16){$(1,1,0)$}
\put(85.5,-16){$(1,0,1)$}
\put(109.5,-16){$(1,1,1)$}
\put(133.5,-16){$(1,1,0)$}

\put(73.5,-8){$(i_1,j_1,1)$}
\put(97.5,-8){$(i',0,1)$}
\put(121.5,-8){$(i_2,j_2,0)$}

\put(69,-24){$(i_1-1,j_1,1)$}
\put(95.5,-24){$(i'+1,0,1)$}
\put(121.5,-24){$(i_2,j_2,1)$}


\put(-3,-30){(f)}

\multiput(10,-40)(24,0){4}{\vector(1,0){10}}
\multiput(15,-35)(24,0){4}{\vector(0,-1){10}}

\put(-3.5,-41){$(1,1,0)$}
\put(20.5,-41){$(1,0,1)$}
\put(44.5,-41){$(0,1,1)$}
\put(68.5,-41){$(1,1,1)$}
\put(92.5,-41){$(1,1,0)$}

\put(8.5,-33){$(i_1,j_1,1)$}
\put(32.5,-33){$(i',j_2,1)$}
\put(56.5,-33){$(i_2,j',1)$}
\put(80.5,-33){$(i_3,j_3,0)$}

\put(4,-49){$(i_1-1,j_1,1)$}
\put(26.5,-49){$(i'+1,j_2-1,1)$}
\put(56.5,-49){$(i_2,j'+1,0)$}
\put(80.5,-49){$(i_3,j_3,1)$}

\end{picture}
\caption{Carrier stable sequences for $\Phi(4,2)$ ~($i,j,i_\ast,j_\ast > 0,~i',j' \geq 0$).}
\label{fig:c-stable-n=4}
\end{figure}

\begin{proof}[Proof of Conjecture \ref{conj:Z-fin} and \ref{thm:BBS} for $\Phi(4,2)$]
We give an outline of the proof.
It is easy to show that one-soliton propagation is given by combining
diagrams (a1)--(c) in Figure~\ref{fig:c-stable-n=4}, 
thus a soliton of the minimal form $(i,j,1)$ has velocity $1/(i+j+1)$. 
Further, from Lemma \ref{lem:n=4-seq} it follows that the propagation of 
any multi-soliton state is described by diagrams (a1)--(f) and 
\begin{align}\label{eq:n=(42)vacuum}
(1,1,0) \otimes (0,0,0) \mapsto (0,0,0) \otimes (1,1,0),
\end{align}
so Conjecture  \ref{conj:Z-fin} follows.
The diagrams in Figure \ref{fig:c-stable-n=4} 
consist of `local' diagrams as in Figure \ref{fig:conf-n=4a}
and \eqref{eq:n=(42)vacuum}, and they are shown to satisfy \eqref{eq:phiR-nk}.
Then we see that in the corresponding $\mathfrak{sl}_4$-BBS 
only solitons of forms as $\underbrace{4\cdots4}_{i} \underbrace{3\cdots3}_{j}2$ for some $i,j>0$ appear, which proves Conjecture \ref{thm:BBS}.
\end{proof}

\begin{figure}[ht]
\unitlength=1.5mm
\begin{picture}(90,75)(0,0)

\put(5,60){\vector(1,0){10}}
\put(10,65){\vector(0,-1){10}}
\put(-4,59){$(1,1,0)$}
\put(15,59){$(1,0,1)$}
\put(6,67){$(i,j,1)$}
\put(2,52){$(i-1,j-1,1)$}

\put(35,60){\vector(1,0){10}}
\put(40,65){\vector(0,-1){10}}
\put(26,59){$(1,1,0)$}
\put(45,59){$(0,1,1)$}
\put(36,67){$(0,j,1)$}
\put(34,52){$(0,j-1,1)$}

\put(65,60){\vector(1,0){10}}
\put(70,65){\vector(0,-1){10}}
\put(56,59){$(1,1,0)$}
\put(75,59){$(1,1,1)$}
\put(66,67){$(0,0,1)$}
\put(66,52){$(0,0,0)$}


\put(5,35){\vector(1,0){10}}
\put(10,40){\vector(0,-1){10}}
\put(-4,34){$(1,0,1)$}
\put(15,34){$(1,1,0)$}
\put(6,42){$(i,0,0)$}
\put(4,27){$(i+1,0,0)$}

\put(35,35){\vector(1,0){10}}
\put(40,40){\vector(0,-1){10}}
\put(26,34){$(1,0,1)$}
\put(45,34){$(0,1,1)$}
\put(37,42){$(i,j,1)$}
\put(32,27){$(i+1,j-1,1)$}

\put(65,35){\vector(1,0){10}}
\put(70,40){\vector(0,-1){10}}
\put(56,34){$(1,0,1)$}
\put(75,34){$(1,1,1)$}
\put(67,42){$(i,0,1)$}
\put(64,27){$(i+1,0,0)$}


\put(5,10){\vector(1,0){10}}
\put(10,15){\vector(0,-1){10}}
\put(-4,9){$(0,1,1)$}
\put(15,9){$(1,1,0)$}
\put(4,17){$(i,j-1,0)$}
\put(6,2){$(i,j,0)$}

\put(35,10){\vector(1,0){10}}
\put(40,15){\vector(0,-1){10}}
\put(26,9){$(0,1,1)$}
\put(45,9){$(1,1,1)$}
\put(34,17){$(i,j-1,1)$}
\put(36,2){$(i,j,1)$}

\put(65,10){\vector(1,0){10}}
\put(70,15){\vector(0,-1){10}}
\put(56,9){$(1,1,1)$}
\put(75,9){$(1,1,0)$}
\put(66,17){$(i,j,0)$}
\put(66,2){$(i,j,1)$}

\end{picture}
\caption{Possible configurations to propagate solitons for $\Phi(4,2)$
($i,j \in \Z_{\geq 1}$).}
\label{fig:conf-n=4a}
\end{figure}
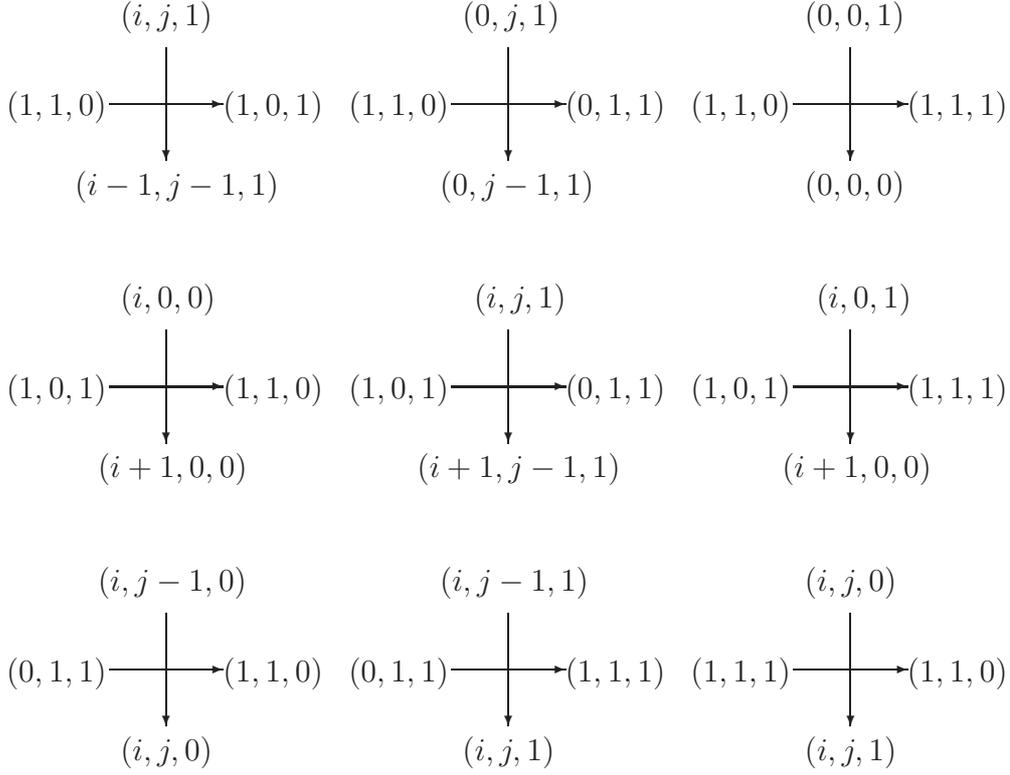

%

Conjectures \ref{conj:Z-fin} and \ref{thm:BBS} for $\Phi(4,0)$
follows from the above result on $\Phi(4,2)$ and 
Proposition \ref{prop:phi-phi-R}, where  
the corresponding $\mathfrak{sl}_4$-BBS only includes solitons of forms as 
$4 \underbrace{3\cdots3}_{i} \underbrace{2\cdots2}_{j}$ for some $i,j>0$.

Next we consider $\Phi(4,1)$. 

\begin{lem}\label{lem:n=41-seq}
The following sequences of states are stable:
\begin{enumerate}
 \item[(a1)] $(i,1,j)(0,0,0); ~ i,j > 0$,
 \item[(a2)] $(i_1,1,j)(i_2,0,0); ~ i_1,i_2,j > 0$,
 \item[(b)] $(0,1,j)(i,0,0); ~ i ,j > 0$,
 \item[(c1)] $(0,0,j)(i,1,0); ~ i ,j > 0$,
 \item[(c2)] $(0,0,j_1)(i,1,j_2); ~ i,j_1,j_2 > 0$,
 \item[(d)] $(0,1,j_1)(i_1,0,j_2)(i_2,1,j'); ~ i_1,i_2,j_1,j_2 > 0, ~j' \geq 0$,
 \item[(e1)] $(i_1,1,j_1)(i',1,j_2)(i_2,0,0); ~ i_1,i_2,j_1,j_2 > 0, ~i'\geq 0$,
 \item[(e2)] $(i_1,1,j_1)(i',0,j_2)(i_2,1,j'); ~ i_1,i_2,j_1,j_2 > 0, ~i',j' \geq 0$, 
 \item[(f)] $(i_1,1,j_1)(i',1,j_2)(i_2,0,j_3)(i_3,1,j'); ~ i_1,i_2,i_3,j_1,j_2,j_3 > 0, ~i',j' \geq 0$.
\end{enumerate} 
Assume that the initial configuration of $(\bY_i^{t=0})_i$ for $\Phi(4,1)$
consists of 
the vacuum state $(0,0,0)$ and a finite number of these sequences.
Then the configuration for $t =1$ again consist of 
the vacuum state and these sequences. 
\end{lem}
 
Using this lemma, one can show that 
Conjectures \ref{conj:Z-fin} and \ref{thm:BBS} holds for $\Phi(4,1)$,
in the same manner as the case of $\Phi(4,2)$. 
The corresponding $\mathfrak{sl}_4$-BBS only includes solitons of forms as 
$\underbrace{4\cdots4}_{i}3 \underbrace{2\cdots2}_{j}$ for some $i,j>0$.

\section{Numerical phenomena: negative solitons, relaxation solitons, and pulsars}

Besides positive solitons, we numerically observe negative solitons, relaxation solitons and pulsars for $\Phi(n,k)$ on $\Z$ with 
the commuting pair given by \eqref{eq:Phi-comm}. 
It might be an interesting future problem to study these phenomena.


\subsection{Negative solitons}

For the definition of a negative soliton, see \S \ref{subsec:soliton}.
We observe that a state $(\underbrace{-p,\ldots,-p}_{n-1})$ for $p \in \Z_{>0}$ is a negative soliton for $\Phi(n,k)$,
whose velocity is $1/(n-1)$ independent of $p$.
The difference among the negative solitons appears in scattering
with positive solitons. 
The following examples show that 
in scatterings of positive and negative solitons
the phase shift of the positive soliton depends on $p$.

\begin{example}
Scatterings of positive and negative solitons in the case of $\Phi(3,1)$.
We write $\bar k$ for $-k$, for $k \in \Z_{>0}$.
\\
(i) $(\b1,\b1) \times (3,1) \mapsto (3,1) \times (\b1, \b1)$: 
\begin{align*}
t = 0: & (00)(\b1 \b1)(00)(31)(00)(00)(00)(00)(00)(00)(00)(00)
\\
t = 1: & (00)(0 \b1)(\b1 0)(21)(10)(00)(00)(00)(00)(00)(00)(00)
\\
t = 2: & (00)(00)(\b1 \b1)(11)(20)(00)(00)(00)(00)(00)(00)(00)
\\
t = 3: & (00)(00)(0 \b1)(\b1 1)(30)(00)(00)(00)(00)(00)(00)(00)
\\
t = 4: & (00)(00)(00)(\b1 \b1)(31)(00)(00)(00)(00)(00)(00)(00)
\\
t = 5: & (00)(00)(00)(0 \b1)(02)(2 \b1)(00)(00)(00)(00)(00)(00)
\\
t = 6: & (00)(00)(00)(00)(\b1 0)(40)(\b1 0)(00)(00)(00)(00)(00)
\\
t = 7: & (00)(00)(00)(00)(0 \b1 )(22)(0 \b1)(00)(00)(00)(00)(00)
\\
t = 8: & (00)(00)(00)(00)(00)(01)(3 \b1)(\b1 0)(00)(00)(00)(00)
\\
t = 9: & (00)(00)(00)(00)(00)(00)(31)(\b1 \b1)(00)(00)(00)(00)
\\
t = 10: & (00)(00)(00)(00)(00)(00)(21)(1 \b1)(\b1 0)(00)(00)(00)
\\
t = 11: & (00)(00)(00)(00)(00)(00)(11)(20)(\b1 \b1)(00)(00)(00)
\\
t = 12: & (00)(00)(00)(00)(00)(00)(01)(30)(0 \b1)(\b1 0)(00)(00)
\\
t = 13: & (00)(00)(00)(00)(00)(00)(00)(31)(00)(\b1 \b1)(00)(00)
\end{align*}
(ii) $(\bt,\bt) \times (3,1) \mapsto (3,1) \times (\bt, \bt)$: 
\begin{align*}
t = 0: & (00)(\bt \bt)(00)(31)(00)(00)(00)(00)(00)(00)(00)(00)
\\
t = 1: & (00)(0 \bt)(\bt 0)(21)(10)(00)(00)(00)(00)(00)(00)(00)
\\
t = 2: & (00)(00)(\bt \bt)(11)(20)(00)(00)(00)(00)(00)(00)(00)
\\
t = 3: & (00)(00)(0 \bt)(\bt 1)(30)(00)(00)(00)(00)(00)(00)(00)
\\
t = 4: & (00)(00)(00)(\bt \bt)(31)(00)(00)(00)(00)(00)(00)(00)
\\
t = 5: & (00)(00)(00)(0 \bt)(\bt 3)(3 \bt)(00)(00)(00)(00)(00)(00)
\\
t = 6: & (00)(00)(00)(00)(\bt \bt)(51)(\bt 0)(00)(00)(00)(00)(00)
\\
t = 7: & (00)(00)(00)(00)(0 \bt )(03)(1 \bt)(00)(00)(00)(00)(00)
\\
t = 8: & (00)(00)(00)(00)(00)(\bt 0)(5 \b1)(\bt 0)(00)(00)(00)(00)
\\
t = 9: & (00)(00)(00)(00)(00)(0 \bt)(23)(\b1 \bt)(00)(00)(00)(00)
\\
t = 10: & (00)(00)(00)(00)(00)(00)(\b1 1)(4 \bt)(\bt 0)(00)(00)(00)
\\
t = 11: & (00)(00)(00)(00)(00)(00)(0 \b1)(32)(\bt \bt)(00)(00)(00)
\\
t = 12: & (00)(00)(00)(00)(00)(00)(00)(11)(2 \bt)(\bt 0)(00)(00)
\\
t = 13: & (00)(00)(00)(00)(00)(00)(00)(01)(30)(\bt \bt)(00)(00)
\\
t = 14: & (00)(00)(00)(00)(00)(00)(00)(00)(31)(0 \bt)(\bt 0)(00)
\\
t = 15: & (00)(00)(00)(00)(00)(00)(00)(00)(21)(1 0)(\bt \bt)(00)
\end{align*}
\end{example}

\begin{remark}
Negative solitons for the $\mathfrak{sl}_2$-BBS were found by Hirota \cite{H},
and studied in \cite{KMT10, WNSRG} and others.
In \cite{KMT10}, it is clarified that the states with negative solitons 
are transformed into the $\mathfrak{sl}_2$-BBS with greater box capacity.
It is not clear for now if some similar mechanism works 
in the general $\mathfrak{sl}_n$-BBS or in $\Phi(n,k)$.    
\end{remark}

\subsection{Relaxation solitons and pulsars} \label{subsec:relax}

Besides solitons, we introduce two phenomena, {\it relaxation solitons} 
and {\it pulsars},
which may satisfy the condition (i) for solitons 
presented in \S \ref{subsec:soliton}, but not (ii).
 
We define a {\it relaxation soliton} as 
a finite sequence of non-vacuum states at $t=0$ such that  
the carrier gets back to the initial one for $t > t_0$ 
for some $t_0 \in \Z_{\geq 0}$, 
but not for $0 \leq t \leq t_0$.
In the other words, it is a finite sequence of non-vacuum states 
which reduces to solitons at $t=t_0 + 1 > 0$. 
In the following examples, we have $t_0 = 0$ in (i) and $t_0=1$ in (ii).

\begin{example} 
Relaxation solitons. 
\\
(i) $\Phi(3,1)$:
\begin{align*}
& (\bY_i^t)_i & & (\bZ_i^t)_i
\\[1mm]
t = 0: & (00)(23)(00)(00)(00)(00)(00)(00) & & (10)(10)(03)(30)(30)(30)(30)(30) 
\\
t = 1: & (00)(11)(30)(00)(00)(00)(00)(00) & & (10)(10)(01)(10)(10)(10)(10)(10)
\\
t = 2: & (00)(01)(40)(00)(00)(00)(00)(00) & & (10)(10)(11)(10)(10)(10)(10)(10)
\\
t = 3: & (00)(00)(41)(00)(00)(00)(00)(00) & & (10)(10)(10)(01)(10)(10)(10)(10)
\\
t = 4: & (00)(00)(31)(10)(00)(00)(00)(00) & & (10)(10)(10)(01)(10)(10)(10)(10)
\\
t = 5: & (00)(00)(21)(20)(00)(00)(00)(00) & & (10)(10)(10)(01)(10)(10)(10)(10)
\\
t = 6: & (00)(00)(11)(30)(00)(00)(00)(00) & & (10)(10)(10)(01)(10)(10)(10)(10)
\end{align*}
(ii) $\Phi(4,0)$:
\begin{align*}
& (\bY_i^t)_i & & (\bZ_i^t)_i
\\[1mm]
t = 0: & (000)(320)(000)(000)(000)(000) & & (011)(011)(023)(023)(023)(023) 
\\
t = 1: & (000)(041)(000)(000)(000)(000) & & (011)(011)(110)(010)(010)(010) 
\\
t = 2: & (000)(031)(100)(000)(000)(000) & & (011)(011)(110)(011)(011)(011) 
\\
t = 3: & (000)(021)(110)(000)(000)(000) & & (011)(011)(110)(011)(011)(011) 
\\
t = 4: & (000)(011)(120)(000)(000)(000) & & (011)(011)(110)(011)(011)(011) 
\\
t = 5: & (000)(001)(130)(000)(000)(000) & & (011)(011)(101)(011)(011)(011) 
\\
t = 6: & (000)(000)(131)(000)(000)(000) & & (011)(011)(011)(111)(011)(011)
\\
t = 7: & (000)(000)(031)(100)(000)(000) & & (011)(011)(011)(110)(011)(011)
\end{align*}
\\
\end{example}

We define a {\em pulsar} as a finite sequence of non-vacuum states satisfying 
\begin{itemize} 
\item[(i)] 
the sequence moves to the right with a constant velocity, 

\item[(ii')]
the final carriers $\bZ_i^t$ for $i \gg 1$ are periodic in $t$.
\end{itemize} 
See the following examples.

\begin{example}
Pulsars:
\\
(i) $\Phi(3,1)$:
\begin{align*}
& (\bY_i^t)_i & & (\bZ_i^t)_i
\\[1mm]
t = 0: & (00)(10)(00)(00)(00)(00)(00)(00) & & (10)(10)(00)(00)(00)(00)(00)(00) 
\\
t = 1: & (00)(01)(00)(00)(00)(00)(00)(00) & & (10)(10)(11)(20)(20)(20)(20)(20) 
\\
t = 2: & (00)(00)(10)(00)(00)(00)(00)(00) & & (10)(10)(10)(00)(00)(00)(00)(00) 
\\
t = 3: & (00)(00)(01)(00)(00)(00)(00)(00) & & (10)(10)(10)(11)(20)(20)(20)(20) 
\\
t = 4: & (00)(00)(00)(10)(00)(00)(00)(00) & & (10)(10)(10)(10)(00)(00)(00)(00)
\\
t = 5: & (00)(00)(00)(01)(00)(00)(00)(00) & & (10)(10)(10)(10)(11)(20)(20)(20) 
\end{align*}

(ii) $\Phi(4,0)$:
\begin{align*}
& (\bY_i^t)_i & & (\bZ_i^t)_i
\\[1mm]
t = 0: & (000)(110)(000)(000)(000)(000) & & (011)(011)(021)(021)(021)(021) 
\\
t = 1: & (000)(011)(000)(000)(000)(000) & & (011)(011)(110)(010)(010)(010) 
\\
t = 2: & (000)(001)(100)(000)(000)(000) & & (011)(011)(101)(002)(002)(002) 
\\
t = 3: & (000)(000)(110)(000)(000)(000) & & (011)(011)(011)(021)(021)(021) 
\\
t = 4: & (000)(000)(011)(000)(000)(000) & & (011)(011)(011)(110)(010)(010) 
\\
t = 5: & (000)(000)(001)(100)(000)(000) & & (011)(011)(011)(101)(002)(002) 
\\
t = 6: & (000)(000)(000)(110)(000)(000) & & (011)(011)(011)(011)(021)(021)
\end{align*}
\end{example}

\subsection{Phase diagram for $\Phi(3,1)$}
\label{subsec:phase}

We close this section with the phase diagram of solitons and pulsars
in the case of $\Phi(3,1)$.
We numerically observe that when an initial state includes 
only one non-vacuum state $(x,y) \in (\Z_{\geq 0})^2$,
it is either a positive soliton, a relaxation soliton, or a pulser as follows:
\begin{enumerate}
\item $(x,1)$ with $x \geq 1$: a positive soliton,

\item $(0,1)$ or $(1,0)$: a pulsar,

\item the other $(x,y)$: a relaxation soliton which reduces to a positive 
soliton of the minimal form $(x+y-1,1)$. 
\end{enumerate}

If we consider the dynamical system on $\Q$, the situation
is more complicated since there are solitons and pulsars 
whose minimal lengths are more than one.
Nevertheless, we numerically find an interesting structure as shown 
in Figure \ref{fig:phase-diagram}.

\begin{figure}[H]
\unitlength=1mm
\begin{picture}(80,60)(0,0)

\put(0,5){\vector(1,0){55}}
\put(5,0){\vector(0,1){55}}

\put(57,3){\scriptsize $x$}
\put(3,57){\scriptsize $y$}

\multiput(20,5)(15,0){3}{\line(0,-1){1.5}}
\multiput(5,20)(0,15){3}{\line(-1,0){1.5}}

\multiput(10,10)(5,5){3}{\circle*{1}}
\multiput(25,20)(5,0){6}{\circle*{1}}

\multiput(5,20)(5,0){10}{\circle{2}}
\multiput(5,15)(5,0){4}{\circle{2}}
\multiput(5,10)(5,0){4}{\circle{2}}
\multiput(10,5)(5,0){3}{\circle{2}}

\multiput(25,15)(5,0){6}{\circle*{1}}
\multiput(30,10)(5,0){5}{\circle*{1}}
\multiput(35,5)(5,0){4}{\circle*{1}}

\multiput(15,25)(5,0){8}{\circle*{1}}
\multiput(10,30)(5,0){9}{\circle*{1}}
\multiput(5,35)(5,0){10}{\circle*{1}}
\multiput(5,40)(5,0){10}{\circle*{1}}
\multiput(5,45)(5,0){10}{\circle*{1}}
\multiput(5,50)(5,0){10}{\circle*{1}}

\multiput(5,20)(5,0){3}{\circle{1}}
\multiput(5,15)(5,0){4}{\circle{1}}
\multiput(5,10)(5,0){5}{\circle{1}}
\multiput(10,5)(5,0){5}{\circle{1}}

\multiput(5,25)(5,0){2}{\circle{1}}
\multiput(5,30)(5,0){2}{\circle{1}}

\put(2,1){\scriptsize $0$}
\put(19,0){\scriptsize $1$}
\put(34,0){\scriptsize $2$}
\put(49,0){\scriptsize $3$}
\put(1,19){\scriptsize $1$}
\put(1,34){\scriptsize $2$}
\put(1,49){\scriptsize $3$}


\put(65,45){\circle*{1}} \put(65,45){\circle{2} \scriptsize{: a soliton}}
\put(65,39){\circle{1}} \put(65,39){\circle{2} \scriptsize{: a pulsar}}
\put(65,33){\circle*{1} \scriptsize{~: a relaxation soliton}}
\put(65,27){\circle{1} \scriptsize{~: a relaxation pulsar}}

\end{picture}
\caption{Phase diagram for $\Phi(3,1)$ on $\Z/3$.}
\label{fig:phase-diagram}
\end{figure}
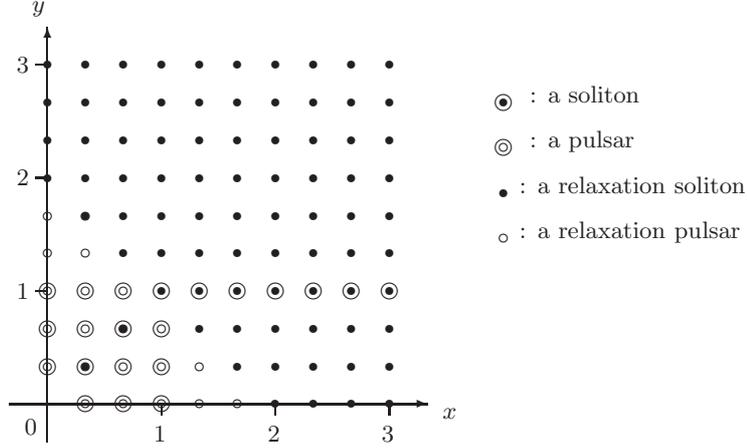

In the phase diagram, the relaxation pulsar is defined to be
a finite sequence of non-vacuum states 
which reduces to a pulsar after a few time steps.
The points $(x,y)$ corresponding to solitons are on a line 
$x=y$ when $x \leq 1$ and on a line $y=1$ when $x \geq 1$.  
We observe that some relaxation solitons $(x,y)$ are reduced to
solitons $(x+y-1,1)$, and the others are reduced to solitons 
out of the diagram, whose minimal lengths are more than one.
This is the same for relaxation pulsars.
See Example \ref{ex:rel-sol} for a relaxation soliton given by
$(x,y)=(\frac{7}{3}, \frac{2}{3})$,
where one sees that the length of the resulting soliton is always two.

\begin{example}\label{ex:rel-sol}
A relaxation soliton of $\Phi(3,1)$ on $\Z/3$:
\begin{align*}
& (\bY_i^t)_i & & (\bZ_i^t)_i
\\[1mm]
t = 0: & (00)(\dds \ddt)(00)(00)(00)(00)(00)(00) & & (10)(10)(0 \ddt)(\ddt 0)(\ddt 0)(\ddt 0)(\ddt 0)(\ddt 0) 
\\
t = 1: & (00)(\ddf 1)(\ddt 0)(00)(00)(00)(00)(00) & & (10)(10)(01)(10)(10)(10)(10)(10) 
\\
t = 2: & (00)(\ddo 1)(\ddfi 0)(00)(00)(00)(00)(00) & & (10)(10)(\ddt 1)(10)(10)(10)(10)(10) 
\\
t = 3: & (00)(0 \ddo)(2 \ddt)(00)(00)(00)(00)(00) & & (10)(10)(1 \ddo)(\ddo \ddt)(10)(10)(10)(10) 
\\
t = 4: & (00)(00)(\ddf 1)(\ddt 0)(00)(00)(00)(00) & & (10)(10)(10)(01)(10)(10)(10)(10)
\\
t = 5: & (00)(00)(\ddo 1)(\ddfi 0)(00)(00)(00)(00) & & (10)(10)(10)(\ddt 1)(10)(10)(10)(10)
\end{align*}
\end{example}

\appendix

\section{Tropical semifield}

\subsection{Tropical limit}
\label{sec:app1}

To a substruction-free rational map, we associate a 
piecewise-linear map via a limiting procedure called {\it tropicalization}. 

The algebra $(\R \cup \{ \infty \}, \oplus, \odot)$ is called the {\it min-plus algebra} (or the {\it tropical semifield}), 
where an addition $\oplus$ and a multiplication $\odot$ are defined by
$$
a \oplus b := \min[a, b], 
\qquad 
a \odot b := a + b.
$$
Note that $\infty$ corresponds to zero in the algebra:
we have $\infty \oplus a = a$ and 
$\infty \odot a = \infty$ for any $a \in \R$.  Moreover
we have an inverse of $\odot$, $a \odot (-a) = 0$, 
but not an inverse of $\oplus$.
In the following we also write $\min$ and $+$ for  
$\oplus$ and $\odot$.

The substruction-free algebra
$(\R_{>0},+,\times)$ is formally linked to the min-plus algebra in the following way.
We define a map $\mathrm{Log}_\ve : \R_{>0} \to \R$ with an infinitesimal 
parameter $\ve > 0$ by
\begin{align}
  \label{i:loge-map}
  \mathrm{Log}_\ve : a \mapsto - \ve \log a.
\end{align}
For $a > 0$, define $A \in \R$ by $a = e^{-\frac{A}{\ve}}$.
Then we have $\mathrm{Log}_\ve (a) = A$. 
Moreover, for $a, b > 0$ define $A,B \in \R$ by $a = e^{-\frac{A}{\ve}}$ and 
$b = e^{-\frac{B}{\ve}}$. 
Then we have  
\begin{align*}
&\mathrm{Log}_\ve (a + b) = 
  -\ve \log (e^{-\frac{A}{\ve}} + e^{-\frac{B}{\ve}})
\stackrel{\ve \to 0}{\mapsto} \min(A,B),
\\ 
&\mathrm{Log}_\ve (a \times b) = A + B.
\end{align*}
In summary, {\it tropicalization} is a procedure which reduce 
the algebra $(\R_{>0},+,\times)$ to the min-plus algebra
by the procedure $\lim_{\ve \to 0} \mathrm{Log}_\ve$ 
with the scale transformation as $a = e^{-\frac{A}{\ve}}$.

Via tropicalization, 
substruction-free rational maps on $\R_{>0}$ formally reduce to 
piecewise-linear maps on $\R$.
We may be able to restrict the resulted piecewise-linear map on $\R$ to that on $\Z$,
which is sometimes called the {\it ultradiscretization} of the original
rational map.

\subsection{Valuation field}
\label{sec:app2}

Let $K = \C\{\{t\}\} := \cup_{n \geq 1}\C((t^{1/n}))$ be the field of 
Puiseux series over $\C$. The field $K$ is an algebraically closed field with non-trivial valuation, where the valuation map $\val: K \to \R \cup \{\infty\}$ is given by
$$
\val: b_1 t^{a_1} + b_2 t^{a_2} + \cdots \mapsto a_1 \quad 
\text{ if $b_1 \neq 0$, $a_1 < a_2 < \cdots \in \Z/n$ for some $n \geq 1$.}
$$

We recall the axioms for the valuation map on $K$: 
\begin{itemize}
\item[(i)] 
$\val(a) = \infty$ iff $a=0$,

\item[(ii)]
$\val(a b) = \val(a) + \val(b)$ for any $a,b \in K$,

\item[(iii)]
$\val(a+b) \geq \min[ \val(a), \val(b)]$ for any $a,b \in K$.  
\end{itemize}
As for the last axiom, we have an important lemma:
\begin{lem}\label{lem:non-arch-eq}
For $a,b \in K$, if $\val(a) \neq \val(b)$, then an equality holds in
the above (iii), i.e.  $\val(a+b) = \min[ \val(a), \val(b)]$.
\end{lem}
For the proof, see \cite[Lemma 2.1.1]{MacSturm-book} for example.

The tropicalization can be regarded as the following composition map:
$$
\begin{matrix}
\R_{>0} & \longrightarrow & K & \stackrel{\val}{\longrightarrow} 
& \R \cup \{\infty \}
\\
a = e^{-\frac{A}{\ve}} & \mapsto & t^A & \mapsto & A
\end{matrix}.
$$
We write $\trop(a) = A$ for the image $A$ of $a$ under this map.

\end{document}